\documentclass[letterpaper, 10 pt, journal, twoside]{ieeetran}
\usepackage{times}
\usepackage{arydshln}
\usepackage{tikz}
\usetikzlibrary{calc}
\usepackage[numbers,sort,compress]{natbib}
\usepackage{multicol}
\usepackage{svg}
\usepackage[bookmarks=true]{hyperref}
 
\usepackage{comment}
\usepackage{amsthm}
\usepackage{mathtools}
\usepackage{amsmath}
\usepackage{amsfonts}
\usepackage{soul} 
\usepackage{amsmath,amsfonts,amssymb}
\usepackage[mathscr]{euscript}
\usepackage{hyperref}
\usepackage{graphicx}
\usepackage[dvipsnames]{xcolor} 

\usepackage{algorithm}
\usepackage[noend]{algpseudocode}
\algrenewcommand{\algorithmicrequire}{\textbf{Input:}}
\algrenewcommand{\algorithmicensure}{\textbf{Output:}}

\hypersetup{
    colorlinks=true,
    linkcolor=black,
    citecolor=blue,
    urlcolor=black
}

\newtheorem{corollary}{Corollary}[section]
\newtheorem{theorem}{Theorem}[section]
\newtheorem{lemma}{Lemma}[section]
\theoremstyle{definition}
\newtheorem{assumption}{Assumption}[section]
\newtheorem{problem}{Problem}[section]

\newtheorem{definition}{Definition}[section]
\newcounter{subproblem}[problem] 

\usepackage{xcolor}

\renewcommand{\bf}[1]{\textbf{#1}}
\newcommand{\mycomment}[1]{}

\usepackage[capitalise]{cleveref}
 \crefname{theorem}{Theorem}{Theorem}
 \crefname{problem}{Prob.}{Problem}
 \crefname{assumption}{Assumption}{Assumption}
 \crefname{algorithm}{Alg.}{Algs.}
\crefname{definition}{Def.}{Defs.}
\Crefname{definition}{Def.}{Defs.}
\crefname{corollary}{Cor.}{Cors.}
\crefname{theorem}{Thm.}{Thms.}

\begin{document}

\title{PRISM: Efficient and Locally Optimal Probabilistic Planning with Reachability Guarantees}

\author{Alex Rose, Christopher Jewison
 and Jonathan P. How, \IEEEmembership{IEEE Fellow} 
\thanks{This work was supported by the Draper Scholars program. Alex Rose and Jonathan P. How are with Laboratory for Information and Decision Systems, Massachusetts
Institute of Technology, Cambridge MA 02139. Alex Rose is also a Draper Scholar and Christopher Jewison is with Draper, Cambridge MA 02139. Email: $\{\texttt{ameredit}, \texttt{jhow} \}$@mit.edu, cjewison@draper.com.
}}



%

\maketitle

\begin{abstract}
Belief-space planning under motion uncertainty and state and control constraints remains a fundamental challenge, largely due to the difficulty of establishing reachability guarantees in constrained belief spaces. Existing constrained belief-space planners rely on sampling to construct multi-query belief roadmaps and explicitly find feasible trajectories between sampled nodes to establish reachability. These methods often struggle to cover the belief space or use robust control techniques that improve coverage at the cost of indirect, high-cost trajectories; they also lack finite-time or finite-memory completeness guarantees. We propose PRISM, a multi-query motion planning algorithm for belief spaces with state and control constraints that targets both high coverage and low cost. We present a new result on controllability of the state covariance under constraints, which is used by PRISM to decompose belief-space planning into deterministic mean planning and covariance shrinking. PRISM further includes an online local optimization method that reduces the cost of feasible belief-space trajectories. Under mild assumptions on the start and goal distributions, we prove that PRISM guarantees full coverage (i.e. completeness) despite actuator and obstacle constraints. In challenging simulated scenarios, PRISM achieves substantially higher roadmap coverage than state-of-the-art belief-space planning methods while producing trajectories with lower mean cost and cost variance. For example, PRISM achieves 100\% coverage in easy and medium-difficulty scenarios, and, in the hardest scenario, which violates PRISM's coverage assumptions, it still achieves 97--100\% coverage, while all other methods achieve less than 45\%.
\end{abstract}

\IEEEpeerreviewmaketitle
\section{Introduction}\label{sec: introduction}
Many real-world robotic applications require navigating through environments with obstacles under motion and observation uncertainty. Developing planning algorithms that can accomplish this task efficiently remains challenging because reachability in belief space is difficult to determine. For example, steering a robot’s belief, represented by its mean and covariance, between two Gaussian distributions may be impossible, while establishing reachability under state and control constraints generally requires solving at least one semidefinite program \cite{okamoto2018optimal, rapakoulias2023discrete}.

Recent works have addressed the challenge of reachability by designing sampling-based planning methods in belief space, and by explicitly establishing feasibility for each node-to-node maneuver. FIRM \cite{agha2014firm} was the first method to break the ``curse of history'' by explicitly establishing reachability between nodes, but only sampled stationary nodes. CS-BRM \cite{zheng2024cs} built on FIRM by establishing reachability between non-stationary nodes. However, belief space is high-dimensional, so methods like FIRM and CS-BRM that jointly sample the mean and covariance of each node may require impractically many nodes to achieve good coverage of the belief space. Further, FIRM and CS-BRM do not explicitly enforce chance constraint satisfaction for each edge, and increase the cost of edges that violate chance constraints instead. Therefore, in the presence of obstacles or actuator constraints, trajectories produced by FIRM and CS-BRM may not be \textit{sound} (i.e. may violate constraints).

\begin{figure}[tbp]
\includegraphics[width=\columnwidth]{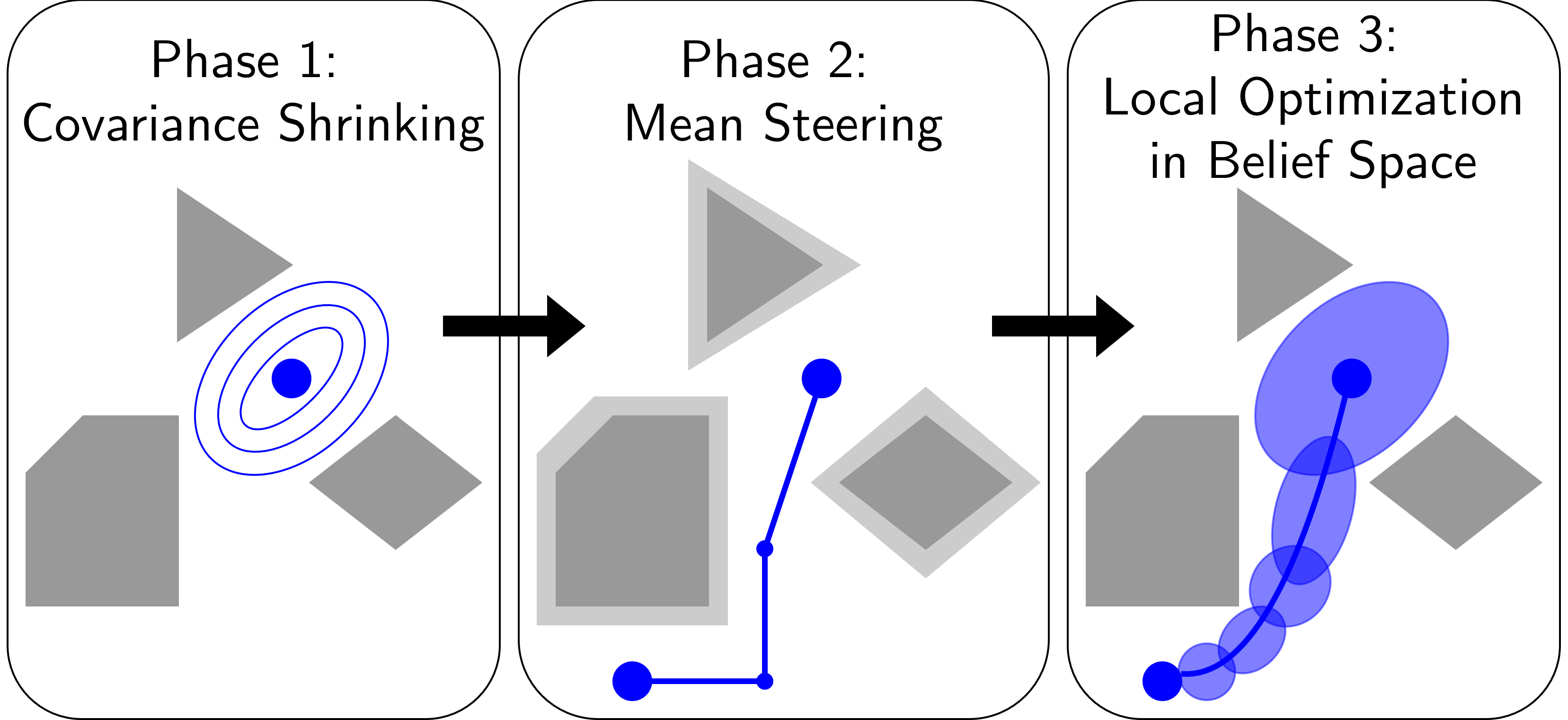}
\vspace{-0.4cm}
\caption{PRISM leverages our new result on reachability of the state covariance (see \cref{sec: shrinkability}) to separate the belief space planning problem into a covariance shrinking phase and a mean steering phase with inflated obstacles and tightened constraints, then jointly locally optimizes the mean and covariance trajectories in belief space.}
\label{fig: prism_overview}
\vspace{-0.5cm}
\end{figure}

More recently, several algorithms have been developed to build \textit{maximal coverage} roadmaps in belief space; that is, belief roadmaps that cover as much of the belief space as possible while maintaining node sparsity, in order to maximize the size of the set of initial distributions that can be steered to a goal. MAX-COVAR \cite{aggarwal2024sdp} introduced an edge controller that maximizes coverage when nodes are represented by Gaussian distributions, and new node representations that further improve coverage were introduced with MAX-ELLIPSOID \cite{aggarwal2025tac} and MAX-COV-BALL \cite{rose2025efficient}. These algorithms also guarantee chance constraint satisfaction with each edge, and therefore are guaranteed to be sound. However, because these algorithms rely on robust controllers to maximize roadmap coverage, they often produce sub-optimal paths with high control and time cost. Further, although these methods avoid sampling the covariance of each node, they still sample node means or centroids and do not guarantee completeness in finite time.

Set-based planning methods like planning through Graphs of Convex Sets (GCS) have demonstrated several advantages over sampling-based planning for deterministic path planning in the presence of obstacles \cite{marcucci2023motion, marcucci2024shortest}. Importantly, if the state space can be exactly decomposed into a union of convex sets, as in \cite{geyer2008optimal, shaikh2025exact}, planning with GCS is guaranteed to be complete \cite{marcucci2023motion}. This is a meaningful advance over sampling-based motion planning algorithms like RRT$^\star$, FMT$^\star$ and PRM \cite{karaman2011sampling, kavraki1996probabilistic, janson2015fast} that have favorable asymptotic properties but cannot guarantee optimality or completeness for a fixed map size or fixed amount of computation time. Unfortunately, although the belief space planning problem can be formulated as a shortest path problem in a graph of convex sets, finding feasible belief space trajectories inside each convex set still requires establishing reachability in belief space, which remains challenging.

Constraint tightening is a common approach used in robust \cite{richards2006robust, kohler2018novel, zanon2020safe} and stochastic \cite{lorenzen2016constraint, hewing2020recursively} model-predictive control  to guarantee recursive feasibility. Current methods for chance-constrained covariance steering \cite{okamoto2018optimal, rapakoulias2023discrete} in belief space jointly optimize the mean and covariance trajectories in belief space, and convert probabilistic chance constraints to deterministic constraints by tightening constraints based on the covariance at each time step. This coupled approach avoids conservatism when tightening constraints, but requires solving an optimization problem to explicitly establish reachability and does not allow reachability of the state mean and covariance to be analyzed separately. 

This paper addresses the problem of multi-query planning in belief space with obstacles with \textit{soundness} and \textit{completeness} guaranteed in \textit{finite time}. We introduce new sufficient conditions for the reachability of the state covariance in belief space in the presence of obstacle and actuator constraints. We leverage these conditions to develop PRISM (Precomputed Regions for Inflation-based Stochastic Motion planning), an algorithm for planning from \textit{any} start distribution to \textit{any} goal distribution in belief space while avoiding obstacles and satisfying actuator constraints. PRISM separates the global planning problem in belief space into a deterministic mean planning problem and a covariance shrinking problem, and then uses online local optimization to simultaneously optimize the resulting feasible mean and covariance trajectories (see \cref{fig: prism_overview}). PRISM always produces \textit{sound} trajectories, and guarantees \textit{completeness} when the initial and goal means are stationary, the initial covariance is upper-bounded, and the goal covariance is lower-bounded (see \cref{thm: completeness_without_obstacles} and \cref{thm: completeness_with_obstacles} for details). Our contributions are:    
\begin{enumerate}
\item{New sufficient conditions for the reachability of the state covariance under obstacle and actuator constraints.}
\item{A new algorithm for multi-query planning in belief space, PRISM, that leverages our new conditions on reachability in belief space to improve coverage. Compared to previous methods, PRISM guarantees high coverage without sampling, avoiding the inefficiency inherent to sampling in high-dimensional belief space.}
\item An online method for local optimization of feasible plans in belief space, included as a submodule of PRISM, that improves performance by more than $2.5\times$ compared to state-of-the-art methods \cite{zheng2024cs, aggarwal2025tac, karaman2011sampling, rose2025efficient}.
\item{Proofs that PRISM is sound, and that it is complete under mild assumptions on the start and goal distributions (see \cref{thm: completeness_without_obstacles,thm: completeness_with_obstacles} for details on completeness).}
\item{Simulation experiments that demonstrate that PRISM outperforms state-of-the-art belief space planning methods \cite{zheng2024cs, aggarwal2025tac, karaman2011sampling, rose2025efficient} in coverage and cost in challenging scenarios in cluttered and confined environments.}
\end{enumerate}

\section{Problem Statement}\label{sec:problem_statement}
\subsection{Notation and Mathematical Preliminaries}
We describe positive semidefinite and positive definite matrices in $\mathbb{R}^{n\times n}$ as $\mathbb{S}^n_+$ and $\mathbb{S}^n_{++}$, respectively. We denote the nonnegative reals, positive reals, and positive integers as $\mathbb{R}_{\geq 0}$, $\mathbb{R}_{>0}$, and $\mathbb{Z}_{>0}$, respectively. We use $\oplus$ to denote the Minkowski sum. We use the notation that $\mathbb{B}_n(c, r) := \{\mathbf{x} \in \mathbb{R}^n: ||\mathbf{x}-c||_2 \leq r \}$ represents an $n$-ball with respect to the Euclidean norm. We normalize all half-space constraints of the form $a^T\mathbf{x} \leq b$ with $a, \mathbf{x} \in \mathbb{R}^n$, $b \in \mathbb{R}$ such that $||a||_2 = 1$. We describe the interior of a polytope $\mathcal{S}$ as $\text{Int}(\mathcal{S})$, such that $\mathcal{S} := \{\mathbf{x} \in \mathbb{R}^n: \bigcap_{i} {a_i}^T\mathbf{x} \leq b_i \}$ and $\text{Int}(\mathcal{S}) := \{\mathbf{x} \in \mathbb{R}^n: \bigcap_{i} {a_i}^T\mathbf{x} < b_i \}$ and describe the interior of an $n$-ball as $\text{Int}(\mathbb{B}_n(c, r)) := \{\mathbf{x} \in \mathbb{R}^n: ||\mathbf{x}-c||_2 < r  \}$. We denote $n$-dimensional multivariate Gaussian distributions by $\mathcal{N}(\mu, \Sigma)$, where $\mu \in \mathbb{R}^n$ is the distribution mean and $\Sigma \in \mathbb{S}^n_+$ is the covariance. We describe planning over $n$-dimensional multivariate Gaussian distributions in $\mathbb{R}^n \times \mathbb{S}^n_+$ as planning in \textit{belief space}. Also, to ease notation, we define $[N] := \{1,\ldots,N\}$.
\subsection{Problem}
Consider a discrete-time linear time-invariant system with dynamics
\begin{equation}\label{eq: dynamics}
\mathbf{x}_{k+1} = A\mathbf{x}_k + B\mathbf{u}_k + D\mathbf{w}_k
\end{equation}
where $\mathcal{X} \subseteq \mathbb{R}^n$ and $\mathcal{U} \subseteq \mathbb{R}^m$ are the safe state and control spaces and $\mathbf{w}_k \sim \mathcal{N}(0, I_n)$ with $\{\mathbf{w}_k\}$ i.i.d. and $A$ non-singular. We assume the state and control constraints are probabilistic, with $\mathbb{P}(\mathbf{x}_k \in \mathcal{X}) \geq 1 - \epsilon_1$ and $\mathbb{P}(\mathbf{u}_k \in \mathcal{U}) \geq 1 - \epsilon_2$ for all $k$.

We define a finite-horizon steering problem where an initial Gaussian distribution $\mathcal{I} := \mathcal{N}(\mu_\mathcal{I}, \Sigma_\mathcal{I})$ is steered to a goal Gaussian distribution $\mathcal{G} := \mathcal{N}(\mu_\mathcal{G}, \Sigma_\mathcal{G})$ in Problem \ref{prob: time_optimal_covariance_steering}, subject to the dynamics in \eqref{eq: dynamics}, minimizing quadratic control cost augmented with a cost on the maneuver time $N$. We assume a constraint on the maximum maneuver time $N \leq N_{\max}$. For fixed $N$, Problem \ref{prob: time_optimal_covariance_steering} reduces to the optimization in \cite[Section II]{rapakoulias2023discrete}. The right panel in \cref{fig: prism_overview} visualizes this belief space steering problem with obstacle constraints.

\begin{problem}\label{prob: time_optimal_covariance_steering}
\textit{Steer from initial distribution $\mathcal{I} := \mathcal{N}(\mu_\mathcal{I}, \Sigma_\mathcal{I})$ to reach a goal distribution $\mathcal{G} := \mathcal{N}(\mu_\mathcal{G}, \Sigma_\mathcal{G})$, minimizing time-augmented control cost, subject to state and control chance constraints.}
\begin{equation}
\min_{\mathbf{u}_k, N} J = \mathbb{E}\left[\sum_{k=1}^{N-1} (\mathbf{u}_k^T R_k \mathbf{u}_k + c_T)\right]
\end{equation}
such that for all $k \in [N]$:
\begin{align}
&\mathbf{x}_{k+1} = A\mathbf{x}_k + B\mathbf{u}_k + D\mathbf{w}_k, \notag \\
&\mathbf{x}_1 \sim \mathcal{N}(\mu_\mathcal{I}, \Sigma_\mathcal{I}), \quad \mathbf{x}_N \sim \mathcal{N}(\mu_N, \Sigma_N), \quad N \leq N_{\max}, \notag \\
&\mathbb{P}(\mathbf{x}_k \in \mathcal{X}) \geq 1 - \epsilon_1 \ \forall k, \quad \mathbb{P}(\mathbf{u}_k \in \mathcal{U}) \geq 1 - \epsilon_2 \ \forall k.
\end{align}
\end{problem}
Under state, control, and maximum time constraints, Problem \ref{prob: time_optimal_covariance_steering} may not have a feasible solution; however, it may still be possible to steer from $\mathcal{I}$ to $\mathcal{G}$ by composing multiple maneuvers. In this paper, we address the following problem:

\begin{problem}\label{prob: find_all_paths}
\textit{Find paths that minimize time-augmented control cost from all initial Gaussian distributions to all final Gaussian distributions, subject to state and control constraints.}
\end{problem}

We solve this problem by constructing a graph-based roadmap in belief space that is guaranteed to be complete for a large subset of the belief space. Our roadmap construction algorithm, PRISM, leverages our new theorem on the reachability of the system covariance under obstacle and actuator constraints to decompose the belief space planning problem into a covariance shrinking problem and a mean steering problem in deterministic space. PRISM builds a graph-based roadmap in deterministic space with tightened constraints and then lifts it into belief space while maintaining feasibility guarantees. Upon extracting a feasible path from the graph-based roadmap, PRISM applies a novel local optimization method to refine the path to local optimality in belief space.

The rest of this paper is organized as follows: in \cref{sec: preliminaries} we present background information on covariance steering in belief space. \cref{sec: shrinkability} presents new results on the reachability of the state covariance in belief space under actuator and obstacle constraints, and in \cref{sec: approach} we use these results to construct a novel global planning algorithm in belief space with finite-time coverage guarantees. \cref{subsec: local_refinement} develops a new local optimization algorithm in belief space and  \cref{sec: theoretical_results} provides results on soundness and completeness for the global and local planning algorithms. \cref{sec: experiments} presents the results of our simulation experiments. Proofs for all lemmas, theorems, and corollaries stated in the main text of the paper are presented in Appendix \ref{sec:appendix}.
\section{Covariance Steering Preliminaries}
\label{sec: preliminaries}
In the absence of chance constraints and for a fixed time horizon $N$, the optimal control for \cref{prob: time_optimal_covariance_steering} is given by a linear state feedback law of the form \cite{rapakoulias2023discrete, liu2024optimal}:
\begin{equation}\label{eq: feedback_control_law}
\mathbf{u}_k = K_k(\mathbf{x}_k - \mu_k) + \mathbf{v}_k,
\end{equation}
where $K_k \in \mathbb{R}^{m \times n}$ is a feedback gain matrix and $\mathbf{v}_k \in \mathbb{R}^m$ is a feedforward control term. Accordingly, the dynamics and cost in \cref{prob: time_optimal_covariance_steering} can be reformulated as
\begin{align}
\mu_{k+1} &= A\mu_k + B\mathbf{v}_k,\label{eq: mean_dynamics}\\
\Sigma_{k+1} &= (A + BK_k)\Sigma_k(A + BK_k)^T + DD^T, \label{eq: covariance_dynamics}\\
J &= \sum_{k=0}^{N-1} (\text{tr}(RK_k\Sigma_kK_k^T) + \mathbf{v}_k^T R_k \mathbf{v}_k + c_T), \label{eq: cost_function}
\end{align}
and the initial and final distribution constraints become $\mu_1 = \mu_\mathcal{I}$, $\mu_N = \mu_\mathcal{G}$, $\Sigma_1 = \Sigma_\mathcal{I}$, and $\Sigma_N = \Sigma_\mathcal{G}$. We relax the final covariance constraint to $\Sigma_N \preceq \Sigma_\mathcal{G}$, adding flexibility without increasing uncertainty at the goal distribution. 

Following \cite{rapakoulias2023discrete, aggarwal2024sdp}, the safe control set $\mathcal{U}$ is a closed polytope in $\mathbb{R}^m$, where $m$ is the control dimension, with $\mathcal{U} = \{\mathbf{u} \in \mathbb{R}^m : \bigcap_{i=1}^{M_u} {a_i^u}^T\mathbf{u} \leq b_i^u\}$. Similarly, $\mathcal{X} = \{\mathbf{x} \in \mathbb{R}^n : \bigcap_{i=1}^{M_x}{a_i^x}^T\mathbf{x} \leq b_i^x\}$. We assume that $\mathcal{X}$ can be separated into a position constraint polytope and a polytope of constraints on other parts of the state. Specifically, we assume that position space has dimension $p$ with $\mathbb{R}^p \subseteq \mathbb{R}^n$, that there exists a transformation matrix $P \in \mathbb{R}^{p\times n}$ such that for any state $\mathbf{x} \in \mathbb{R}^n$, its position is given by $\mathbf{p}  = P\mathbf{x}$, and that there exists a transformation matrix $S \in \mathbb{R}^{(n-p)\times n}$ such that for any state $\mathbf{x} \in \mathbb{R}^n$, its non-position state components are given by $\mathbf{s} = S\mathbf{x}$. We assume that for any stationary point $\mathbf{x} \in \mathbb{R}^n$ satisfying $A\mathbf{x} = \mathbf{x}$, $S\mathbf{x} = 0$. We assume that $\mathcal{X} = \mathcal{X}_p \bigcap \mathcal{X}_s$, with $\mathcal{X}_p = \{\mathbf{x} \in \mathbb{R}^n : \bigcap_{i=1}^{M_{x_p}}{a_i^{x_p}}^TP\mathbf{x} \leq b_i^{x_p} \}$ and $\mathcal{X}_s = \{\mathbf{x} \in \mathbb{R}^n : \bigcap_{i=1}^{M_{x_s}}{a_i^{x_s}}^TS\mathbf{x} \leq b_i^{x_s} \}$, with $M_x = M_{x_s} + M_{x_p}$.

Polytopic obstacles also exist in position space $\mathbb{R}^p \subseteq \mathbb{R}^n$. We lift the obstacle set to $\mathbb{R}^n$ by defining it as $\mathcal{P} := \bigcup_{j=1}^{M_p} \mathcal{P}_j$, with $\mathcal{P}_j := \{\mathbf{x} \in \mathbb{R}^n: \bigcap_{i=1}^{M_{p_j}}{a_{i}^{p_j}}^TP\mathbf{x} \leq b_i^{p_j} \}$, with $a_i^{p_j} \in \mathbb{R}^p$ for all $i$. Then, the safe state set is $\mathcal{X}_s \bigcap (\mathcal{X}_p \setminus \mathcal{P})$. We restrict the position domain of \cref{prob: time_optimal_covariance_steering} to a sequence of safe polytopes $\mathcal{S}_1, \ldots \mathcal{S}_N$, with $\mathcal{S}_k \subseteq \mathcal{X}_p \setminus \mathcal{P}$ for all $k$. We lift $\mathcal{S}_k$ to $\mathbb{R}^n$ as $\mathcal{S}_k := \{\mathbf{x} \in \mathbb{R}^n: \bigcap_{i=1}^{M_{s_k}}{a_i^{s_k}}^TP\mathbf{x} \leq b_i^{s_k} \}$. Then, as in \cite{okamoto2018optimal, rapakoulias2023discrete}, the state chance constraint $\mathbb{P}(\mathbf{x}_k \in \mathcal{S}_k \bigcap\mathcal{X}_s) \geq 1-\epsilon_1$ can be written as
\begin{align}\label{eq: pos_chance_constraint}
\Phi^{-1}(1-\epsilon_x)\sqrt{{a_\ell^{s_k}}^TP\Sigma_kP^Ta_\ell^{s_k}} +{a_\ell^{s_k}}^TP\mu_k \leq b_\ell^{s_k}, 
\end{align}
for all $k \in [N]$, $\ell \in [M_{s_k}]$, with $\Phi^{-1}(\cdot)$ as the inverse cumulative distribution function of the standard normal distribution, with the non-position state chance constraints as
\begin{align}\label{eq: nonpos_state_chance_constraint}
\Phi^{-1}(1-\epsilon_x)\sqrt{{a_i^{x_s}}^TS\Sigma_kS^Ta_i^{x_s}} +{a_i^{x_s}}^TS\mu_k \leq b_i^{x_s},
\end{align}
for $k \in [N]$, $i \in [M_{x_s}]$. Similarly, the control chance constraints state that for $k \in [N-1]$, $j \in [M_u]$,
\begin{equation}\label{eq: control_chance_constraint}
\Phi^{-1}(1-\epsilon_u)\sqrt{{a_j^u}^TK_k\Sigma_kK_k^Ta_j^u} +{a_j^u}^T\mathbf{v}_k \leq b_j^u. 
\end{equation}
Then, we can define soundness of a covariance steering trajectory based on satisfaction of the reformulated constraints.
\begin{definition}[{Soundness}]\label{def: soundness}
Given an initial Gaussian distribution $\mathcal{I} = \mathcal{N}(\mu_\mathcal{I}, \Sigma_\mathcal{I})$ and a goal Gaussian distribution $\mathcal{G} = \mathcal{N}(\mu_\mathcal{G}, \Sigma_\mathcal{G})$, a covariance steering trajectory $(\mathbf{\mu}_{1:N}, \mathbf{v}_{1:N-1}, \Sigma_{1:N}, K_{1:N-1})$ is sound if \cref{eq: mean_dynamics} is satisfied for $k\in[N-1]$, \cref{eq: covariance_dynamics} is satisfied for $k \in[N-1]$, \cref{eq: pos_chance_constraint} is satisfied for $k \in [N]$, $\ell \in [M_{s_k}]$, \cref{eq: nonpos_state_chance_constraint} is satisfied for $k \in [N]$, $i \in [M_{x_s}]$, \cref{eq: control_chance_constraint} is satisfied for $k \in [N-1]$, $j \in [M_u]$, and the boundary constraints $\mu_1 = \mu_\mathcal{I}$, $\mu_N = \mu_\mathcal{G}$, $\Sigma_1 \succeq \Sigma_\mathcal{I}$, and $\Sigma_N \preceq \Sigma_\mathcal{G}$ are satisfied.
\end{definition}
\cref{eq: covariance_dynamics}, \cref{eq: cost_function}, and \cref{eq: control_chance_constraint} are bilinear in $K_k$ and $\Sigma_k$. Following \cite{rapakoulias2023discrete, liu2024optimal}, we (losslessly) substitute $U_k = K_k\Sigma_k$, add a variable $Y_k = U_k\Sigma_k^{-1}U_k^T$, reformulate \cref{eq: covariance_dynamics}, \cref{eq: cost_function}, \cref{eq: control_chance_constraint}, add a new constraint $\forall k$, then rewrite \cref{prob: time_optimal_covariance_steering}:
\begin{align}
&\Sigma_{k+1}\!\!=\!\! A\Sigma_kA^T \!+\! BU_kA^T \!+\! AU_k^TB^T \!+\! BY_kB^T \!+\! DD^T,\label{eq: convex_covariance_dynamics}\\
&J = \sum_{k=0}^{N-1} (\text{tr}(RY_k) + \mathbf{v}_k^T R \mathbf{v}_k + c_T),\label{eq: convex_cost}\\
&\Phi^{-1}(1-\epsilon_u)\sqrt{{a_j^u}^TY_ka_j^u} +{a_j^u}^T\mathbf{v}_k \leq b_j^u,\label{eq: Yk_control_constraint} \\
&\begin{bmatrix}
\Sigma_k & U_k^T \\
U_k & Y_k
\end{bmatrix} \succeq 0. \label{eq: schur_complement_constraint}
\end{align}
\begin{problem}\label{prob: reformulated}
\textit{Steer from initial distribution $\mathcal{I}$ to goal distribution $\mathcal{G}$, minimizing time-augmented control cost, subject to state and control chance constraints, with the position domain restricted to a sequence of safe polytopes $\mathcal{S}_1, \ldots \mathcal{S}_N$.}


\begin{equation}
\min_{N, \mu_k, \mathbf{v}_k, \Sigma_k U_k, Y_k} (\ref{eq: convex_cost}) \notag
\end{equation}
such that for all $i \in [M_{x_s}]$, $j \in [M_u]$, $k \in [N]$, $\ell \in [M_{s_k}]$:
\begin{align}
&(\ref{eq: mean_dynamics}), (\ref{eq: convex_covariance_dynamics}), (\ref{eq: schur_complement_constraint}), (\ref{eq: Yk_control_constraint}), (\ref{eq: nonpos_state_chance_constraint}), (\ref{eq: pos_chance_constraint}), \notag \\
&\mu_1=\mu_\mathcal{I},\ \Sigma_1=\Sigma_\mathcal{I},\ \mu_N=\mu_\mathcal{G},\ \Sigma_N \preceq \Sigma_\mathcal{G},\ N \leq N_{\max}. \notag
\end{align}
\end{problem}
 \cref{eq: pos_chance_constraint}, \cref{eq: control_chance_constraint}, and \cref{eq: nonpos_state_chance_constraint} require taking a square root of a matrix decision variable and are non-convex. Following \cite{rapakoulias2023discrete}, we overestimate the square root using a tangent line. For any symmetric positive semidefinite matrix $M_k$, we define
\begin{equation}
\mathbb{M}_1(a, M_r, M_k) := \frac{1}{2\sqrt{a^TM_ra}}a^TM_ka + \frac{\sqrt{a^TM_ra}}{2}, \notag
\end{equation}
with $\mathbb{M}_1(a, M_r, M_k) \leq \sqrt{a^T M_k a}$ for any $a \in \mathbb{R}^n$, $M_r, M_k \in \mathbb{S}^n_+$. As such, we can rewrite the chance constraints as
\begin{align}
&\Phi^{-1}(1-\epsilon_x)\mathbb{M}_1(P^Ta_\ell^{s_k}, \Sigma_{r, k}^\ell, \Sigma_k) + {a_\ell^{s_k}}^TP\mu_k \leq b_\ell^{s_k}, \label{eq: convex_pos_chance_constraint} \\
&\Phi^{-1}(1-\epsilon_x)\mathbb{M}_1(S^Ta_i^{x_s}, \Sigma_{r, k}^i, \Sigma_k) + {a_i^{x_s}}^TS\mu_k \leq b_i^{x_s}, \label{eq: convex_nonpos_state_chance_constraint} \\
&\Phi^{-1}(1-\epsilon_u)\mathbb{M}_1(a_j^u, Y_{r, k}^j, Y_k) + {a_j^u}^T\mathbf{v}_k \leq b_j^u. \label{eq: convex_ctrl_chance_constraint}
\end{align}
The convexification of the chance constraints is not lossless, but $\mathbb{M}_1(a, M_r, M_k) = \sqrt{a^T M_k a}$ when $M_r = M_k$. So, for any feasible solution to \cref{prob: reformulated}, it is possible to choose $\Sigma_{r, k}^\ell$, $\Sigma_{r, k}^i$, $Y_{r, k}^j$ such that the constraints given by \cref{eq: convex_pos_chance_constraint} -- \cref{eq: convex_ctrl_chance_constraint} are satisfied. For a fixed final time $N$, we rewrite \cref{prob: reformulated} as a convex semidefinite program:
\begin{problem}\label{prob: steer_edge}
\textit{Steer from $\mathcal{I}$ to $\mathcal{G}$ over $N$ steps, minimizing time-augmented control cost, subject to convex tightened state and control chance constraints, with the position domain restricted to a sequence of safe polytopes $\mathcal{S}_1, \ldots \mathcal{S}_N$.}


\begin{equation}
\min_{\mu_k, \mathbf{v}_k, \Sigma_k, U_k, Y_k} (\ref{eq: convex_cost}) \notag
\end{equation}
such that for all $i \in [M_{x_s}]$, $j \in [M_u]$, $k \in [N]$, $\ell \in [M_{s_k}]$:
\begin{align}
&(\ref{eq: mean_dynamics}), (\ref{eq: convex_covariance_dynamics}), (\ref{eq: schur_complement_constraint}), (\ref{eq: convex_ctrl_chance_constraint}), (\ref{eq: convex_nonpos_state_chance_constraint}), (\ref{eq: convex_pos_chance_constraint}), \notag \\
&\mu_1=\mu_\mathcal{I},\ \Sigma_1=\Sigma_\mathcal{I},\ \mu_N=\mu_\mathcal{G},\ \Sigma_N \preceq \Sigma_\mathcal{G}. \notag
\end{align}
\end{problem}
Even with the convex reformulated chance constraints, covariance steering with free final time is a mixed-integer semidefinite program. It can be solved (inefficiently) by solving \cref{prob: steer_edge} for $N \in [N_{\max}]$, interpolating the sequence of safe polytopes and reference values for linearization if needed, and selecting the trajectory with the lowest cost.


\section{Reachability of the State Covariance}\label{sec: shrinkability}
Consider the discrete-time linear time-invariant dynamics given by \cref{eq: dynamics}. The controllability of the state covariance over a finite horizon in the absence of state and control constraints, but in the presence of additive noise, has previously been studied in \cite{liu2024reachability}, and has been studied in the unconstrained and noise-free case in \cite{abdelgalil2025collective, liu2026reachability}. However, to the authors' knowledge, no previous studies of the controllability of the state covariance in the presence of state or control constraints appear in the literature. Further, even in the absence of state and control constraints, when $D \neq 0$ for all time, the state covariance is not controllable from time $1$ to any time $k>1$ (see Theorem 4 in \cite{liu2024reachability}). Accordingly, we introduce the concept of {shrinkability}, where $\Sigma_1$ is defined to be {shrinkable} over $k$ time steps if and only if there exists $\bar{\Sigma}_k \in \mathbb{S}^n_+$ such that $\bar{\Sigma}_k \prec \Sigma_1$ and such that $\bar{\Sigma}_k$ is reachable from $\Sigma_1$ over $k$ time steps. We use the concept of shrinkability to develop sufficient conditions on the reachability of the state covariance in the presence of actuator constraints (see \cref{lemma: delta_shrinkability}) and in the presence of actuator and obstacle constraints (see \cref{thm: r_c_shrinkability}).

\subsection{Reachability Without Position Constraints}
\label{subsec: shrinkability_without_position_constraints}
Suppose the state $\mu_1$ is stationary, meaning that $A\mu_1 = \mu_1$. Then, under the feedback control law given by \cref{eq: feedback_control_law} and mean dynamics given by \cref{eq: mean_dynamics}, if $\mathbf{v}_k = 0$ for all $k$, then $\mu_k = \mu_1$ for all $k \geq 1$. Further, because $\mu_1$ is stationary, its non-position state component vector $\mathbf{s}_1 = S\mu_1 = 0$, so $\mathbf{s}_k = 0$ $\forall k$ and thus the non-position state constraints at time $k$ depend only on the covariance $\Sigma_k$. Therefore, in the presence of control constraints and non-position state constraints (but in the absence of position constraints), we can separate the mean and covariance steering problems if $\mu_1$ is stationary and $\mu_k = \mu_1$ for all $k$.

Then, a Gaussian distribution $\mathcal{N}(\mu_1, \Sigma_1)$ with stationary $\mu_1$ has {${\delta}$-strictly shrinkable} covariance over $N$ steps if the following optimization problem has a solution, with $\delta \in \mathbb{R}_{\geq 0}$:
\begin{problem}[{Strictly Shrink Covariance by $\delta$}]\label{prob: covariance_shrinking_delta}
\begin{equation}
\min_{\Sigma_k, U_k, Y_k} 0
\end{equation}
\text{such that for $i\in[M_{x_s}]$, $j\in[M_u]$, $k\in[N]$:}
\begin{align}
&\Phi^{-1}(1-\epsilon)^2({a_j^u}^TY_ka_j^u) \leq ({b_j^u}-\delta)^2, \label{eq: control_constraint_fixed_v_delta} \\
&\Phi^{-1}(1-\epsilon)^2({a_i^{x_s}}^TS\Sigma_kS^T{a_i^{x_s}}) \leq ({b_i}^{x_s}-\delta)^2, \label{eq: non_state_constraint_fixed_delta} \\
& (\ref{eq: convex_covariance_dynamics}), \quad (\ref{eq: schur_complement_constraint}), \quad \lambda_{\max}(\Sigma_N) \leq \lambda_{\min}(\Sigma_1)-\delta.
\end{align}
\end{problem}

For any $\delta \in \mathbb{R}_{\geq 0}$, we can define $C_\delta \subset \mathbb{R}$ as $C_\delta := \{c \in \mathbb{R}_{\geq 0}: cI\text{ is $\delta$-strictly shrinkable in }$N$ \text{ steps} \}$, with $c_{1, \delta}^{(N)} = \inf C_\delta$, $c_{2, \delta}^{(N)} = \sup C_\delta$. Then, if $\Sigma_1 \prec c_{1, \delta}^{(N)}I$ it is not $\delta$-strictly shrinkable in $N$ steps and if $\Sigma_1 \succ c_{2, \delta}^{(N)}I$, $\Sigma_1$ is not $\delta$-strictly shrinkable in $N$ steps. When $\delta > 0$, $c_{1, \delta}^{(N)} I \preceq \Sigma_1 \preceq c_{2, \delta}^{(N)} I$ is a sufficient condition for $\delta$-strict shrinkability over a finite horizon, and $\delta-$strict shrinkability implies shrinkability. We will define the upper bound on the finite time horizon required to shrink $a I$ to $b I$, with $c_{1, \delta}^{(N)} \leq b \leq a \leq c_{2, \delta}^{(N)}$ as:
\begin{equation}
\mathbb{T}(a, b, \delta, N) = \left\lceil\frac{a - b}{\delta}\right\rceil(N-1) + N.
\end{equation}

\begin{theorem}\label{lemma: delta_shrinkability}
For any $\delta \in \mathbb{R}_{\geq 0}$, suppose that $c_{1, \delta}^{(N)} = \inf \{c \in \mathbb{R}_{\geq 0}: cI \text{ is $\delta$-strictly shrinkable in }$N$ \text{ steps} \}$ and $c_{2, \delta}^{(N)} = \sup \{c \in \mathbb{R}_{\geq 0}: cI \text{ is $\delta$-strictly shrinkable in }$N$ \text{ steps} \}$. Then, for any initial covariance $\Sigma_1$ under the same control and non-position state constraints:
\begin{itemize}
    \item{If $\Sigma_1 \prec c_{1, \delta}^{(N)} I$, $\Sigma_1$ is not $\delta$-strictly shrinkable over $N$ time steps.}
    \item{If $\delta > 0$ and $c_{1, \delta}^{(N)} I \preceq \Sigma_1 \preceq c_{2, \delta}^{(N)} I$, $\Sigma_1$ is $\delta$-strictly shrinkable over $T = \mathbb{T}(\lambda_{\max}(\Sigma_1), \lambda_{\min}(\Sigma_1), \delta, N)$ time steps.}
    \item{If $c_{2, \delta}^{(N)}I \prec \Sigma_1$, $\Sigma_1$ is not $\delta$-strictly shrinkable over $N$ time steps.}
\end{itemize}
\end{theorem}
If $\delta > 0$ and $c_{1, \delta}^{(N)} I \preceq \Sigma_1 \preceq c_{2, \delta}^{(N)} I$, $\Sigma_1$ is not only $\delta$-strictly shrinkable in finite time; it can also shrink to $cI$ for any $c \in [c_{1, \delta}^{(N)}, c_{2, \delta}^{(N)}]$ in finite time.
\begin{lemma}\label{lemma: delta_shrinkability_induction}
Suppose $\delta \in \mathbb{R}_{>0}, N \in \mathbb{Z}_{>0}$. For any $\Sigma_1$ satisfying $\Sigma_1 \preceq c_{2, \delta}^{(N)}I$ and any $c \in [c_{1, \delta}^{(N)}, \lambda_{\max}(\Sigma_1)]$, there exists a control policy over $T = \mathbb{T}(\lambda_{\max}(\Sigma_1), c, \delta, N)$ steps that steers $\Sigma_1$ to $\Sigma_T \preceq (c-\delta)I$ while satisfying all control and non-position state constraints.
\end{lemma}

A natural corollary to \cref{lemma: delta_shrinkability_induction} is that if $\delta > 0$ and $\Sigma_1 \preceq c_{2, \delta}^{(N)} I$, then $\Sigma_1$ can be reduced to $c_{1, \delta}^{(N)} I$ in finite time. 
\begin{corollary}\label{corollary: finite_time_shrinkability}
Suppose $\delta \in \mathbb{R}_{>0}, N \in \mathbb{Z}_{>0}$. For any $\Sigma_1 \preceq c_{2, \delta}^{(N)}I$, there exists a control policy satisfying all control and non-position state constraints that shrinks $\Sigma_1$ to $\Sigma_T \preceq (c_{1, \delta}^{(N)}-\delta) I$ over $T = \mathbb{T}(\lambda_{\max}(\Sigma_1), c_{1, \delta}^{(N)}, \delta, N)$ time steps.
\end{corollary}
An important consequence of \cref{lemma: delta_shrinkability} is that for any $N \in \mathbb{Z}_{>0}$ and $\delta \geq 0$, $c_{1, \delta}^{(N)} I$ is 0-strictly shrinkable in $N$ steps:
\begin{corollary}\label{corollary: maintainability}
For any $\delta \geq 0$, $N \in \mathbb{Z}_{>0}$, there exists a control policy satisfying all control and non-position state constraints that shrinks $c_{1, \delta}^{(N)} I$ to $\Sigma_N \preceq c_{1, \delta}^{(N)} I$ over $N$ steps.
\end{corollary}
\subsection{Reachability With Position Constraints}
\label{subsec: shrinkability_with_position_constraints}
In the presence of position constraints, the covariance steering problem depends on the state $\mu_k$ and the geometry of the safe set even if $\mu_1$ is stationary and $\mu_k = \mu_1$ for all $k$. We add position constraints to \cref{prob: covariance_shrinking_delta} and define a Gaussian distribution $\mathcal{N}(\mu_1, \Sigma_1)$ with stationary $\mu_1$ to have ${\delta}${-strictly shrinkable covariance in a safe polytope ${\mathcal{S}}$}, with $\mathcal{S} \subset \mathbb{R}^p$, over $N$ steps if \cref{prob: covariance_shrinking_delta_S} is feasible, with $\delta \in \mathbb{R}_{\geq 0}$.
\begin{problem}[Strictly Shrink Covariance by $\delta$ in Polytope]\label{prob: covariance_shrinking_delta_S}
\begin{equation}
\min_{\Sigma_k, U_k, Y_k} 0
\end{equation}
\text{such that for $i\in[M_{x_s}]$, $j\in[M_u]$, $k\in[N]$:}
\begin{align}
&\Phi^{-1}(1-\epsilon)^2({a_i^{s}}^TP\Sigma_kP^T{a_i^{s}}) \leq ({b_i}^{s} - {a_i^s}^TP\mu_1)^2, \label{eq: position_constraint_fixed_delta} \\
& (\ref{eq: control_constraint_fixed_v_delta}), \quad (\ref{eq: non_state_constraint_fixed_delta}), \quad (\ref{eq: convex_covariance_dynamics}), \quad (\ref{eq: schur_complement_constraint}), \quad \lambda_{\max}(\Sigma_N) \leq \lambda_{\min}(\Sigma_1)-\delta. \notag
\end{align}
\end{problem}
We can calculate $c_{1, \delta}^{(N)}(\mu_1, \mathcal{S})$ and $c_{2, \delta}^{(N)}(\mu_1, \mathcal{S})$ such that $c_{1, \delta}^{(N)}(\mu_1, \mathcal{S})I \preceq \Sigma_1 \preceq c_{2, \delta}^{(N)}(\mu_1, \mathcal{S})I$ implies that $\Sigma_1$ is $\delta$-strictly shrinkable in $\mathcal{S}$ over a finite number of time steps. However, the bounds $c_{1, \delta}^{(N)}(\mu_1, \mathcal{S})$ and $c_{2, \delta}^{(N)}(\mu_1, \mathcal{S})$ depend on $\mu_1$ and $\mathcal{S}$ and do not easily generalize to the whole set of stationary means in obstacle-free space. Therefore, we seek a more efficient and generalizable certificate of $\delta$-strict shrinkability in $\mathcal{S}$ while accounting for position constraints and obstacle geometry. Consider the following:
\begin{problem}[Strictly Shrink Covariance by $\delta$ in Ball]\label{prob: covariance_shrinking_r}
\begin{equation}
\min_{\Sigma_k, U_k, Y_k} r
\end{equation}
\text{such that for $j\in[M_u]$, $k\in[N]$:}
\begin{align}
& P\Sigma_kP^T \preceq rI, \quad (\ref{eq: control_constraint_fixed_v_delta}), \quad (\ref{eq: non_state_constraint_fixed_delta}), \notag \\
& \quad (\ref{eq: convex_covariance_dynamics}), \quad (\ref{eq: schur_complement_constraint}), \quad \lambda_{\max}(\Sigma_N) \leq \lambda_{\min}(\Sigma_1)-\delta. \notag
\end{align}
\end{problem}
This problem is convex and is feasible if $\Sigma_1$ is $\delta$-strictly shrinkable in $N$ steps (i.e. if \cref{prob: covariance_shrinking_delta} is feasible). Suppose $r^\star$ is the optimal solution to \cref{prob: covariance_shrinking_r} for a given $\Sigma_1$. Then, for a given $\mu_1, \mathcal{S}$, $\mathcal{N}(\mu_1, \Sigma_1)$ is $\delta$-strictly shrinkable in $\mathcal{S}$ if $\mu_1$ is stationary and $\mathbb{B}_p(P\mu_1, \Phi^{-1}(1-\epsilon)\sqrt{r^\star}) \subset \mathcal{S}$. We state this formally below:
\begin{lemma}
\label{lemma: r_delta_strict_shrinkability}
Suppose $\mu_1$ is stationary and $\Sigma_1$ is $\delta$-strictly shrinkable in $N$ steps (i.e. \cref{prob: covariance_shrinking_delta} is feasible). Then, there exists $r^\star$ that is the optimal solution to \cref{prob: covariance_shrinking_r} for $\Sigma_1$. Also, if $\mathbb{B}_p(P\mu_1, \Phi^{-1}(1-\epsilon)\sqrt{r^\star}) \subset \mathcal{S}$, $\mathcal{N}(\mu_1, \Sigma_1)$ is $\delta$-strictly shrinkable in $\mathcal{S}$.
\end{lemma}

Because $cI$ is $\delta$-strictly shrinkable in $N$ steps for any $c \in [c_{1, \delta}^{(N)}, c_{2, \delta}^{(N)}]$, $r^*$ is well-defined when $\Sigma_1 = cI$ for $c \in [c_{1, \delta}^{(N)}, c_{2, \delta}^{(N)}]$. The function $r^\star(c)$ can be used to derive bounds on the set of $\delta$-strictly shrinkable covariances for a given safe set and initial stationary mean.
\begin{theorem}
\label{thm: r_c_shrinkability}
Suppose $\mathcal{S} \subset \mathbb{R}^p$ is a safe polytope, $\mu_1$ is stationary, $\mathbb{B}_p(P\mu_1, d_1) \subset \mathcal{S}$ and that $\mathcal{S} \subset \text{Int}(\mathbb{B}_p(P\mu_1, d_2))$. Further suppose that $\Phi^{-1}(1-\epsilon)\sqrt{r^\star(c_{d_1})} = d_1$ and $\Phi^{-1}(1-\epsilon)\sqrt{r^\star(c_{d_2})} = d_2$ for some $c_{d_1}, c_{d_2} \in [c_{1, \delta}^{(N)}, c_{2, \delta}^{(N)}]$. Then:
\begin{itemize}
\item{If $\Sigma_1 \prec c_{1, \delta}^{(N)}I$, $\mathcal{N}(\mu_1, \Sigma_1)$ is not $\delta$-strictly shrinkable in $\mathcal{S}$ over $N$ steps.}
\item{If $\delta > 0$ and $c_{1, \delta}^{(N)}I \preceq \Sigma_1 \preceq c_{d_1}I$, $\mathcal{N}(\mu_1, \Sigma_1)$ is $\delta$-strictly shrinkable in $\mathcal{S}$ over $T$ steps, possibly with $T > N$.}
\item{If $\Sigma_1 \succ c_{d_2}I$, $\mathcal{N}(\mu_1, \Sigma_1)$ is not $\delta$-strictly shrinkable in $\mathcal{S}$ over $N$ steps.}
\end{itemize}
\end{theorem}
Finally, if $(c_{1, \delta}^{(N)}, c_d)$ is nonempty, $\delta > 0$, and $\Sigma_1 \preceq c_d I$ where $c_{1, \delta}^{(N)}I \preceq \Sigma_1 \preceq c_dI$ implies $\delta$-strict shrinkability over $T$ steps in $\mathcal{S}$, then $\Sigma_1$ can be reduced to $c_{1, \delta}^{(N)}I$ in a finite number of time steps while remaining in $\mathcal{S}$.
\begin{corollary}
\label{corollary: constrained_path_existence}
Suppose $\delta \in \mathbb{R}_{> 0}$, $N \in \mathbb{Z}_{>0}$, $\mathbb{B}_p(P\mu_1, d_1) \subset \mathcal{S}$ where $\mathcal{S} \subset \mathbb{R}^p$, and $\Phi^{-1}(1-\epsilon)\sqrt{r^\star(c_{d_1})} = d_1$ for some $c_{d_1} \in [c_{1, \delta}^{(N)}, c_{2, \delta}^{(N)}]$. Then, there exists a control policy satisfying all control, non-position state constraints, and position constraints (from $\mathcal{S}$) that shrinks $\Sigma_1$ to $\Sigma_T \preceq (c_{1, \delta}^{(N)}-\delta) I$ over $T = \mathbb{T}(\lambda_{\max}(\Sigma_1), c_{1, \delta}^{(N)}, \delta, N)$ time steps.
\end{corollary}
\section{PRISM Global Planner}\label{sec: approach}

\subsection{Offline Covariance Shrinking \& Graph Construction}
The first key idea of PRISM is to shrink the initial covariance to $c_{1, \delta}^{(N)} I$ and then plan in deterministic space with a constant covariance of $c_{1, \delta}^{(N)} I$. PRISM separates covariance and mean steering while guaranteeing completeness under mild assumptions on the start and goal distribution (see \cref{thm: completeness_without_obstacles} and \cref{thm: completeness_with_obstacles} for details). As $\delta$ decreases and $N$ increases, the upper bound on the number of time steps required to shrink any covariance to $c_{1, \delta}^{(N)}I$ increases. Then, $\delta$ and $N$ trade off between the size of the range $[c_{1, \delta}^{(N)}, c_{2, \delta}^{(N)}]$ of shrinkable isotropic covariances and the maximum number of time steps required to reach $c_{1, \delta}^{(N)}I$. The first step of PRISM is to calculate $c_{1, \delta}^{(N)}$ and $c_{2, \delta}^{(N)}$ for fixed parameters $\delta$ and $N$ by minimizing or maximizing $c$ subject to the constraints given in \cref{prob: get_c_1_delta}.
\begin{problem}[{Find $c_{1, \delta}^{(N)}$}]
\label{prob: get_c_1_delta}
$$\min_{\Sigma_k, U_k, Y_k} c$$
such that for $i \in [M_{x_s}], j \in [M_u], k \in [N]$:
$$(\ref{eq: control_constraint_fixed_v_delta}), (\ref{eq: non_state_constraint_fixed_delta}), (\ref{eq: convex_covariance_dynamics}), (\ref{eq: schur_complement_constraint}), \Sigma_1 = cI, \lambda_{\max}(\Sigma_N) \leq \lambda_{\min}(\Sigma_1)-\delta.$$
\end{problem}
Recall from \cref{corollary: maintainability} and \cref{lemma: r_delta_strict_shrinkability} that there exists a control policy obeying all control and non-position state constraints such that $\Sigma_1 = c_{1, \delta}^{(N)} I$, $\Sigma_N \preceq c_{1, \delta}^{(N)} I$, and $P\Sigma_kP^T \preceq r^\star I$ for all $k$, with $r^\star$ minimizing $r$. PRISM will compute this policy and the resulting fixed feedback control and covariance trajectories $K_{1:N-1}^{(c_{1, \delta})}$ and $\Sigma_{1:N}^{(c_{1, \delta})}$. Next, PRISM will use this solution to deflate the safe state and control sets $\mathcal{X}_p$, $\mathcal{X}_s$, and $\mathcal{U}$ to $\mathcal{X}_p^-$, $\mathcal{X}_s^-$, and $\mathcal{U}^-$ and to inflate the obstacle set $\mathcal{P}$ to $\mathcal{P}^+$ (see Alg. \ref{alg: prism_constraint_tightening}). Then, PRISM will build a graph for complete planning in the tightened deterministic space.
\setlength{\textfloatsep}{10pt}
\begin{algorithm}[b]
\caption{Tighten constraints and inflate obstacles.}
\label{alg: prism_constraint_tightening}
\begin{algorithmic}[1]
    \Require $K^{(c_{1, \delta})}_{1:N-1}, \Sigma_{1:N}^{(c_{1, \delta})}, \mathcal{U}, \mathcal{X}_p, \mathcal{X}_s, \mathcal{P}$  
    \Ensure  $\mathcal{U}^-, \mathcal{X}_p^-, \mathcal{X}_s^-, \mathcal{P}^+$ 
    \For {$i \in [M_{x_p}]$}
    \State $\delta_{b_i^{x_p}} \gets \max\limits_{k\in[N]} \sqrt{{a_i^{x_p}}^TP\Sigma_k^{(c_{1, \delta})}P^T{a_i^{x_p}}}$
    \EndFor
    \For {$i \in [M_{x_s}]$}
        \State $\delta_{b_i^{x_s}} \gets \max\limits_{k\in [N]} \sqrt{{a_i^{x_s}}^TS\Sigma_k^{(c_{1, \delta})}S^T{a_i^{x_s}}}$ 
    \EndFor
    \For {$i \in [M_u]$}
    \State $\delta_{b_i^{u}} \gets \max\limits_{k\in [N-1]} \sqrt{{a_i^{u}}^TK_k^{(c_{1, \delta})}\Sigma_k^{(c_{1, \delta})}{K_k^{(c_{1, \delta})}}^T{a_i^{u}}}$
    \EndFor
    \State $r^\star \gets \max\limits_{k\in [N]} \lambda_{\max}(P\Sigma_k^{(c_{1, \delta})}P^T)$
    \State $\mathcal{U}^- \gets \bigcap_{i=1}^{M_u}\{\mathbf{u} \!\in\!  \mathbb{R}^m: {a_i^u}^T\mathbf{u} \leq b_i^u - \Phi^{-1}(1-\epsilon)\delta_{b_i^u} \}$
    \State $\mathcal{X}_p^- \gets \bigcap_{i=1}^{M_{x_p}}\{\mathbf{x} \!\in\!  \mathbb{R}^n\!:\! {a_i^{x_p}}^TP\mathbf{x} \leq b_i^{x_p} - \Phi^{-1}(1-\epsilon)\delta_{b_i^{x_p}} \}$
    \State $\mathcal{X}_s^- \gets \bigcap_{i=1}^{M_{x_s}}\{\mathbf{x} \!\in\!  \mathbb{R}^n\!:\! {a_i^{x_s}}^TS\mathbf{x} \leq b_i^{x_s} - \Phi^{-1}(1-\epsilon)\delta_{b_i^{x_s}} \}$
    \State $\mathcal{P}^+ \!\!=\!\!\bigcup_{j=1}^{M_p} \bigcap_{i=1}^{M_{p_j}}\{\mathbf{x}\!\in\!  \mathbb{R}^n\!\!:\! {a_i^{p_j}}^TP\mathbf{x} \!\leq\! {b_i^{p_j}} \!+\! \Phi^{-1}(1-\epsilon)\sqrt{r^\star}\}$
\end{algorithmic}
\end{algorithm}

We assume that $\mathcal{X}_p^- \setminus \mathcal{P}^+$ can be written as the Minkowski sum $\mathcal{S}^- \oplus \mathbb{B}_p(0, \epsilon_p)$ for some connected set $\mathcal{S}^-$ and $\epsilon_p > 0$; intuitively, $\mathcal{X}_p^- \setminus \mathcal{P}^+$ is connected by a tube of radius $\epsilon_p > 0$. This always holds when the narrowest passage in $\mathcal{X}_p \setminus \mathcal{P}$ has radius $> \Phi^{-1}(1-\epsilon)\sqrt{r^\star}$. Then, $\mathcal{X}_p^- \setminus \mathcal{P}^+$ can be exactly decomposed into a union of overlapping convex polytopes $\bigcup_{i=1}^{M_{s-}} \mathcal{S}_i$ such that for any $\mathcal{S}_A, \mathcal{S}_B$ where $\mathcal{X}_p^- \setminus \mathcal{P}^+ = \mathcal{S}_A \cup \mathcal{S}_B = \bigcup_{i=1}^{M_{s-}} \mathcal{S}_i$, $\text{Int}(\mathcal{S}_A \cap \mathcal{S}_B) \neq 0$. PRISM achieves such a decomposition using a modified version of \cite{shaikh2025exact} that iteratively grows each cell after the initial enumeration of obstacle-free cells and then prunes duplicates, resulting in an exact space decomposition with a relatively low number of relatively large non-disjoint convex sets.

We also assume that $(A, B)$ is controllable. Then, it is always possible to find a finite-time deterministic path between two stationary means that lie strictly in the interior of the same safe bounded convex set, while satisfying any nonempty set of tightened state and control constraints. We formally state these assumptions and related lemmas as:

\begin{assumption}\label{asm: connectedness}
$\mathcal{X}_p^- \setminus \mathcal{P}^+$ can be written as the Minkowski sum $\mathcal{S}^- \oplus \mathbb{B}_p(0, \epsilon_p)$ for some connected set $\mathcal{S}^-$ and $\epsilon_p > 0$.
\end{assumption}
\begin{lemma}\label{lemma: decomposition}
If \cref{asm: connectedness} holds, $\mathcal{X}_p^- \setminus \mathcal{P}^+$ can be exactly written as union of convex polytopes $\mathcal{X}_p^- \setminus \mathcal{P}^+ = \bigcup_{i=1}^{M_{s-}} \mathcal{S}_i$ such that for any $\mathcal{S}_A, \mathcal{S}_B$ where $\mathcal{X}_p^- \setminus \mathcal{P}^+ = \mathcal{S}_A \cup \mathcal{S}_B = \bigcup_{i=1}^{M_{s-}} \mathcal{S}_i$, $\text{Int}(\mathcal{S}_A \cap \mathcal{S}_B) \neq \emptyset$.
\end{lemma}
\begin{assumption}\label{asm: controllability}
	$(A, B)$ is controllable.
\end{assumption}

\begin{lemma}\label{thm: controllability}
If \cref{asm: controllability} holds, for any position-constrained polytope $\mathcal{S} := \{\mathbf{x} \!\in\!  \mathbb{R}^n : \cap_{i=1}^{M_s} {a_i^s}^TP\mathbf{x} \leq b_i^s \}$, if $\mu_\mathcal{I}, \mu_\mathcal{G} \!\in\!  \text{Int}(\mathcal{S})$, $A\mu_\mathcal{I} = \mu_\mathcal{I}$, $A\mu_\mathcal{G} = \mu_\mathcal{G}$, and $\mathcal{S}$ is a closed set of finite volume, then \cref{prob: mean_steering_single_polytope} has a feasible solution for some finite $T$ if $\mathcal{U}^-$ and $\mathcal{X}_s^-$ both contain the origin and have nonempty interior.
\end{lemma}
\begin{problem}[{Steer Mean Inside Polytope}]\label{prob: mean_steering_single_polytope}
\begin{equation}
\min_{\mu_k, \mathbf{v}_k} \sum_{k=1}^{T-1} \mathbf{v}_k^TR\mathbf{v}_k
\end{equation}
such that for all $i \in [M_{x_s}], j \in [M_u], 
k \in [T], \ell \in [M_s]$:
\begin{align}
&\quad \mu_1 = \mu_\mathcal{I}, \mu_{k+1} = A\mu_k + B\mathbf{v}_k, \quad {a_i^{x_s}}^TS\mu_k \leq b_i^{x_s} - \delta_{b_i^{x_s}}, \notag \\
&\quad {a_j^u}^T\mathbf{v}_k \leq b_j^u - \delta_{b_j^u},\quad {a_\ell^s}^TP\mu_k \leq b_\ell^s,  \quad \mu_T = \mu_\mathcal{G}. \notag
\end{align}
\end{problem}
By \cref{thm: controllability}, if $\mu_\mathcal{I} \!\in\!  \mathcal{S}_\mathcal{I}$, $\mu_\mathcal{J} \!\in\!  \mathcal{S}_\mathcal{J}$, with $\text{Int}(\mathcal{S}_\mathcal{I} \cap \mathcal{S}_\mathcal{J}) \neq \emptyset$, $A\mu_\mathcal{I} = \mu_\mathcal{I}$, and $A\mu_\mathcal{J} = \mu_\mathcal{J}$, there must exist $T$ such that there is a feasible solution to \cref{prob: mean_steering_single_polytope} steering $\mu_\mathcal{I}$ to some stationary $\mu \!\in\!  \text{Int}(\mathcal{S}_\mathcal{I} \cap \mathcal{S}_\mathcal{J})$ over $T$ or fewer steps and there is a feasible solution to \cref{prob: mean_steering_single_polytope} steering the same stationary $\mu \!\in\!  \text{Int}(\mathcal{S}_\mathcal{I} \cap \mathcal{S}_\mathcal{J})$ to $\mu_\mathcal{J}$ over $T$ or fewer steps. Therefore, the following optimization problem must be feasible for any trajectory length $\geq 2T-1$ for some $T$:
\begin{problem}[{Steer Mean Across Intersecting Polytopes}]\label{prob: mean_steering}

\begin{equation}
\min_{\mu_k, \mathbf{v}_k} \sum_{k= 1}^{2(T-1)} \mathbf{v}_k^TR\mathbf{v}_k
\end{equation}
s.t. for all $i \in [M_{x_s}], j \in [M_u], k \in [2T-1]$:
\begin{align}
&\mu_{k+1} = A\mu_k + B\mathbf{v}_k, \quad \mu_1 = \mu_\mathcal{I}, \quad \mu_{2T-1} = \mu_\mathcal{J}, \notag \\
&{a_i^{x_s}}^TS\mu_k \leq b_i^{x_s} - \delta_{b_i^{x_s}}, \quad {a_j^u}^T\mathbf{v}_k \leq b_j^u - \delta_{b_j^u},\notag
\end{align}
for all $k \in [T]$ and $\ell_\mathcal{I} \in [M_{\mathcal{S}_\mathcal{I}}]$, ${a_{\ell_\mathcal{I}}^{S_\mathcal{I}}}^TP\mu_k \leq b_{\ell_\mathcal{I}}^{S_\mathcal{I}},$
and for all $k = T-1 + [T]$ and $\ell_\mathcal{J} \in [M_{\mathcal{S}_\mathcal{J}}]$, ${a_{\ell_\mathcal{J}}^{S_\mathcal{J}}}^TP\mu_k \leq b_{\ell_\mathcal{J}}^{S_\mathcal{J}}.$
\end{problem}

After decomposing the safe state space, PRISM creates one node at the centroid of each convex set, and finds edges between nodes at intersecting sets by solving \cref{prob: mean_steering} with increasing values of $t \!\in\!  \mathbb{Z}_{>0}$, with $T = (N-1)t + 1$, until a feasible solution is found (see \cref{fig: mean_steering} and lines 4-6 in \cref{alg: prism_mean_steering}). These edges are lifted back into belief space by combining the mean trajectory found by \cref{prob: mean_steering} with the covariance trajectory found by \cref{prob: covariance_shrinking_r} (lines 7-11 in \cref{alg: prism_mean_steering}). In belief space, these edges are associated with inflated versions of $\mathcal{S}_i$, $\mathcal{S}_j$ and $\mathcal{S}_i \cap \mathcal{S}_j$ that are subsets of $\mathcal{X}_p \setminus \mathcal{P}$ (lines 2-3 and 10 in \cref{alg: prism_mean_steering}). A full description of edge construction and safe set inflation is provided in the proof of \cref{lemma: mean_soundness} (see Appendix \ref{sec:appendix}). The offline graph construction algorithm is presented in \cref{alg: prism_offline_phase}.

\begin{figure}[t]
\centering
\includegraphics[width=0.9\columnwidth]{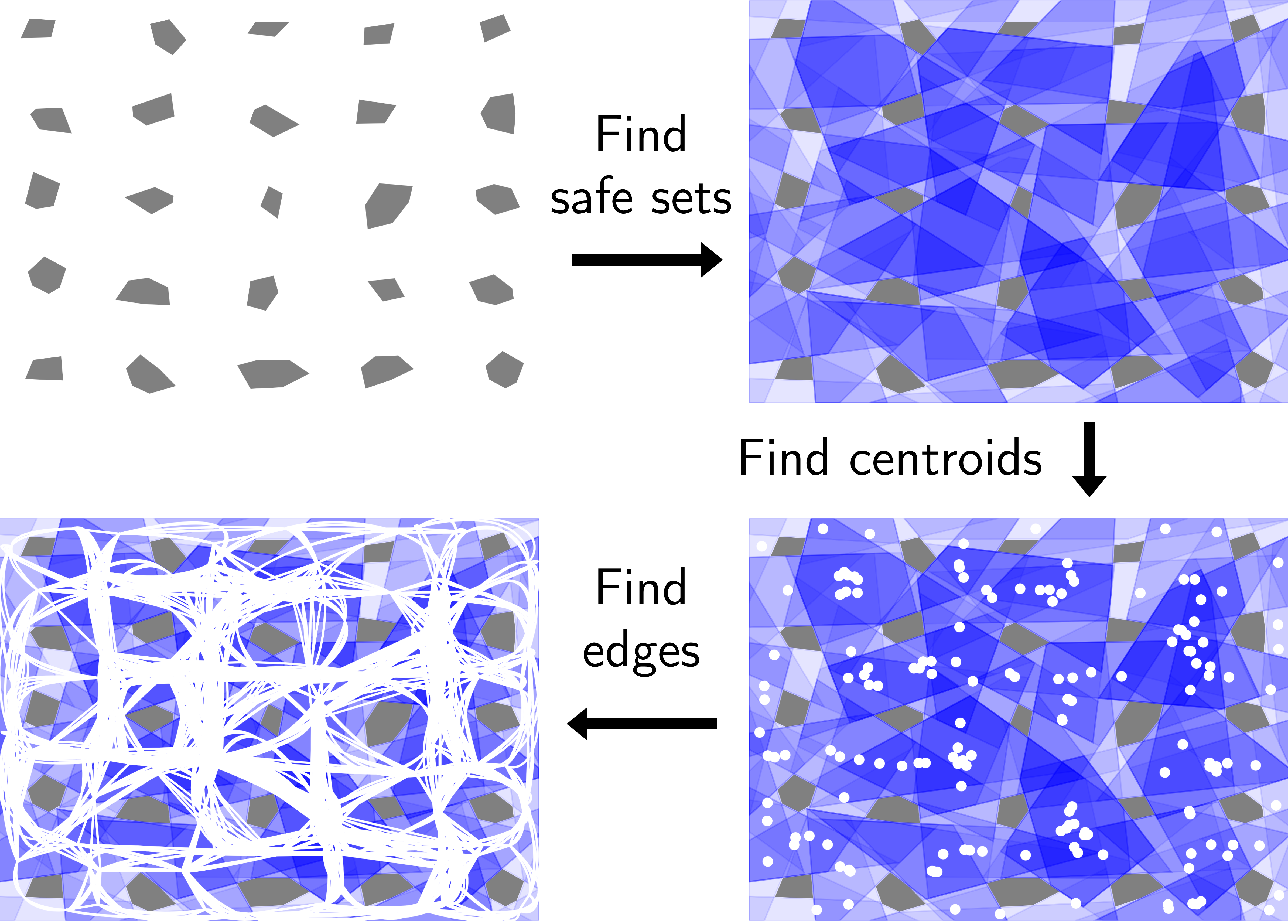}
\vspace{-0.25cm}
\caption{Clockwise from top left: Obstacle environment, overlapping safe set decomposition, safe set decomposition with nodes, and PRISM mean steering graph with edges. Obstacles are gray, safe sets are transparent blue, and centroids and edges are white.}
\label{fig: mean_steering}
\end{figure}

\setlength{\textfloatsep}{10pt}
\begin{algorithm}[htbp]
\caption{Find an edge with tightened constraints.}
\label{alg: prism_mean_steering}
\begin{algorithmic}[1]
    \Require $N, R,  \mu_i, \mu_j, \mathcal{U}^-, \mathcal{X}_s^-, \mathcal{S}_i, \mathcal{S}_j, \Sigma_{1:N}^{(c_{1, \delta})}, K_{1:N-1}^{(c_{1, \delta})}$ 
    \Ensure  $e$ 
    \State $t, \text{success} \gets 1, \textbf{false}$
    \State $\mathcal{S}_i^+, \mathcal{S}_j^+ \gets \text{INFLATE}(\mathcal{S}_i, \mathcal{X}_p, \mathcal{P}), \text{INFLATE}(\mathcal{S}_j, \mathcal{X}_p, \mathcal{P})$
    \State $\mathcal{S}_{i \cap j}^+ \gets \text{INFLATE}(\mathcal{S}_i \cap \mathcal{S}_j, \mathcal{X}_p, \mathcal{P})$
    \While {\text{\textbf{not} success}}
        \State $T \gets t(N-1) + 1$
        \State $\mu_{1:2T-1}, \mathbf{v}_{1:2T-2}, \text{success} \gets \text{Prob. \ref{prob: mean_steering}}(\mu_i, \mu_j, \mathcal{U}^-,$
        $ \mathcal{X}_s^-, \mathcal{S}_i, \mathcal{S}_j, T, R)$
        \If {\text{success}}
             \State $K_{1:2T-2} \gets \{K_{1:N-1}^{(c_{1, \delta})} \times 2t\}$
             \State $\Sigma_{1:2T-1} \gets \text{PROPAGATE}(c_{1, \delta}^{(N)}I, K_{1:2T-2}, A, B)$
             \State $\mathcal{S}_{1:2T-1} \gets \{\mathcal{S}_i^+ \times (T-1), \mathcal{S}_{i \cap j}^+, \mathcal{S}_j^+ \times (T-1)\}$
			\State $e \gets (v_i, v_j, \mu_{1:2T-1}, \mathbf{v}_{1:2T-2}, \Sigma_{1:2T-1}, K_{1:2T-2},$
            $  \mathcal{S}_{1:2T-1})$
        \EndIf
        \State $t \gets t + 1$
    \EndWhile
\end{algorithmic}
\end{algorithm}

\setlength{\textfloatsep}{10pt}
\begin{algorithm}[htbp]
\caption{PRISM offline graph construction.}
\label{alg: prism_offline_phase}
\begin{algorithmic}[1]
    \Require $\delta, N, R, \mathcal{X}_p, \mathcal{X}_s, \mathcal{U}, \mathcal{P}$ 
    \Ensure  $\Sigma_{1:N}^{(c_{1, \delta})}, K_{1:N-1}^{(c_{1, \delta})}, \mathcal{X}_p^-, \mathcal{X}_s^-, \mathcal{U}^-, \mathcal{P}^+, \mathcal{E}, \mathcal{V}$ 
    \State $c_{1, \delta}^{(N)} \gets \text{\cref{prob: get_c_1_delta}}(\delta, N, \mathcal{X}_s, \mathcal{U})$
    \State $\Sigma_{1:N}^{(c_{1, \delta})}, K_{1:N-1}^{(c_{1, \delta})} \gets \text{\cref{prob: covariance_shrinking_r}}(c_{1, \delta}^{(N)}I, \mathcal{U}, \mathcal{X}_s)$
    \State $\mathcal{U}^-, \mathcal{X}_p^-, \mathcal{X}_s^-,\!\mathcal{P}^+\!\gets\! \text{Alg.\! \ref{alg: prism_constraint_tightening}}(K^{(c_{1, \delta})}_{1:N-1}, \Sigma_{1:N}^{(c_{1, \delta})}, \mathcal{U}, \mathcal{X}_p, \mathcal{X}_s, \mathcal{P})$
    \State $\bigcup_{i=1}^{M_{s^-}} \mathcal{S}_i \gets \text{DECOMPOSE}(\mathcal{X}_p^- \setminus \mathcal{P}^+)$
    \State $\mathcal{E}, \mathcal{V} \gets \emptyset, \emptyset$
    \For {$i \in [M_s]^-$}
        \State $\mu_i \gets \text{CENTROID}(\mathcal{S}_i)$
        \State $\mathcal{V} \gets \mathcal{V} \cup (\mathcal{S}_i, \mu_i)$
    \EndFor
    \For {$i \in [M_s]^-$}
        \For {$j \in [M_s]^-$}
            \State $(\mu_i, \mathcal{S}_i), (\mu_j, \mathcal{S}_j) \gets \mathcal{V}_i, \mathcal{V}_j$
            \If {$\text{Int}(\mathcal{S}_i \cap \mathcal{S}_j) \neq \emptyset$}
                \State $e \gets \text{Alg. \ref{alg: prism_mean_steering}}(N, R, \mu_i, \mu_j, \mathcal{U}^-, \mathcal{X}_s^-, \mathcal{S}_i, \mathcal{S}_j, \ \ \ \ $
                $\Sigma_{1:N}^{(c_{1, \delta})}, K_{1:N-1}^{(c_{1, \delta})})$
                \State $\mathcal{E} \gets \mathcal{E} \cup \{e\}$
            \EndIf
        \EndFor
    \EndFor
\end{algorithmic}
\end{algorithm}
 
\subsection{Online Planning through the PRISM Graph}
Given an initial covariance $\Sigma_\mathcal{I}$, PRISM seeks to safely reduce the system covariance to $c_{1, \delta}^{(N)} I$. If the initial mean $\mu_\mathcal{I}$ is stationary and there are no position constraints, this is achieved by repeatedly minimizing $\lambda_{\max}(\Sigma_N)$ subject to the constraints given in \cref{prob: covariance_shrinking_delta}. In the presence of position constraints, by \cref{thm: r_c_shrinkability} and \cref{corollary: constrained_path_existence}, this can be achieved by repeatedly minimizing $\lambda_{\max}(\Sigma_N)$ subject to the constraints given in \cref{prob: covariance_shrinking_delta_S} as long as $\mu_\mathcal{I}$ is stationary and $\mathbb{B}_p(P\mu_\mathcal{I}, d) \subset (\mathcal{X}_p \setminus \mathcal{P})$, with $d = \Phi^{-1}(1-\epsilon)\sqrt{r^\star(\lambda_{\max}(\Sigma_\mathcal{I}))}$.

However, if $\mu_\mathcal{I}$ is non-stationary or \cref{corollary: constrained_path_existence} does not hold (i.e. because $\mathbb{B}_p(P\mu_\mathcal{I}, d) \not\subset (\mathcal{X}_p \setminus \mathcal{P})$), shrinking the system covariance in place is not guaranteed to succeed. Then, PRISM will find a safe position polytope $\mathcal{S} \subset \mathcal{X}_p \setminus \mathcal{P}$ which contains the initial distribution with high probability, then steer to a stationary mean while optimizing for position constraint looseness. The idea of this maneuver is to reach a stationary mean $\mu_T$ while approximately optimizing for satisfaction of \cref{corollary: constrained_path_existence} at time $T$. PRISM solves \cref{prob: stopping_phase_convex} with $T = Nt$ for increasing values of $t \!\in\!  \mathbb{Z}_{>0}$ until a maximum path length parameter $T_{\max}$ is reached (at which point PRISM terminates without returning a path) or a solution is found with $\mu_T = A\mu_T$ and $(\mu_T, \Sigma_T)$ satisfying \cref{corollary: constrained_path_existence} (lines 3-10 in \cref{alg: prism_online}):
\begin{problem}[{Steer to Stationary Mean with Position Constraints}]\label{prob: stopping_phase_convex}

\begin{align}
\max \min_{i, k} b_i^s - &\Big({a_i^s}^TP\mu_k \notag \\
&+ \Phi^{-1}(1-\epsilon)\mathbb{M}_1(P^Ta_i^s, \lambda_{\max}(\Sigma_k)I, \lambda_{\max}(\Sigma_{r,k}^i)I)\Big) \notag
\end{align}
such that for $k \in [T]$, $i \in [M_{x_s}]$, $j \in [M_u]$, 
$\ell \in [M_{s}]$:
\begin{align}
&(\ref{eq: convex_covariance_dynamics}), (\ref{eq: mean_dynamics}), (\ref{eq: schur_complement_constraint}), (\ref{eq: convex_ctrl_chance_constraint}), (\ref{eq: convex_nonpos_state_chance_constraint}), (\ref{eq: convex_pos_chance_constraint}), \notag \\
& A\mu_T = \mu_T, \quad \mu_1 = \mu_\mathcal{I}, \quad \Sigma_1 = \Sigma_\mathcal{I}, \quad \Sigma_T \preceq c_{2, \delta}^{(N)}I.
\end{align}
\end{problem}
After solving \cref{prob: stopping_phase_convex} and reaching a stationary mean $\mu_T = A\mu_T$ and covariance $\Sigma_T$ such that \cref{corollary: constrained_path_existence} holds, PRISM repeatedly minimizes $\lambda_{\max}(\Sigma_N)$ subject to the constraints in \cref{prob: covariance_shrinking_delta_S} (lines 11-16 in \cref{alg: prism_online}):
\begin{problem}[{Shrink Covariance with Position Constraints}]\label{prob: shrinking_subject_to_constraints}

\begin{equation}
\min_{\Sigma_k, U_k, Y_k} \lambda_{\max}(\Sigma_N)
\end{equation}
s.t. for $k \in [T]$, $i \in [M_{x_s}]$, $j \in [M_u]$, 
$\ell \in [M_{s}]$: $(\ref{eq: convex_covariance_dynamics}), (\ref{eq: schur_complement_constraint}), (\ref{eq: control_constraint_fixed_v_delta}),(\ref{eq: non_state_constraint_fixed_delta}), (\ref{eq: position_constraint_fixed_delta}), \mu_1 = \mu_T, \Sigma_1 = \Sigma_T.$
\end{problem}
By \cref{corollary: constrained_path_existence}, this will reduce the covariance to $\preceq (c_{1, \delta}^{(N)}-\delta)I$ within a finite number of time steps. Then, PRISM connects to the graph constructed offline by \cref{alg: prism_offline_phase}. First, PRISM finds a centroid $\mu_i$ and associated safe polytope $\mathcal{S}_i$ in the PRISM graph such that $\mu_T \!\in\!  \mathcal{S}_i$ (line 17 in \cref{alg: prism_online}). As long as $\mu_T \!\in\!  \mathcal{X}_p^- \setminus \mathcal{P}^+$, such a centroid and safe polytope must exist. PRISM then solves \cref{prob: mean_steering} to find a mean trajectory and reuses the covariance trajectory from \cref{prob: covariance_shrinking_r}, constructing an edge that connects $\mathcal{N}(\mu_T, c_{1, \delta}^{(N)}I)$ to a node in the PRISM graph (line 18 in \cref{alg: prism_online}). The goal distribution $\mathcal{N}(\mu_\mathcal{G}, \Sigma_\mathcal{G})$ is similarly connected to the PRISM graph (lines 17 and 19 in \cref{alg: prism_online}). Finally, PRISM performs a graph search to connect the start and goal distributions (line 20 in \cref{alg: prism_online}).

\subsection{Planning in the Absence of Position Constraints}\label{subsec: planning_without_position_constraints}
In the absence of obstacle constraints and other position constraints, PRISM does not build a graph. If $\mu_\mathcal{I}$ is non-stationary, PRISM solves a modified version of \cref{prob: stopping_phase_convex} with $T = (N-1)t + 1$ for increasing values of $t \!\in\!  \mathbb{Z}_{>0}$ until $\mu_T = A\mu_T$ or $T \geq T_{\max}$ (at which point PRISM terminates):
\begin{problem}[{Steer to Stationary Mean}]\label{prob: stopping_phase_convex_no_pos_constraints}
\begin{equation}
\min_k \lambda_{\max}(\Sigma_k)
\end{equation}
s.t. for $k\! =\! [T]$, $i\! =\! [M_{x_s}]$, $j\! =\! [M_u]$: $(\ref{eq: convex_covariance_dynamics}), (\ref{eq: mean_dynamics}), (\ref{eq: schur_complement_constraint}), (\ref{eq: convex_ctrl_chance_constraint}), (\ref{eq: convex_nonpos_state_chance_constraint}),\   A\mu_T\! =\! \mu_T, \ \mu_1\! =\! \mu_\mathcal{I}, \ \Sigma_1\! = \Sigma_\mathcal{I}, \  \Sigma_T \preceq c_{2, \delta}^{(N)}I.$
\end{problem}

After solving \cref{prob: stopping_phase_convex_no_pos_constraints}, PRISM will repeatedly minimize $\lambda_{\max}(\Sigma_N)$ subject to the constraints in \cref{prob: covariance_shrinking_delta}:
\begin{problem}[{Shrink Covariance}]\label{prob: shrinking_without_position_constraints}

\begin{equation}
\min_{\Sigma_k, U_k, Y_k} \lambda_{\max}(\Sigma_N)
\end{equation}
s.t. for $k \in [T]$, $i \in [M_{x_s}]$, $j \in [M_u]$, 
$\ell \in [M_{s}]$: $(\ref{eq: convex_covariance_dynamics}), (\ref{eq: schur_complement_constraint}), (\ref{eq: control_constraint_fixed_v_delta}),(\ref{eq: non_state_constraint_fixed_delta}), \mu_1 = \mu_T, \Sigma_1 = \Sigma_T.$
\end{problem}
After shrinking the covariance to $\preceq c_{1, \delta}^{(N)}I$, which is possible in a finite number of time steps by \cref{corollary: finite_time_shrinkability}, PRISM will then find a position polytope $\mathcal{S}$ such that the current mean and the goal mean $\mu_\mathcal{G}$ both are in $\text{Int}(\mathcal{S})$, and use \cref{alg: prism_mean_steering} to connect the current distribution to $\mathcal{N}(\mu_\mathcal{G}, c_{1, \delta}^{(N)}I)$.
\setlength{\textfloatsep}{10pt}
\begin{algorithm}[htbp]
\caption{PRISM online planning phase}
\label{alg: prism_online}
\begin{algorithmic}[1]
	\Require $\mathcal{I}, \mathcal{G}, c_{1, \delta}^{(N)}, \mathcal{X}_p, \mathcal{X}_s, \mathcal{U}, \mathcal{P}, \Sigma_{1:N}^{(c_{1, \delta})}, K_{1:N-1}^{(c_{1, \delta})}, \mathcal{X}_p^-, R$
$\mathcal{X}_s^-, \mathcal{U}^-, \mathcal{P}^+, \mathcal{E}, \mathcal{V}, T_{\max}, \Sigma_r^{1:M_{s_k}}, \Sigma_r^{1:M_{x_s}}, Y_r^{1:M_u}, N$ 
    \Ensure $\mathcal{T}$
    \State $(\Sigma_T, \mu_T), t, T, \mathcal{T} \gets (\Sigma_\mathcal{I}, \mu_\mathcal{I}), 1, N, \emptyset$
    \State $\mathcal{S} \gets \text{FIND-SAFE-POLYTOPE}(\mu_\mathcal{I}, \Sigma_\mathcal{I}, \mathcal{P}, \mathcal{X}_p)$
    \State $\text{safe} \gets (\mu_T = A\mu_T) \textbf{ and } (\mathbb{B}_p(P\mu_T, \Phi^{-1}(1-\epsilon)\sqrt{r^\star(\lambda_{\max}(\Sigma_T))}) \subset \mathcal{S})$
    \While {(\textbf{not} \text{ safe }) \textbf{and } $T \leq T_{\max}$}
        \State $\mu_{1:T}, \mathbf{v}_{1:T-1}, \Sigma_{1:T}, K_{1:T-1}, \text{success} \gets \text{Prob. \ref{prob: stopping_phase_convex}}($
        $\mu_\mathcal{I}, \Sigma_\mathcal{I}, T, c_{2, \delta}^{(N)}, \mathcal{X}_s, \mathcal{U}, \mathcal{S}, \Sigma_r^{1:M_{s_k}}, \Sigma_r^{1:M_{x_s}}, Y_r^{1:M_u})$
        \If {\text{success}}
            \State $\mathcal{T} \gets \mathcal{T} \cup (\mu_{1:T}, \mathbf{v}_{1:T-1}, \Sigma_{1:T}, K_{1:T-1}, \{\mathcal{S} \times T \})$
            \State $\text{safe} \gets \mathbb{B}_p(P\mu_T, \Phi^{-1}(1-\epsilon)\sqrt{r^\star(\lambda_{\max}(\Sigma_T))}) \subset \mathcal{S}$
        \EndIf
        \State $t, T \gets t + 1, N(t+1)$
    \EndWhile
    \If {\textbf{not} \text{ safe }}
        \textbf{return } \text{fail}
    \EndIf
    \State $c \gets \lambda_{\max}(\Sigma_T)$
    \While {$c > c_{1, \delta}^{(N)}$}
        \State $\Sigma_{1:N}, K_{1:N-1} \gets \text{Prob. \ref{prob: shrinking_subject_to_constraints}}(\mu_T, \Sigma_T, N, \mathcal{X}_s, \mathcal{U}, \mathcal{S})$
        \State $\mu_{1:N}, \mathbf{v}_{1:N-1} \gets \{\mu_T \times N\}, \{0 \times (N-1)\}$
        \State $c, \Sigma_T \gets \lambda_{\max}(\Sigma_N), \Sigma_N$
        \State $\mathcal{T} \gets \mathcal{T} \cup (\mu_{1:N}, \mathbf{v}_{1:N-1}, \Sigma_{1:N}, K_{1:N-1}, \{\mathcal{S} \times N \})$
    \EndWhile
    \State $(\mu_i,\!\mathcal{S}_i), (\mu_j,\!\mathcal{S}_j)\!\! \gets\!\! \text{GET-NODE}(\mu_T\!,\!\mathcal{V}), \text{GET-NODE}(\mu_\mathcal{G}\!,\!\mathcal{V})$
    \State $e_i\! \gets\! \text{Alg.\! \ref{alg: prism_mean_steering}}(N,\! R,\! \mu_T,\! \mu_i, \mathcal{U}^-, \mathcal{X}_s^-, \mathcal{S}_i, \mathcal{S}_i, \Sigma_{1:N}^{(c_{1, \delta})}, K_{1:N-1}^{(c_{1, \delta})})$
    \State $e_j \!\gets\! \text{Alg.\! \ref{alg: prism_mean_steering}}(N,\! R,\! \mu_j,\! \mu_\mathcal{G},\! \mathcal{U}^-, \mathcal{X}_s^-, \mathcal{S}_j, \mathcal{S}_j, \Sigma_{1:N}^{(c_{1, \delta})}, K_{1:N-1}^{(c_{1, \delta})})$
    \State $\mathcal{T} \gets \mathcal{T} \cup e_i \cup \text{GRAPH-SEARCH}(\mathcal{V}, \mathcal{E}, v_i, v_j) \cup e_j$
\end{algorithmic}
\end{algorithm}

\section{PRISM Local Planner}\label{subsec: local_refinement}
\begin{figure}[tbp]
\centering
\includegraphics[width=0.9\columnwidth]{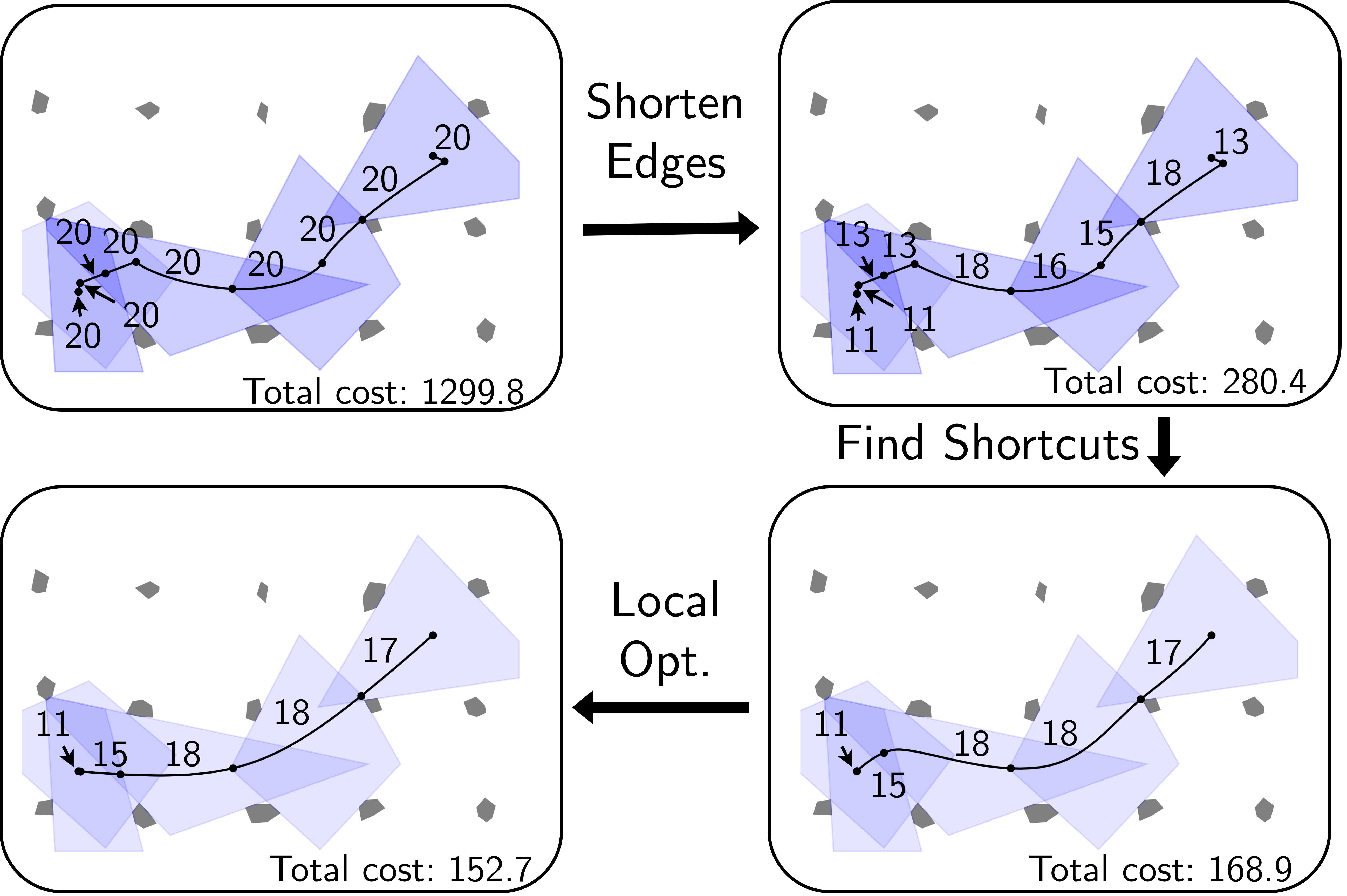}
\caption{Clockwise from top left: Initial feasible trajectory, trajectory after shortening edges, trajectory after shortcutting, trajectory after local optimization. Time allocation for each edge is labeled for all trajectories.}
\label{fig: local_opt}
\end{figure}

Recall the covariance steering problem with time-augmented control cost from \cref{prob: time_optimal_covariance_steering}. With this cost formulation, trajectory cost can be reduced by either reducing control cost while maintaining trajectory length, or by reducing trajectory length while keeping the control cost low enough to avoid increasing the total cost. PRISM's global planning module returns a path consisting of a sequence of edges, where each edge is associated with a time allocation and a sequence of safe convex sets over that time allocation. 

The second key idea of PRISM is that once a feasible global plan is found, online local optimization can reduce the cost by orders of magnitude. PRISM's local optimization algorithm alternates between three convex steps: reducing the time allocation for each edge, finding shortcuts to eliminate unnecessary edges, and local optimization for fixed time allocations for each edge. Given a feasible initial guess, each step is guaranteed to return a trajectory with cost no worse than that of the input trajectory. Therefore, our local optimization algorithm is guaranteed to return a solution (i.e. is complete) and is anytime (i.e. will always return a feasible solution if stopped early). Further, the final trajectory is guaranteed to be locally optimal for the fixed set of edges and time allocations utilized at the last optimization step. The algorithm is visualized in \cref{fig: local_opt}.

First, a preprocessing phase finds reference values for the state and control covariance for each edge in order to linearize the state and control chance constraints. Recall from \cref{sec: preliminaries} that a set of reference values such that the convex tightened constraints given in \cref{eq: convex_ctrl_chance_constraint}, \cref{eq: convex_nonpos_state_chance_constraint}, \cref{eq: convex_pos_chance_constraint} are satisfied must exist for any feasible belief space trajectory, so this step preserves feasibility and completeness of our local optimization algorithm. Although linearizing the constraints is lossy and is not necessary if e.g. SCP is used to solve \cref{prob: reformulated} directly instead of tightening to the SDP given in \cref{prob: steer_edge}, performing this constraint linearization transforms each step of the local optimization algorithm into a convex SDP, greatly improving computational efficiency.

The first step of our local optimization algorithm is to shorten each edge in the trajectory, checking if a lower-cost node-to-node maneuver exists with a shorter trajectory length by solving \cref{prob: steer_edge} for $N \in [N_c]$, where $N_c$ is equal to the current edge length, and selecting the trajectory length with the lowest cost. Next, given edges $e_i, e_j$ such that $j > i + 1$, we search for a shortcut by solving \cref{prob: steer_edge} with $\mu_1 = \mu_1^{(i)}, \mu_N = \mu_N^{(j)}$ for $N \in [N_{\max}]$, with the safe polytopes and constraint linearization for the first half of the trajectory taken from the first half of $e_i$ and the safe polytopes and constraint linearization for the second half of the trajectory taken from the second half of $e_j$.

After enumerating all potential shortcuts, we use uniform cost search over a graph consisting of the shortened edges and candidate shortcut edges (if any were found) to find the minimum-cost sequence of edges from the start node to the goal node. Finally, we concatenate the constraints and time allocations from each edge and solve \cref{prob: steer_edge} for a fixed final time and sequence of constraints, producing a locally optimal trajectory. Then, we convert the resulting trajectory back into a sequence of edges based on the time allocations and constraints for each edge prior to the local optimization step. If the sequence of edges has changed, we return to the edge shortening step and repeat; otherwise, the algorithm has converged and returns the locally optimal trajectory.
\section{Theoretical Results}\label{sec: theoretical_results}
PRISM constructs the mean steering graph offline (\cref{sec: approach}, \cref{alg: prism_offline_phase}) 
and then finds a feasible path online, as described in \cref{alg: prism_online}, before performing constraint linearization and local optimization as described in \cref{subsec: local_refinement}. 

Assuming $(A, B)$ is controllable and $\delta > 0$, in the absence of obstacles and other position constraints, PRISM can always find a path from the start distribution $\mathcal{I} := \mathcal{N}(\mu_\mathcal{I}, \Sigma_\mathcal{I})$ to the goal distribution $\mathcal{G} := \mathcal{N}(\mu_\mathcal{G}, \Sigma_\mathcal{G})$ in finite time if $\mu_\mathcal{I}$ and $\mu_\mathcal{G}$ are stationary and finite, $\Sigma_\mathcal{G} \succeq c_{1, \delta}^{(N)}I$, and $\Sigma_\mathcal{I} \preceq c_{2, \delta}^{(N)}I$.
\begin{theorem}\label{thm: completeness_without_obstacles}
(Completeness of PRISM without obstacles) Suppose there are no obstacles or position constraints, $\mu_\mathcal{I}$ and $\mu_\mathcal{G}$ are stationary with $||\mu_\mathcal{I} - \mu_\mathcal{G}||_2 \leq B$ for finite $B$, $\Sigma_\mathcal{G} \succeq c_{1, \delta}^{(N)}I$ and $\Sigma_\mathcal{I} \preceq c_{2, \delta}^{(N)}I$, \cref{asm: controllability} holds, and $\delta > 0$. PRISM will find a path from $\mathcal{I}$ to $\mathcal{G}$ in finite time.
\end{theorem}

In the presence of a finite number of obstacles or other position constraints, PRISM can always find a path in finite time if the above assumptions hold, $\mu_\mathcal{I}, \mu_\mathcal{G} \in \text{Int}(\mathcal{X}_p^- \setminus \mathcal{P}^+)$, \cref{asm: connectedness} holds, and $\Sigma_\mathcal{I} \preceq cI$ with $\Phi^{-1}(1-\epsilon)\sqrt{r^\star(c)} = d$, where $\mathbb{B}_p(P\mu_\mathcal{I}, d) \subset (\mathcal{X}_p \setminus \mathcal{P})$. 
\begin{theorem}\label{thm: completeness_with_obstacles}
(Completeness of PRISM in presence of obstacles)
Suppose that $\mu_\mathcal{I}$ and $\mu_\mathcal{G}$ are stationary and $\mu_\mathcal{I}, \mu_\mathcal{G} \in \text{Int}(\mathcal{X}_p^- \setminus \mathcal{P}^+)$, Assumptions \ref{asm: connectedness} and \ref{asm: controllability} hold, $\delta > 0$, and $\Sigma_\mathcal{G} \succeq c_{1, \delta}^{(N)}I$. Finally, suppose that $\Sigma_\mathcal{I} \preceq cI$ with $\Phi^{-1}(1-\epsilon)\sqrt{r^\star(c)} = d$, where $\mathbb{B}_p(P\mu_\mathcal{I}, d) \subset (\mathcal{X}_p \setminus \mathcal{P})$. Then, PRISM will find a path from $\mathcal{I}$ to $\mathcal{G}$ in finite time.
\end{theorem}

Further, PRISM is guaranteed to return a sound trajectory that obeys all mean dynamics constraints, covariance dynamics constraints, obstacle chance constraints, state chance constraints, and control chance constraints.
\begin{theorem}
\label{thm: soundness}
(Soundness of PRISM) Any path found by PRISM is sound, satisfying \cref{def: soundness}.
\end{theorem}
These theorems are supported by several lemmas:
\begin{lemma}\label{lemma: graph_connectivity}
(Connectedness of the mean space graph) Suppose \cref{asm: controllability} and \cref{asm: connectedness} hold and $\delta > 0$. Then, \cref{alg: prism_offline_phase} will return a mean space graph with vertices $\mathcal{V}$ and edges $\mathcal{E}$ such that the mean space graph is strongly connected; namely, $\forall v_i, v_j \in \mathcal{V}$, there exists a sequence of directed edges in $\mathcal{E}$ connecting $v_i$ to $v_j$, and a second sequence of directed edges connecting $v_j$ to $v_i$.
\end{lemma}
\begin{lemma}\label{lemma: mean_soundness}
(Soundness of mean space planning) Suppose that $\mu_\mathcal{I}$ is stationary and $\Sigma_\mathcal{I} = c_{1, \delta}^{(N)} I$ and that $\mathcal{J} = \mathcal{N}(\mu_\mathcal{J}, \Sigma_\mathcal{J})$ has stationary mean $\mu_\mathcal{J} = A\mu_\mathcal{J}$ and covariance $\Sigma_\mathcal{J} = c_{1, \delta}^{(N)}I$. Suppose that $\mu_\mathcal{I} \in \text{Int}(\mathcal{S}_\mathcal{I})$ and $\mu_\mathcal{J} \in \text{Int}(\mathcal{S}_\mathcal{J})$ with $\mathcal{S}_\mathcal{I}, \mathcal{S}_\mathcal{J} \subseteq ( \mathcal{X}_p^- \setminus \mathcal{P}^+)$ and $\text{Int}(\mathcal{S}_\mathcal{I} \cap \mathcal{S}_\mathcal{J}) \neq \emptyset$, that Assumptions \ref{asm: connectedness} and \ref{asm: controllability} hold, and that $\delta > 0$.
Then, \cref{alg: prism_mean_steering} constructs an edge from $\mathcal{I}$ to $\mathcal{J}$ that satisfies the soundness conditions in \cref{def: soundness}.
\end{lemma}
\begin{lemma}\label{lemma: recursive_feasibility}
(Recursive feasibility of local optimization) Local optimization preserves soundness at every intermediate step and is guaranteed to return a solution. If the input trajectory $\mathcal{T}$ satisfies the soundness conditions  in \cref{def: soundness}, at every step of local optimization, the local optimization module will yield an output trajectory $\mathcal{T}$ that is also sound.
\end{lemma}

\section{Experiments}\label{sec: experiments}
\begin{figure}[bpt]
\centering
\includegraphics[width=0.1575\textwidth]{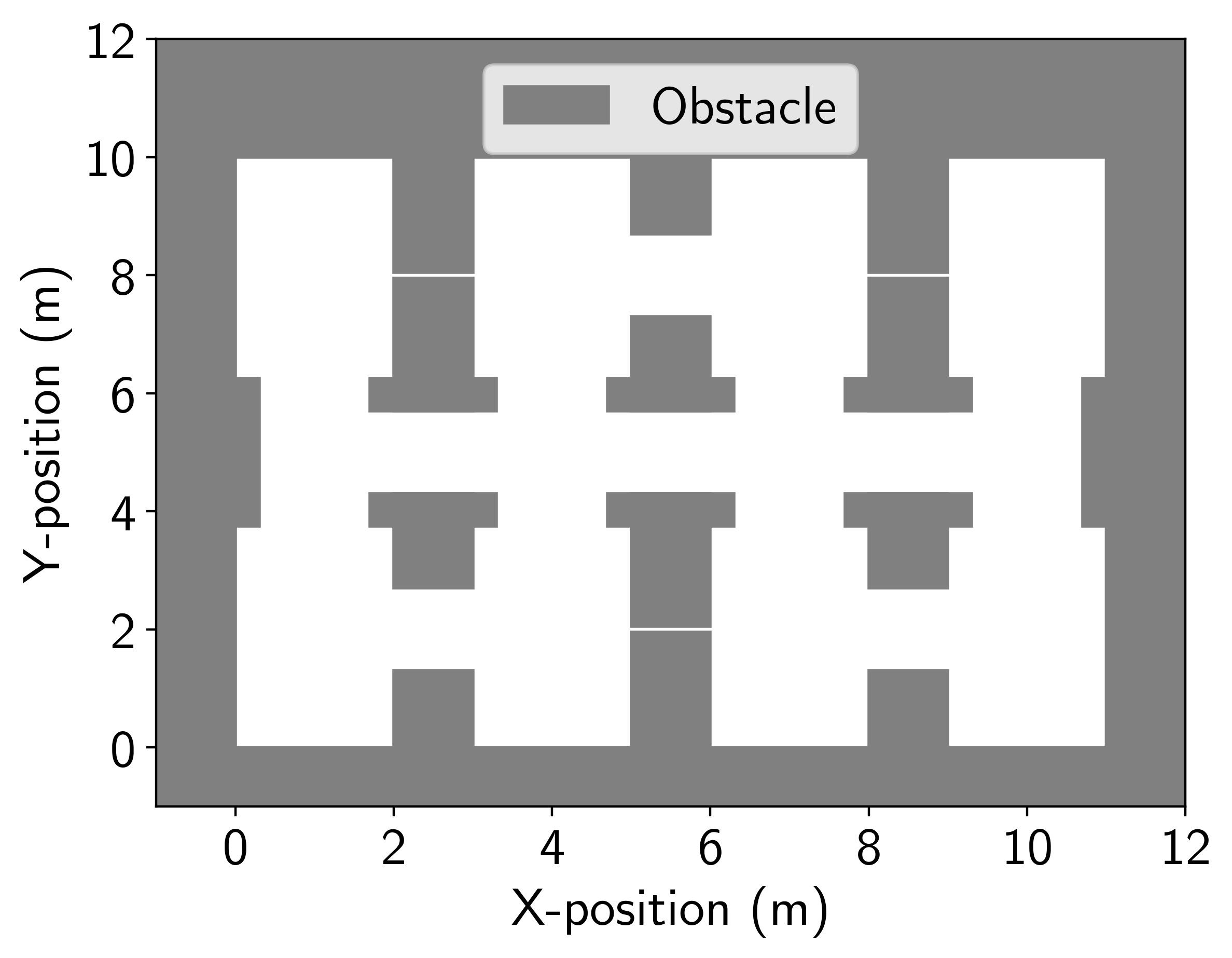}
\includegraphics[width=0.1575\textwidth]{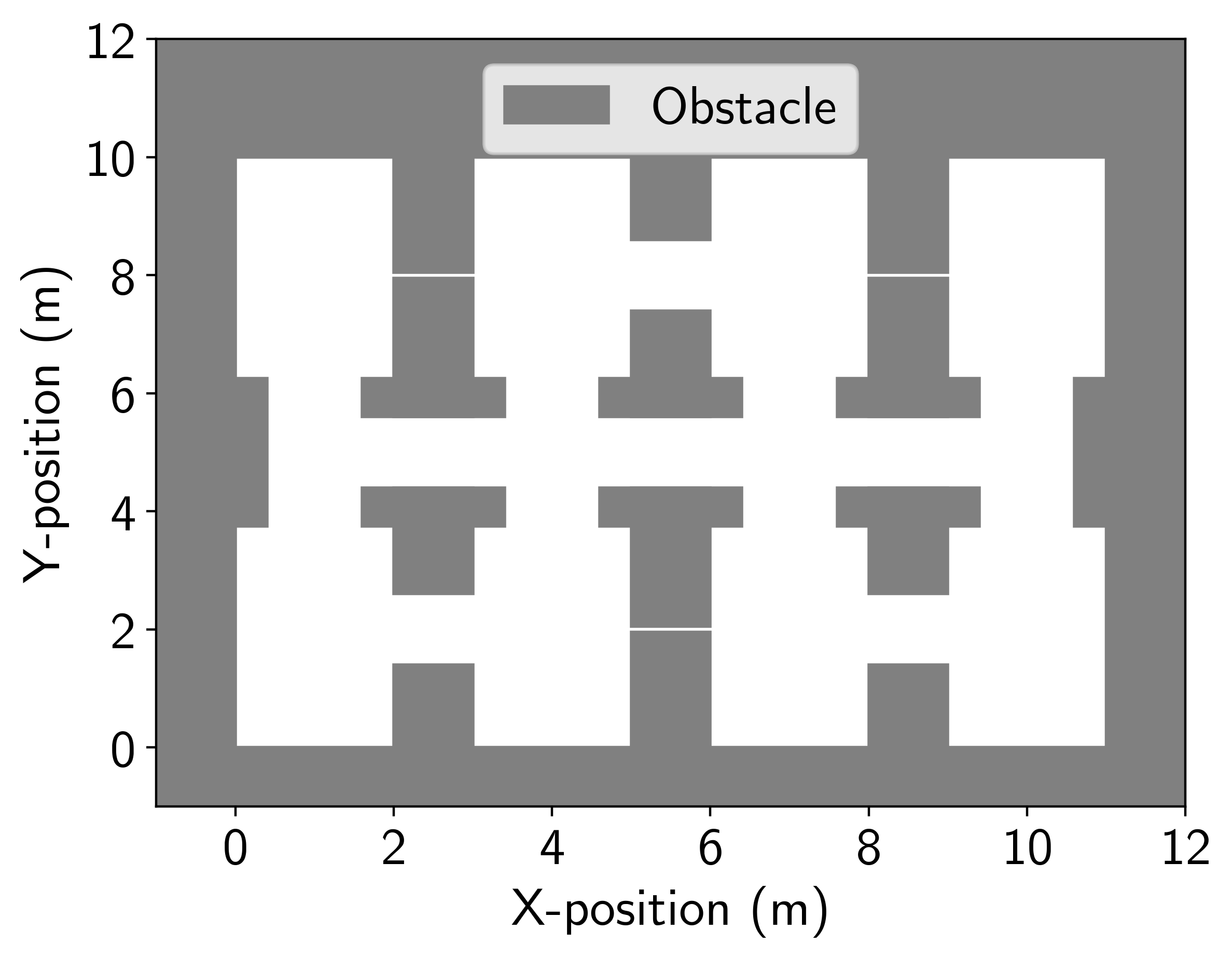}
\includegraphics[width=0.1575\textwidth]{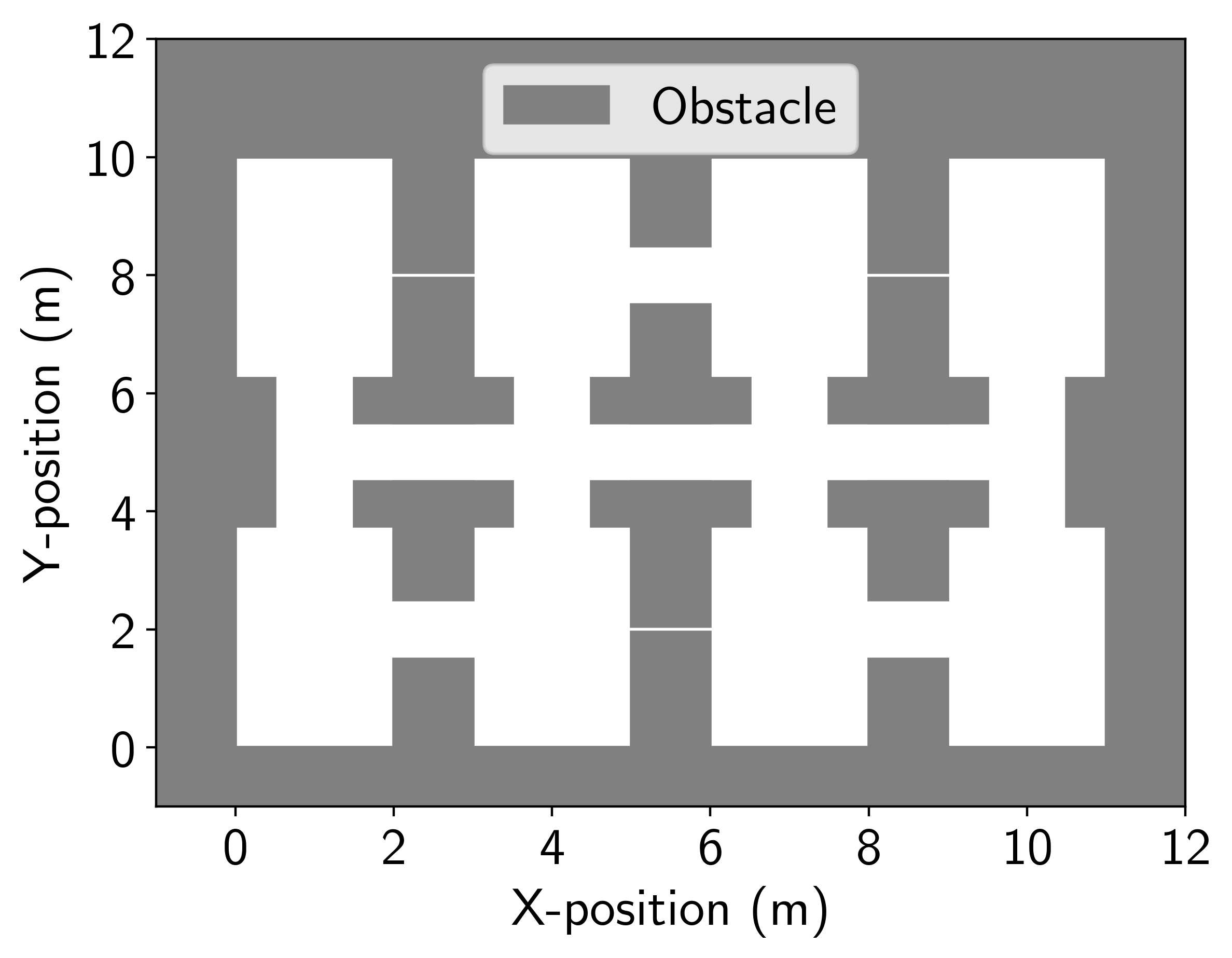}
\includegraphics[width=0.1575\textwidth]{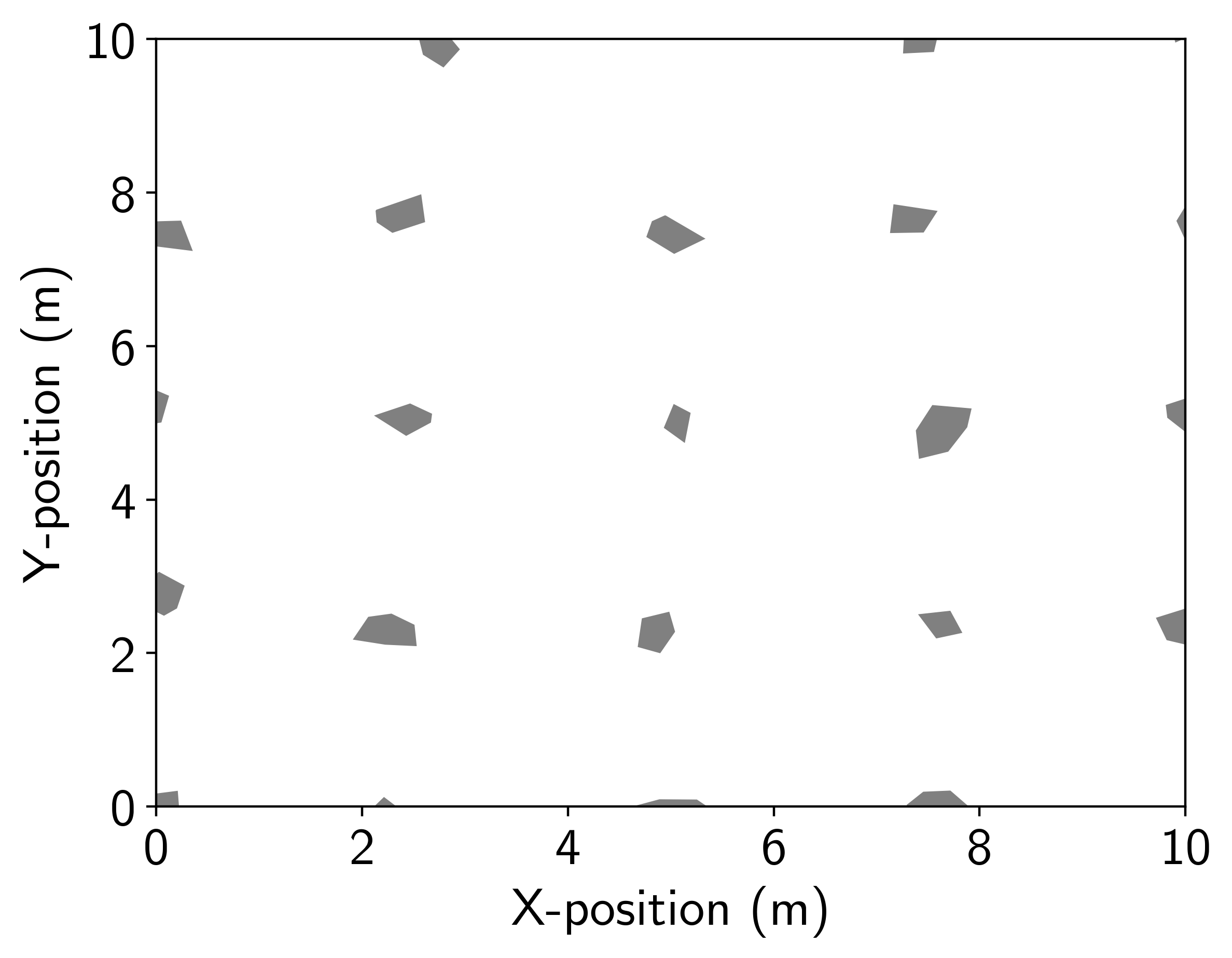}
\includegraphics[width=0.1575\textwidth]{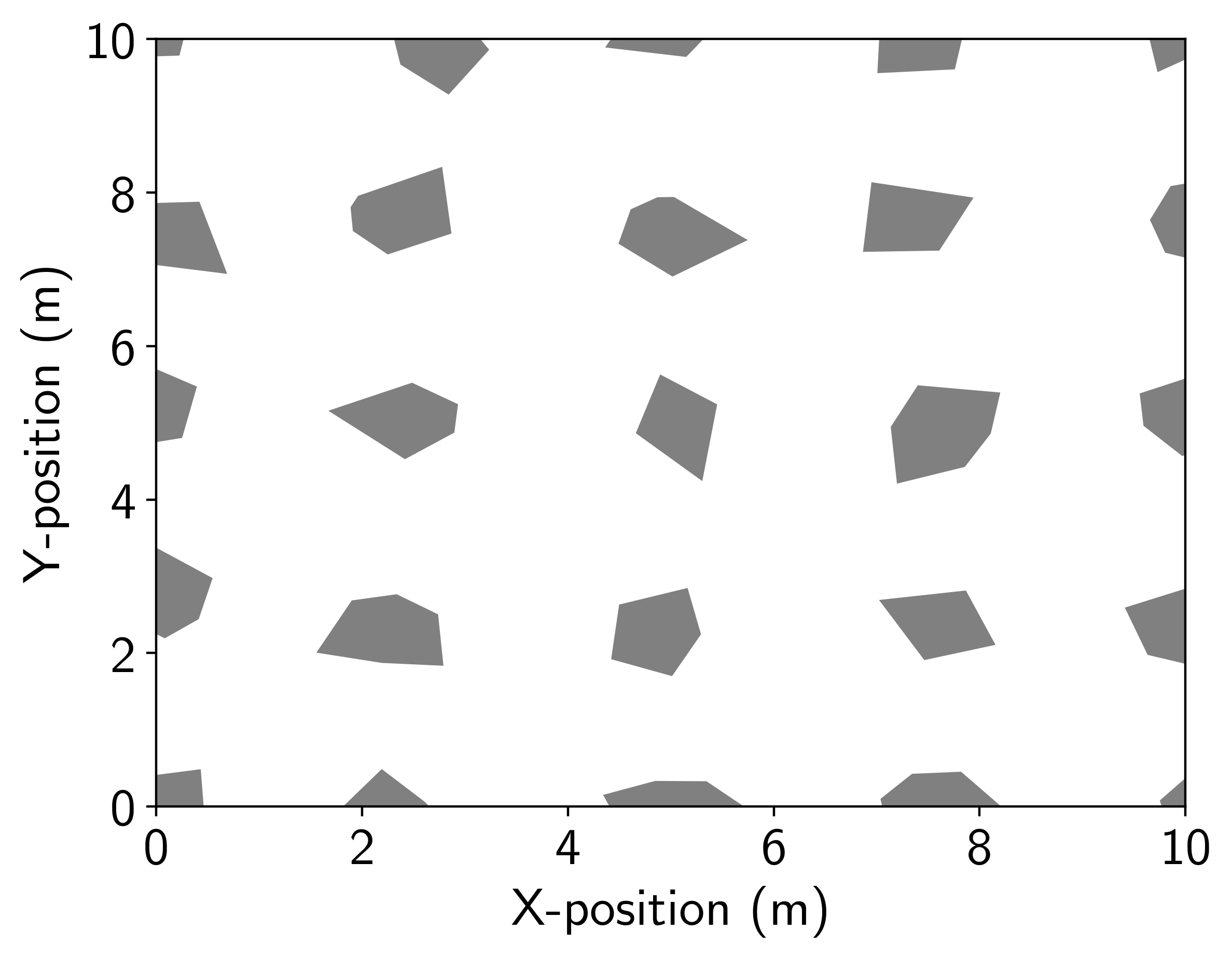}
\includegraphics[width=0.1575\textwidth]{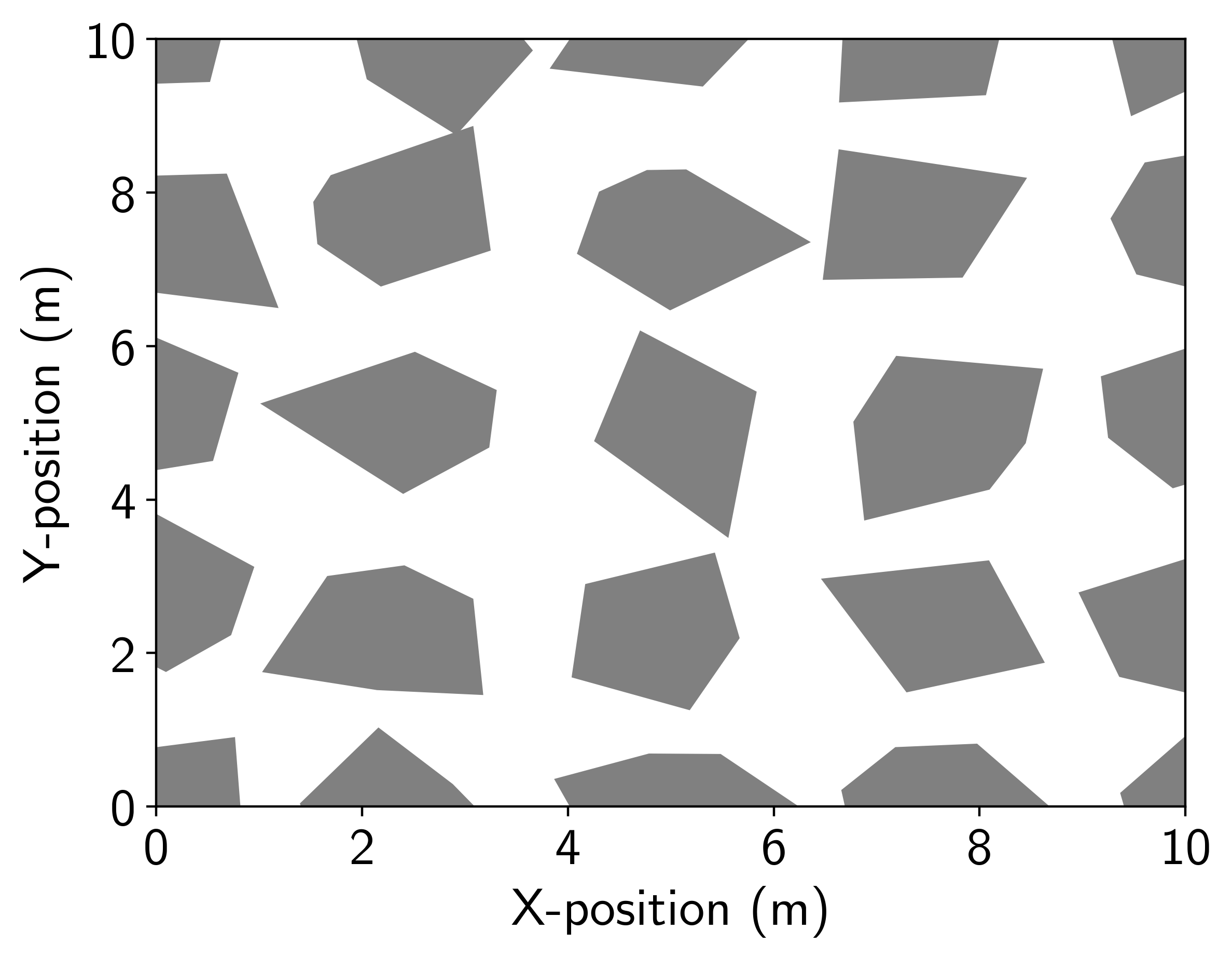}
\caption{Top: Confined office environments with varying corridor width. Bottom: $\mathcal{P}^+$ for a cluttered environment with varying levels of actuator noise.}
\label{fig:varying_office_envs}
\end{figure}

\begin{figure}[tbp]
\centering 
\includegraphics[width=0.4875\columnwidth]{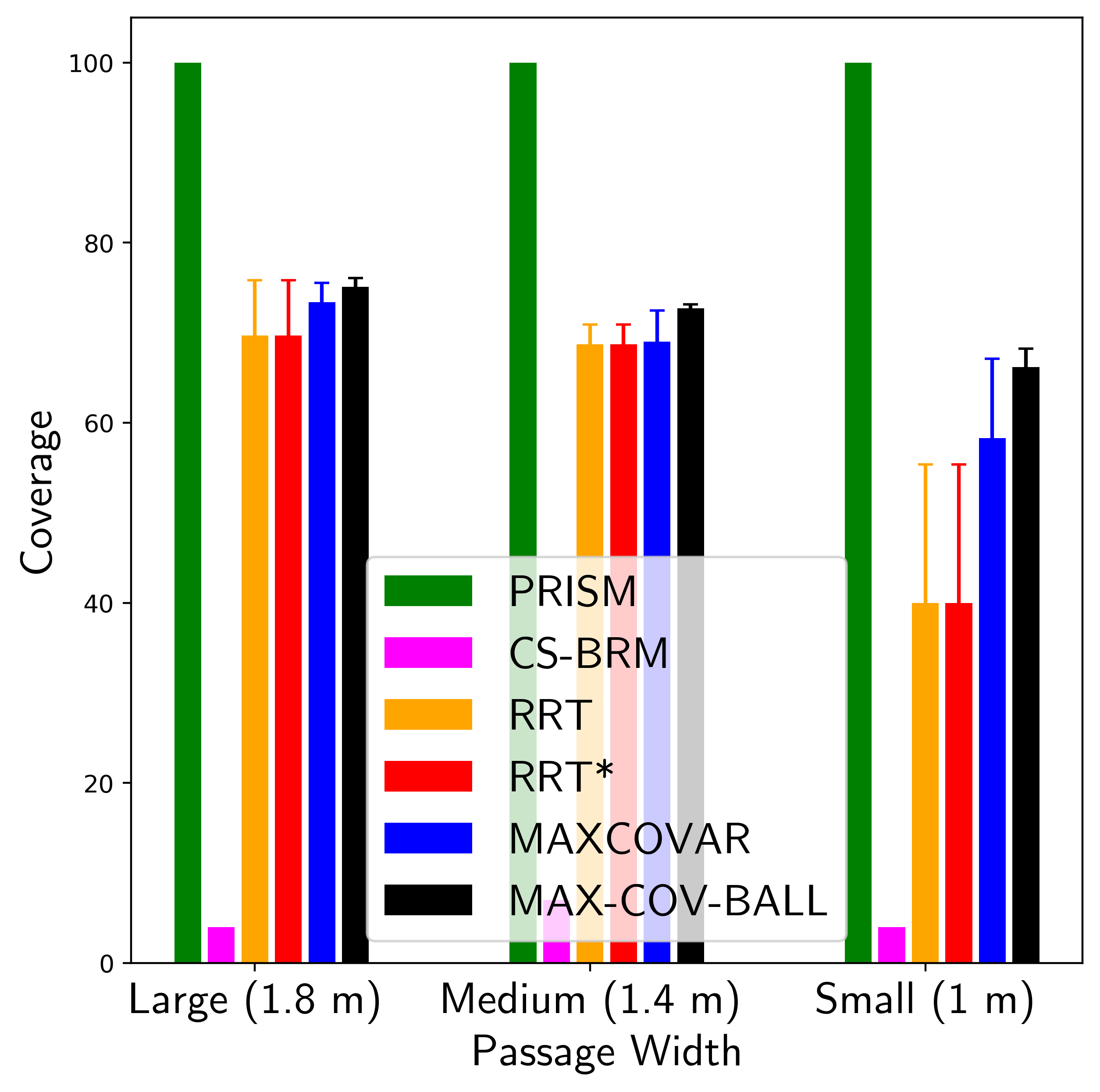}
\includegraphics[width=0.4875\columnwidth]{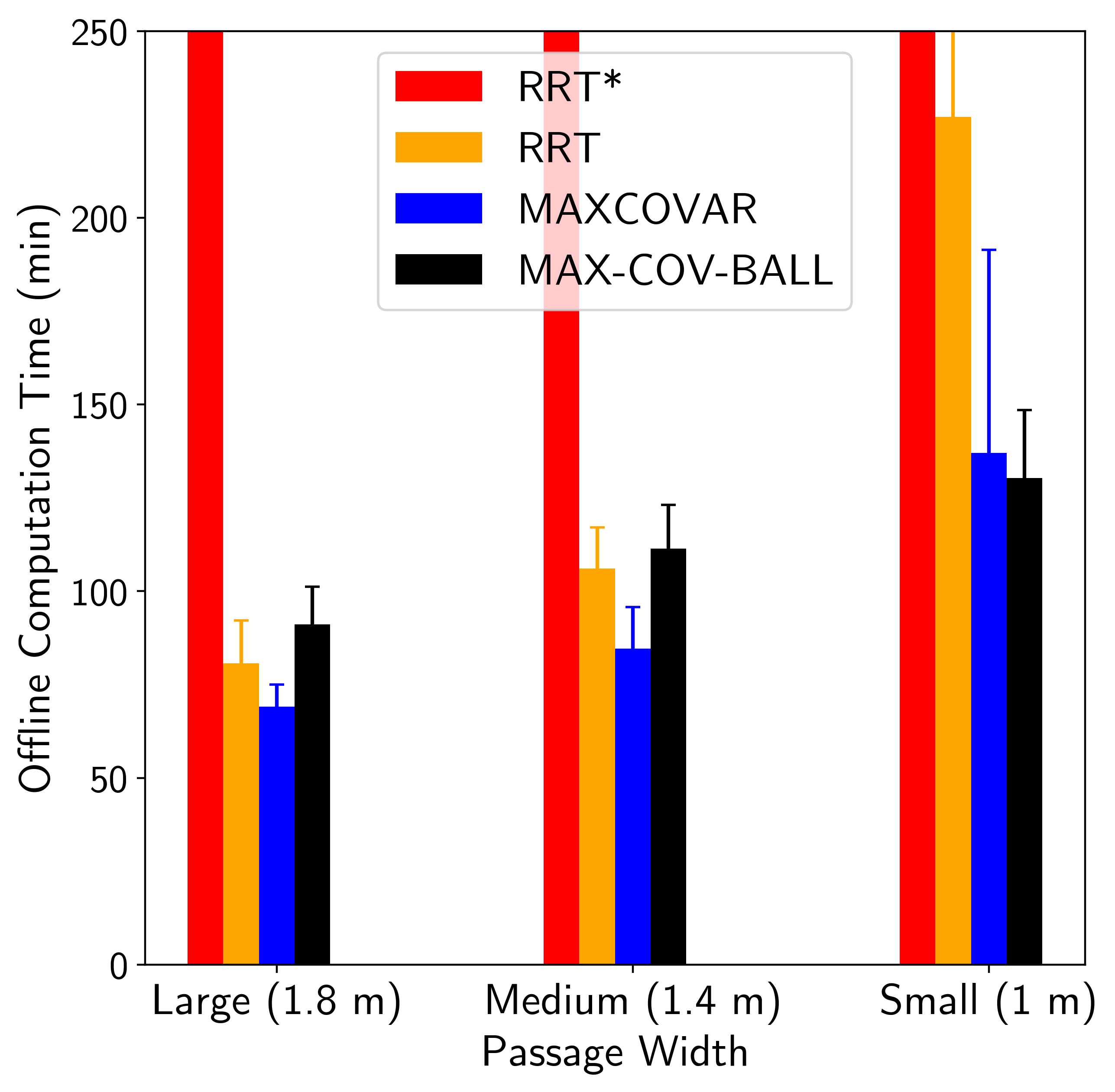}
\caption{Left: Coverage in the office environment across varying passage widths, Right: Offline computation time for sampling-based methods in the office environment. RRT$^\star$ takes $496.7 \pm 28.0$ min with the large passage, $575.8 \pm 56.5$ min with the medium passage, and $680.0 \pm 148.7$ min with the small passage.}
\label{fig: office_comp_time}
\end{figure}

\begin{figure}[tbp]
\centering
\includegraphics[width=0.4875\columnwidth]{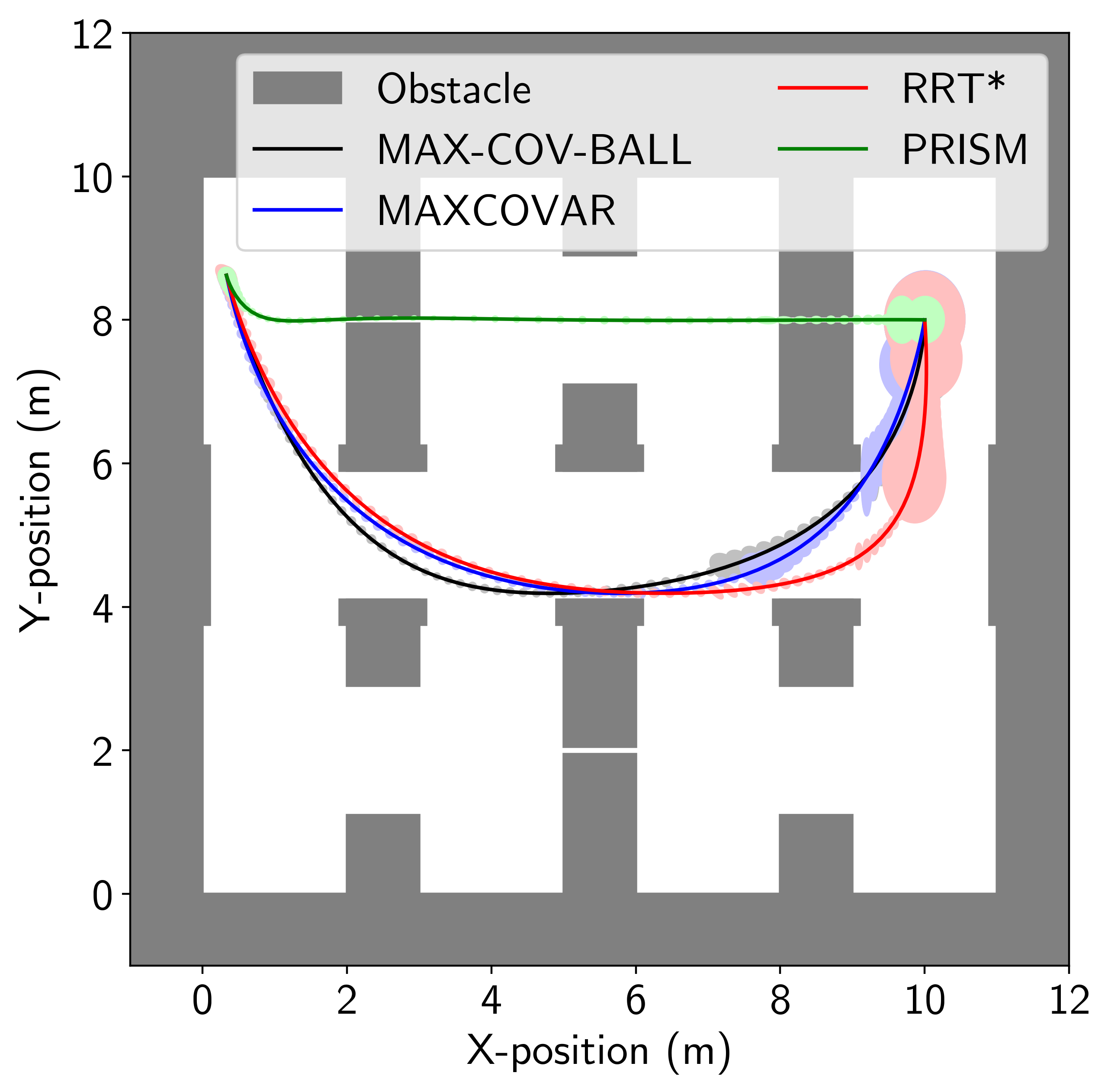}
\includegraphics[width=0.4875\columnwidth]{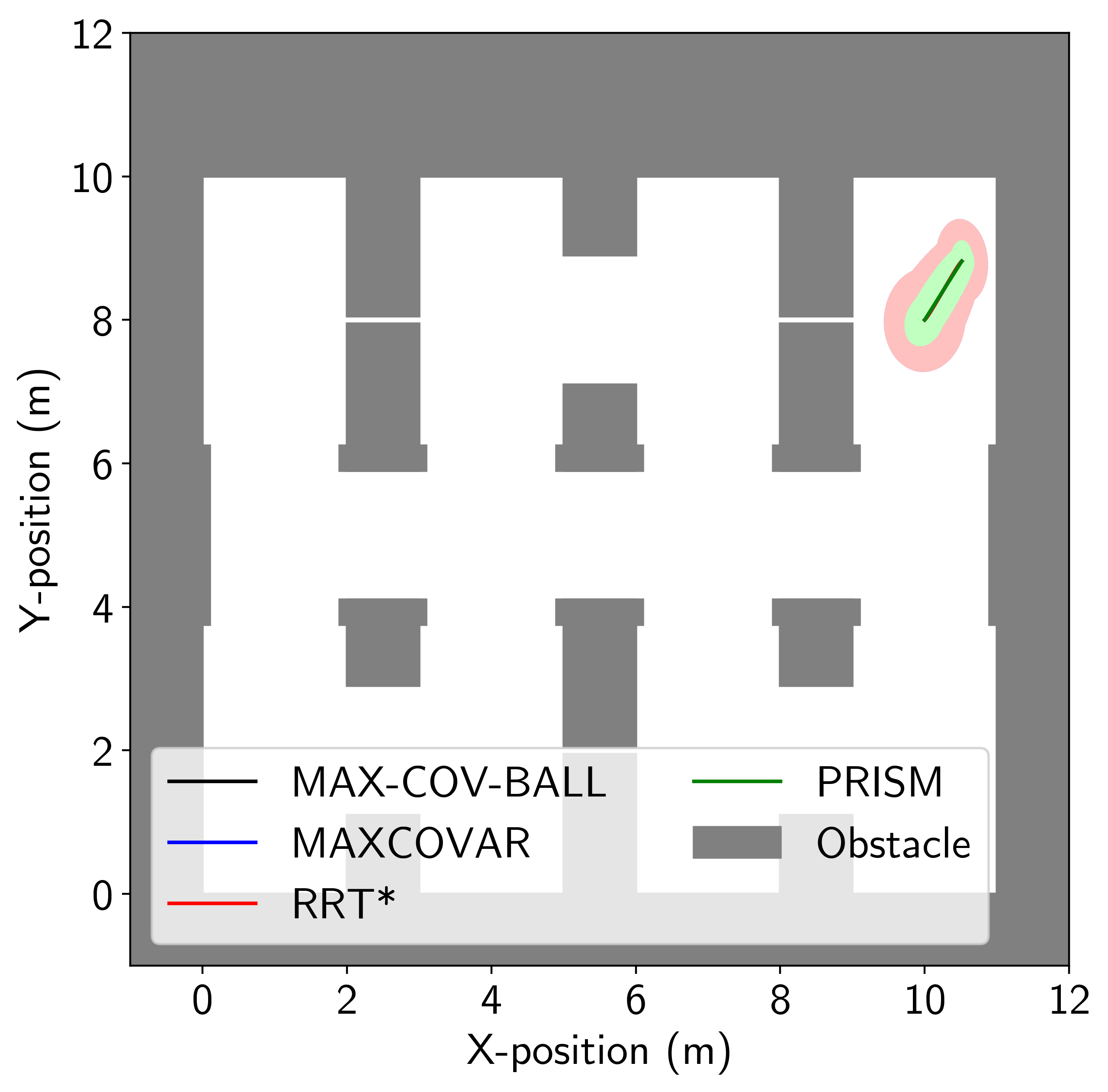}
\includegraphics[width=0.4875\columnwidth]{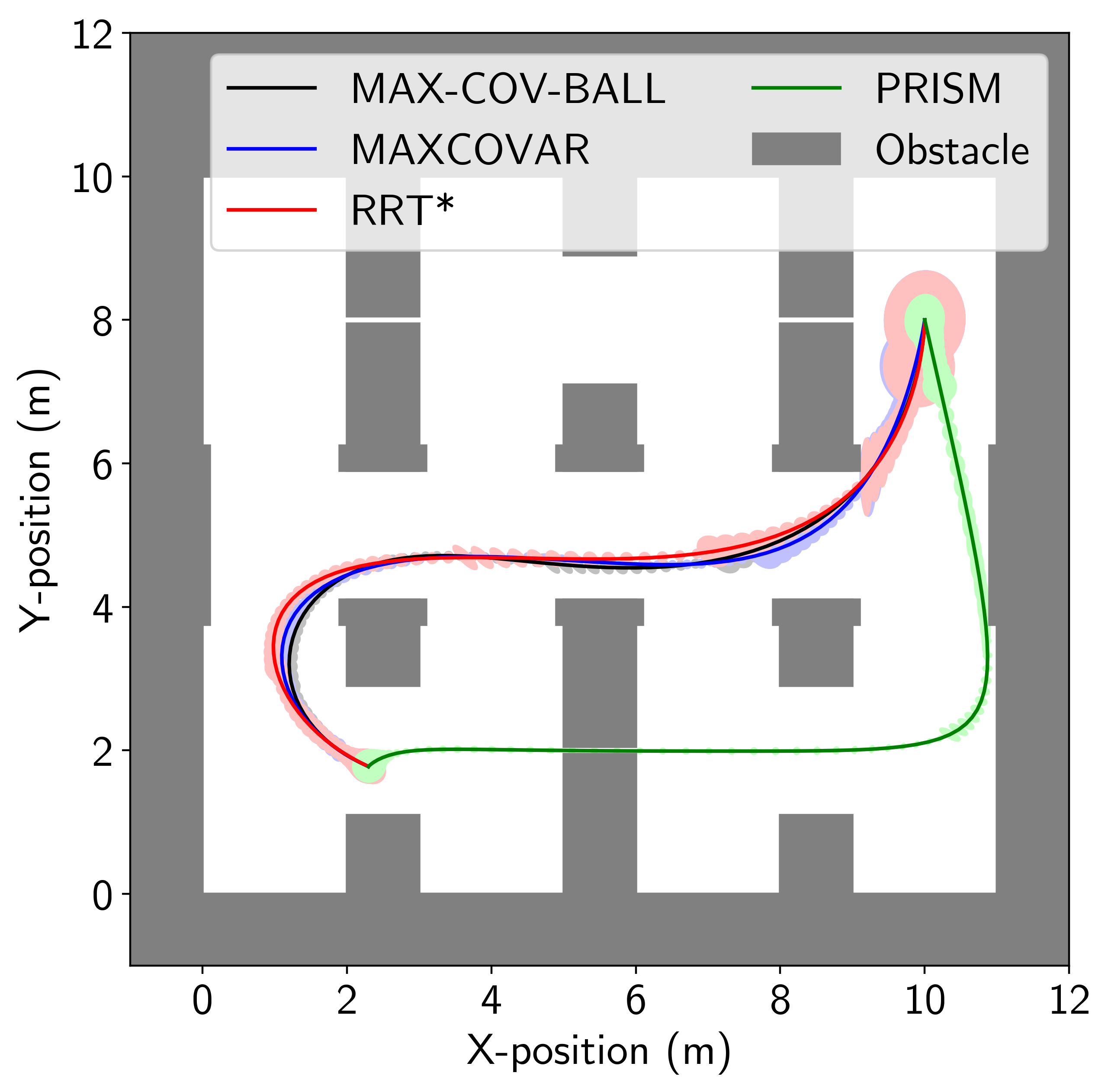}
\includegraphics[width=0.4875\columnwidth]{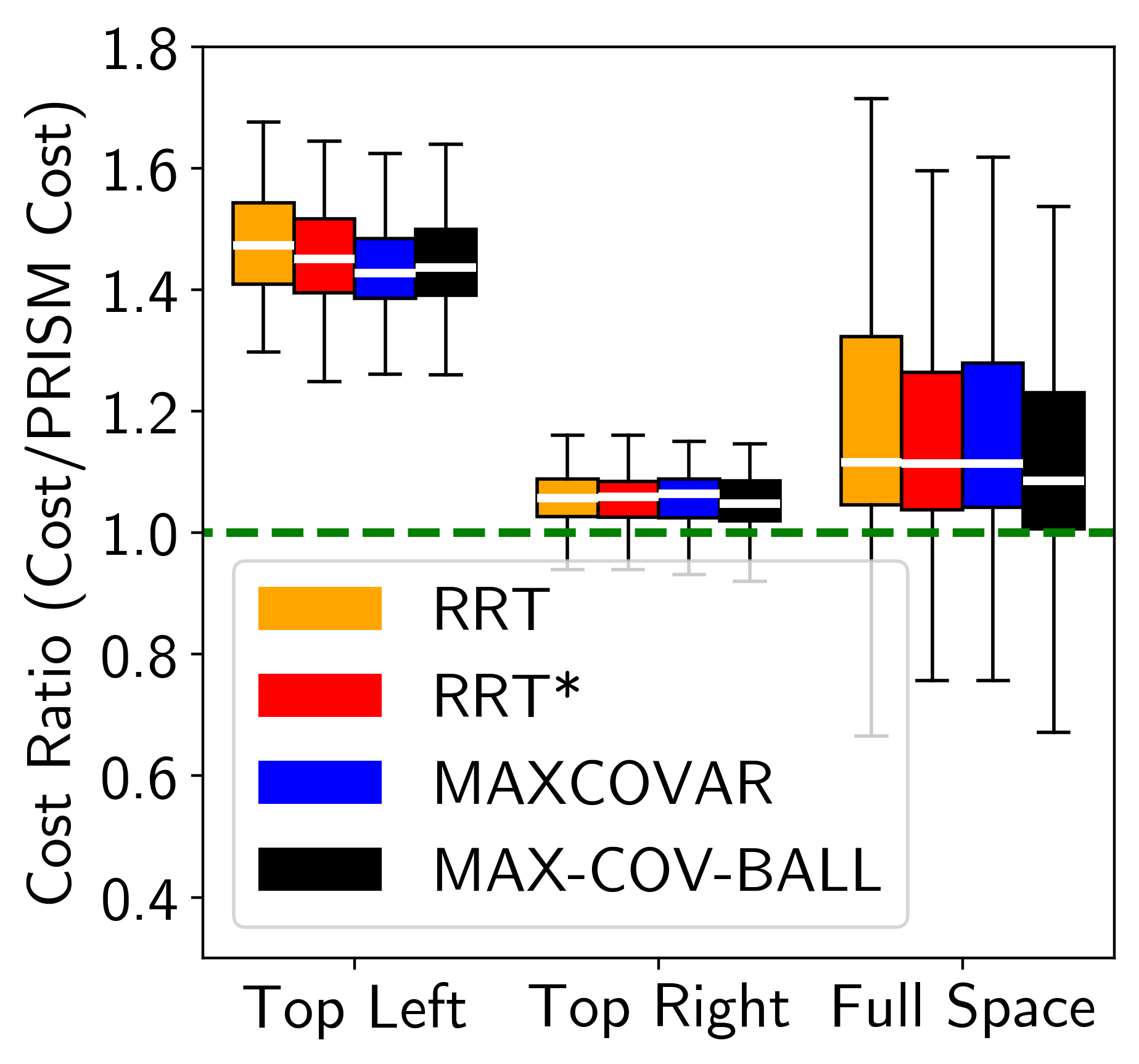}
\caption{Top left: Locally optimized trajectories starting in the top left (hard set); Top right: locally optimized trajectories starting in the top right (easy set), Bottom left: locally optimized trajectories starting in the bottom left (full space set); Bottom right: cost ratio (baseline cost/PRISM cost) across all three evaluation sets.}
\label{fig: cost_office}
\end{figure}

\begin{figure*}[t]
\centering
\includegraphics[width=0.325\textwidth]{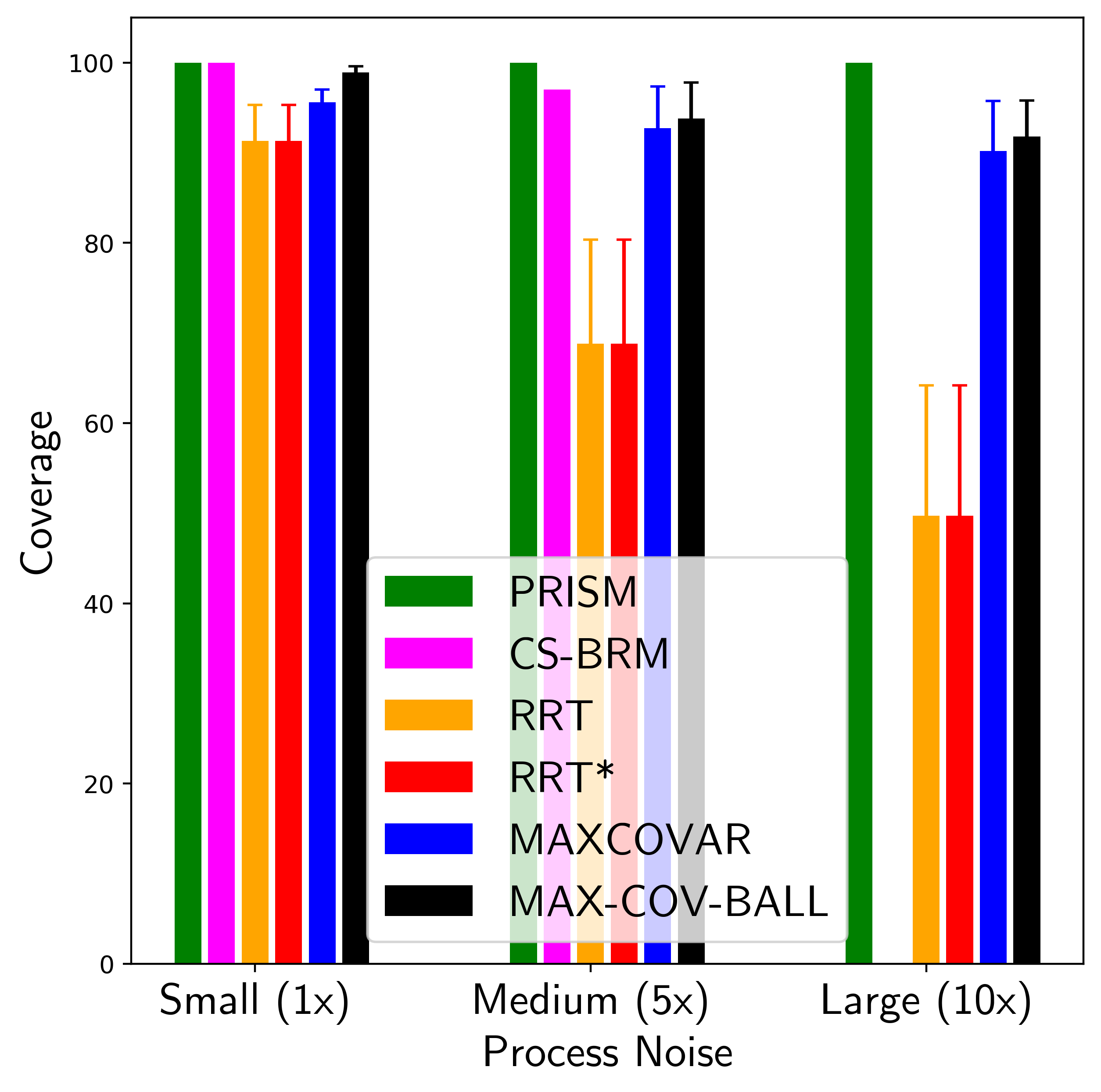}
\includegraphics[width=0.325\textwidth]{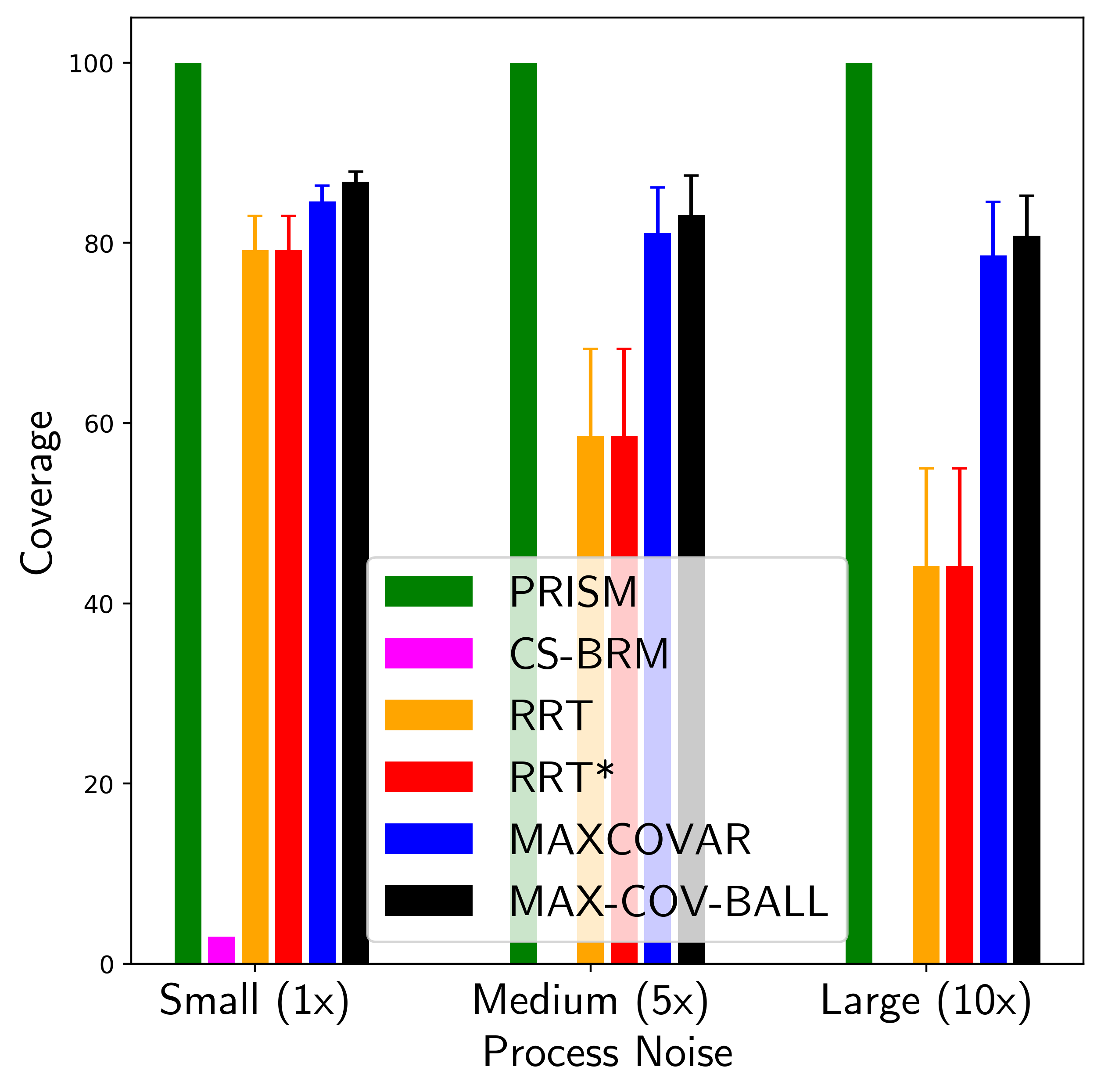}
\includegraphics[width=0.325\textwidth]{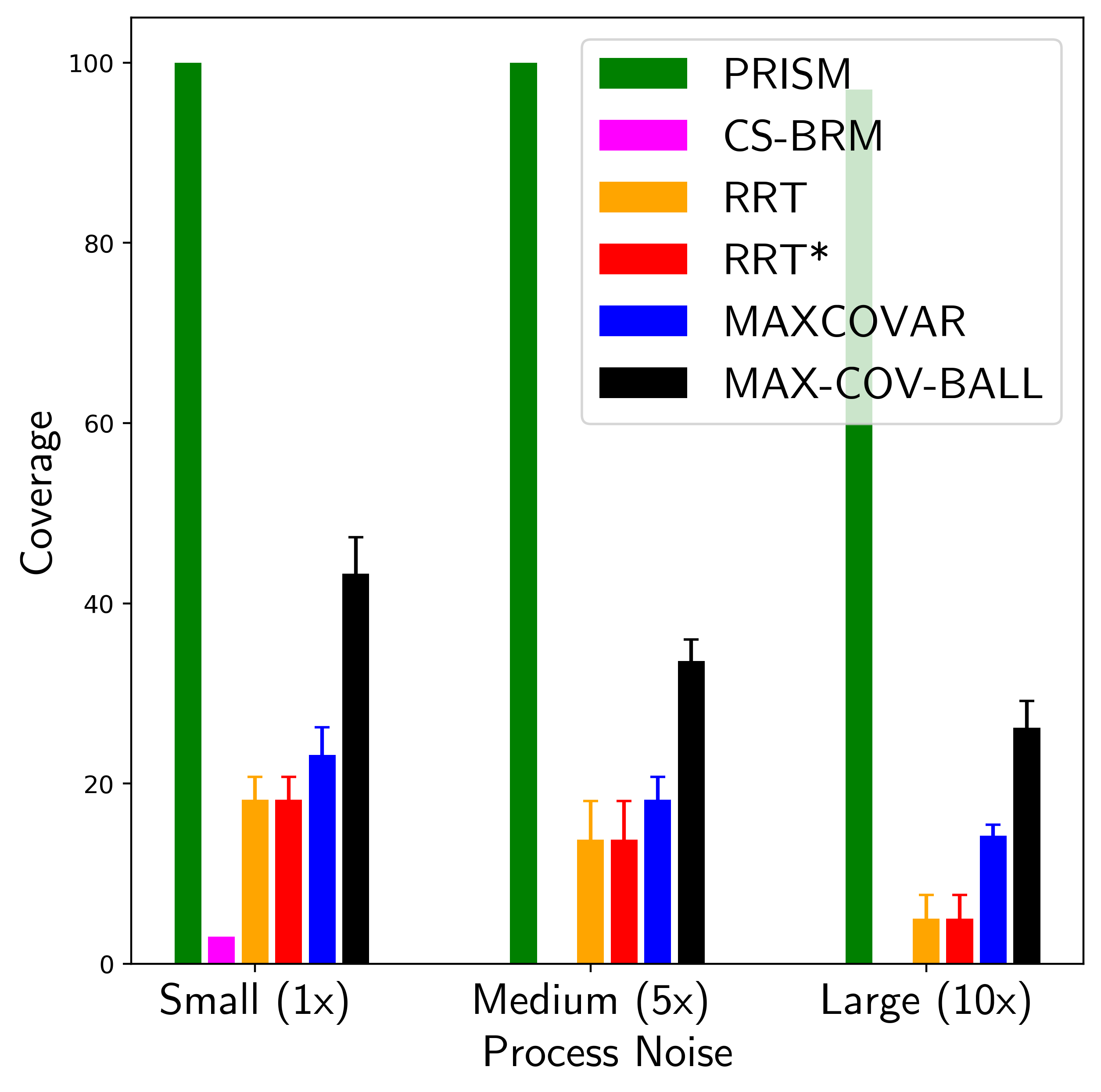}
\caption{Coverage in the cluttered environment. L-R: easy evaluation set - initial distributions satisfy \cref{thm: completeness_with_obstacles} with $\Sigma = 0.1cI$,  medium-difficulty evaluation set - initial distributions satisfy \cref{thm: completeness_with_obstacles} with $\Sigma = 0.5cI$,  hard evaluation set - initial distributions have non-stationary means.}
\label{fig: stationary_coverage_results}
\end{figure*}

\begin{figure}[tbp]
\centering 
\includegraphics[width=0.4875\columnwidth]{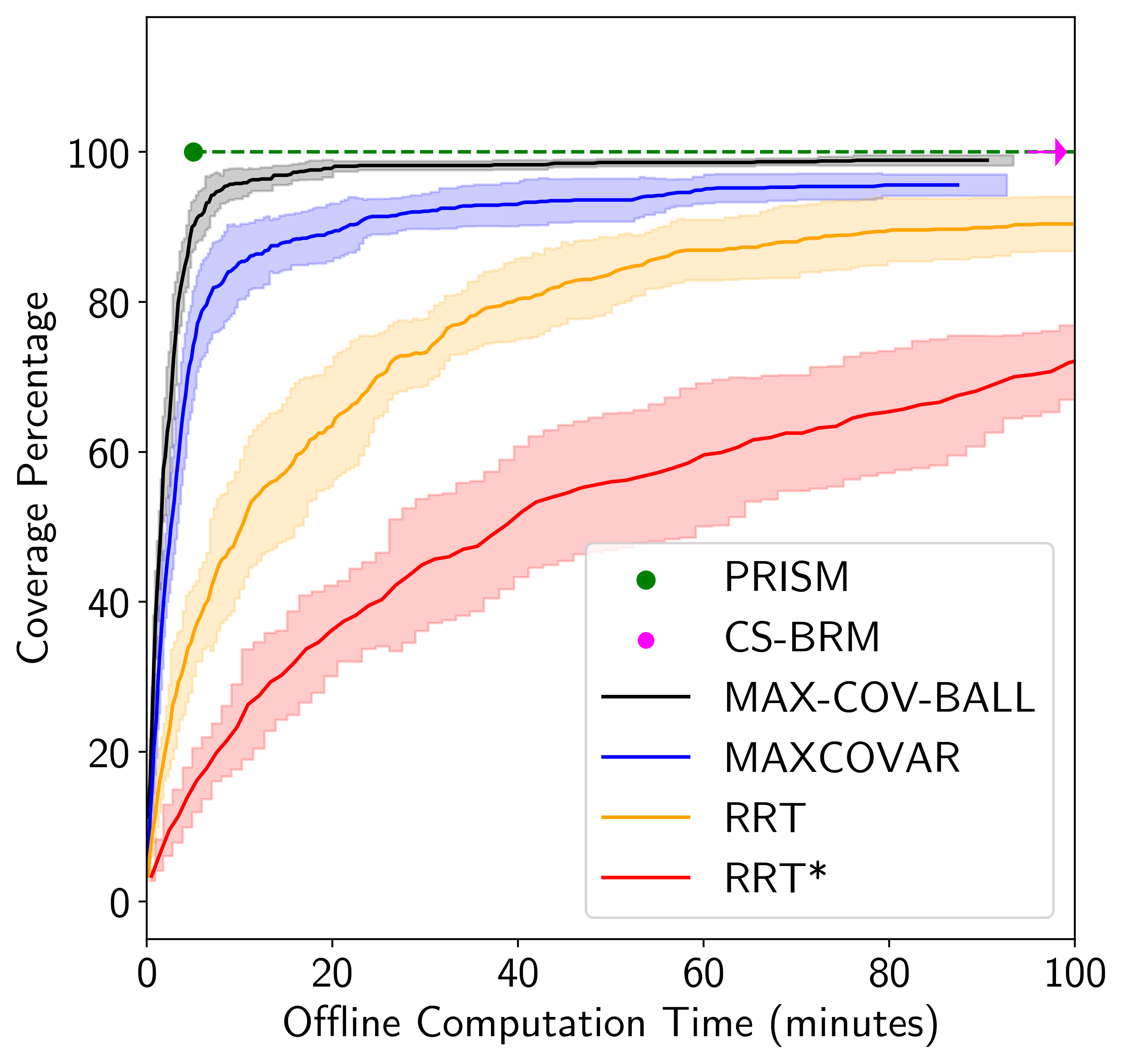}
\includegraphics[width=0.4875\columnwidth]{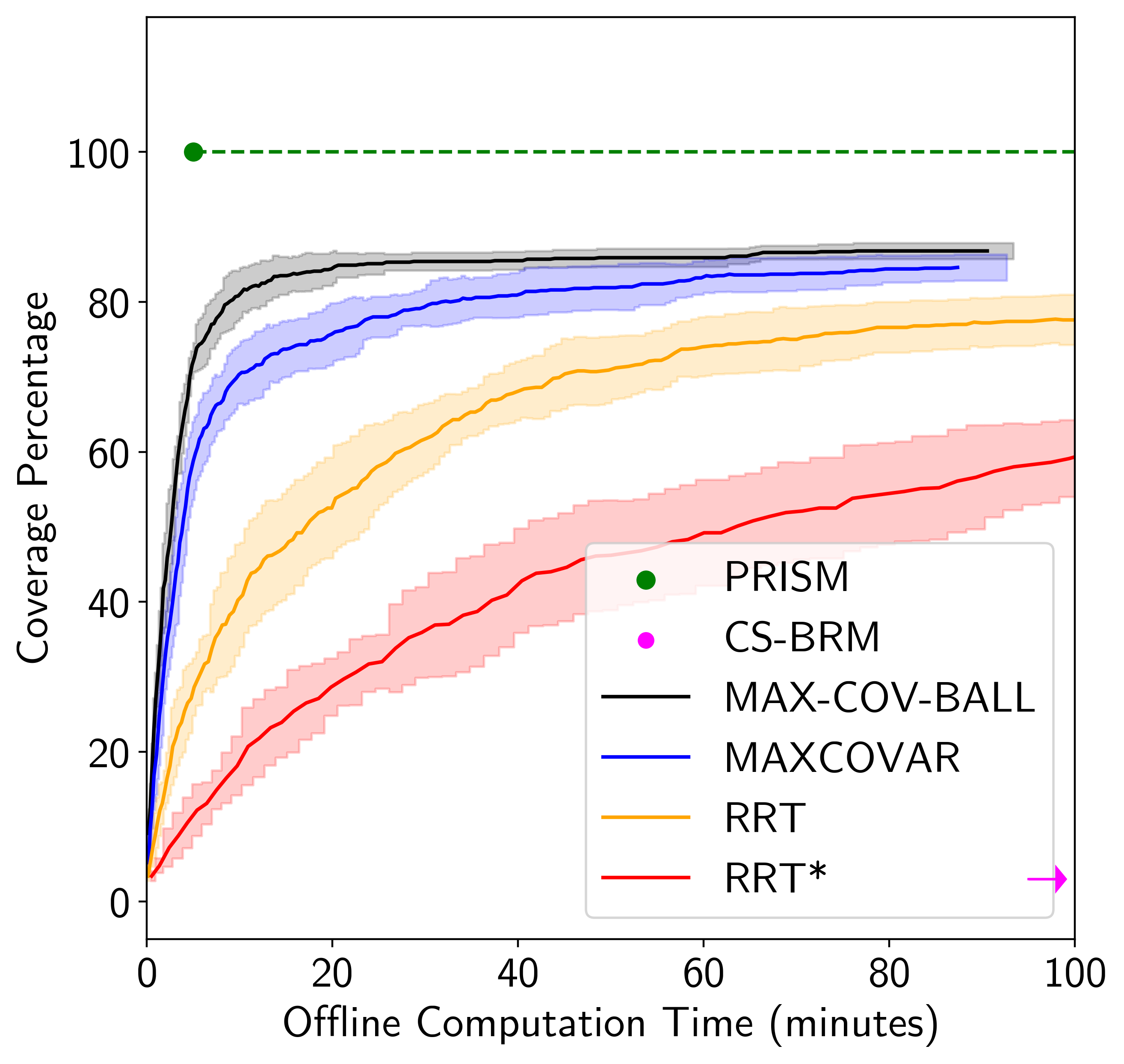}
\caption{Coverage plotted against offline computation time. Left: Easy evaluation set with small actuator noise. Right: Medium-difficulty evaluation set with small actuator noise.}
\label{fig: coverage_per_iteration}
\end{figure}

\begin{figure}[tbp]
\centering
\includegraphics[width=0.4875\columnwidth]{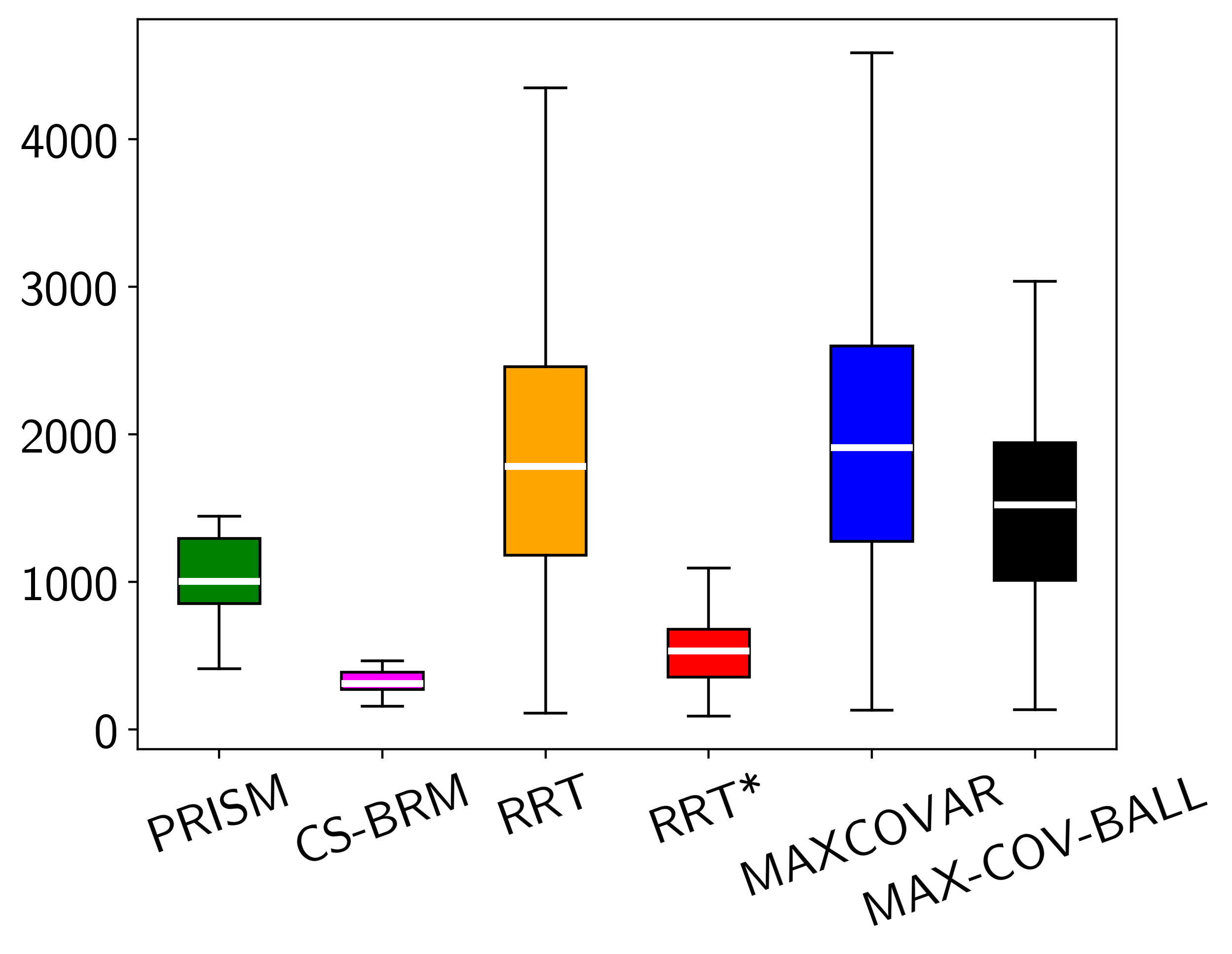}
\includegraphics[width=0.4875\columnwidth]{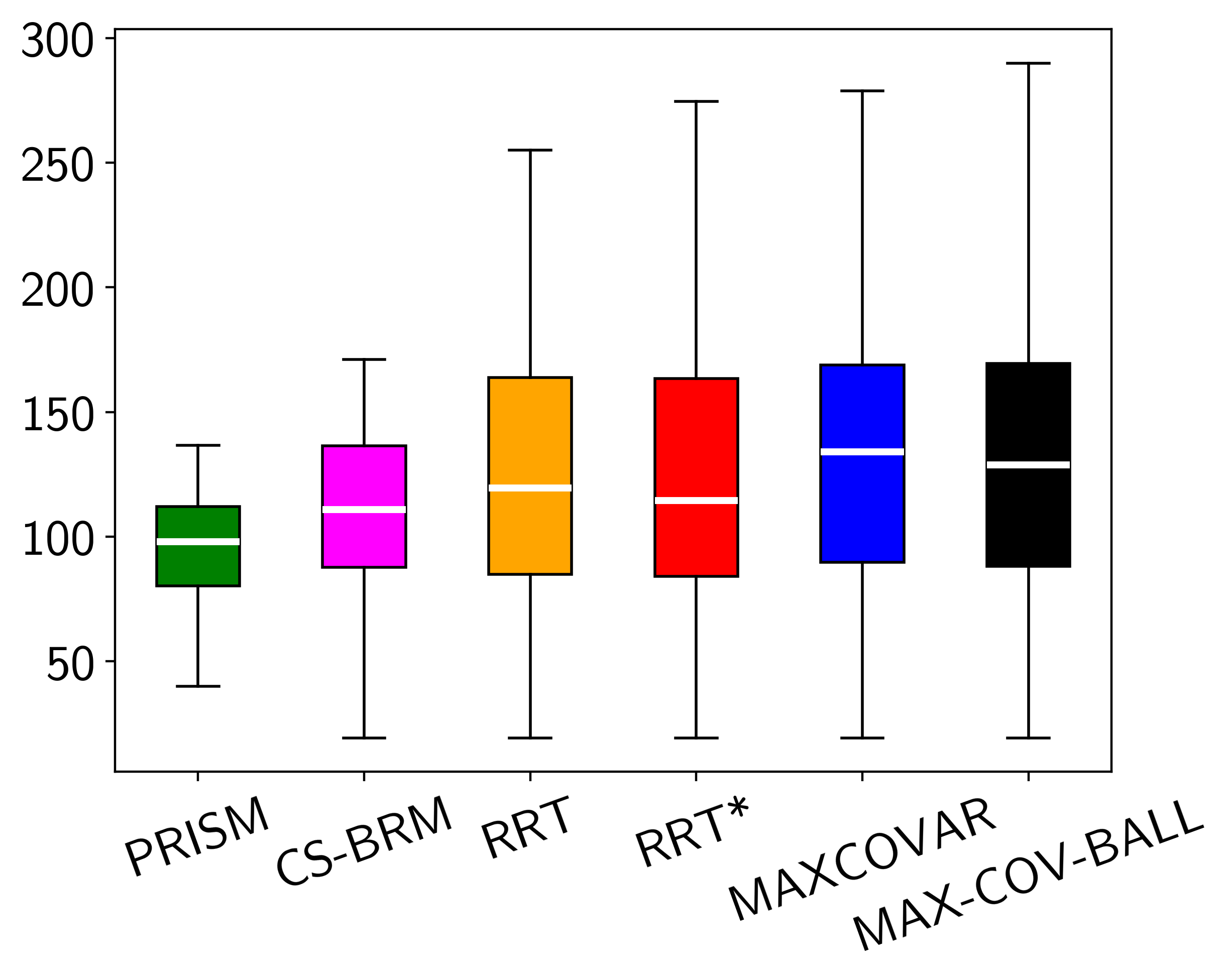}
\vspace{-0.25cm}
\caption{Cost before (left) and after (right) local optimization for PRISM, CS-BRM, RRT, RRT$^\star$, MAXCOVAR, and MAX-COV-BALL in the cluttered environment.}
\label{fig: cost_scatter}
\end{figure}
\begin{figure}[tbp]
\centering 
\includegraphics[width=0.4875\columnwidth]{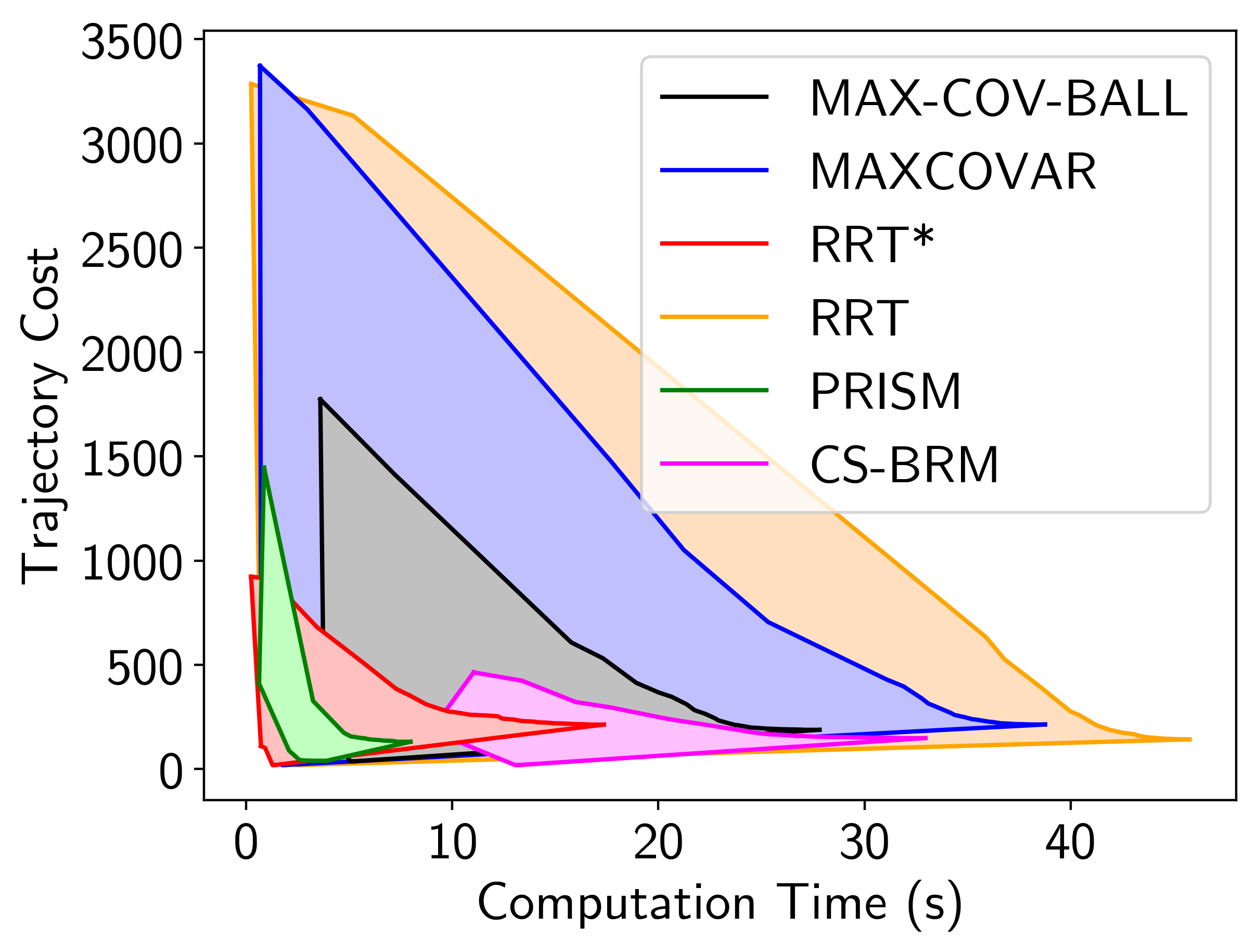}
\includegraphics[width=0.4975\columnwidth]{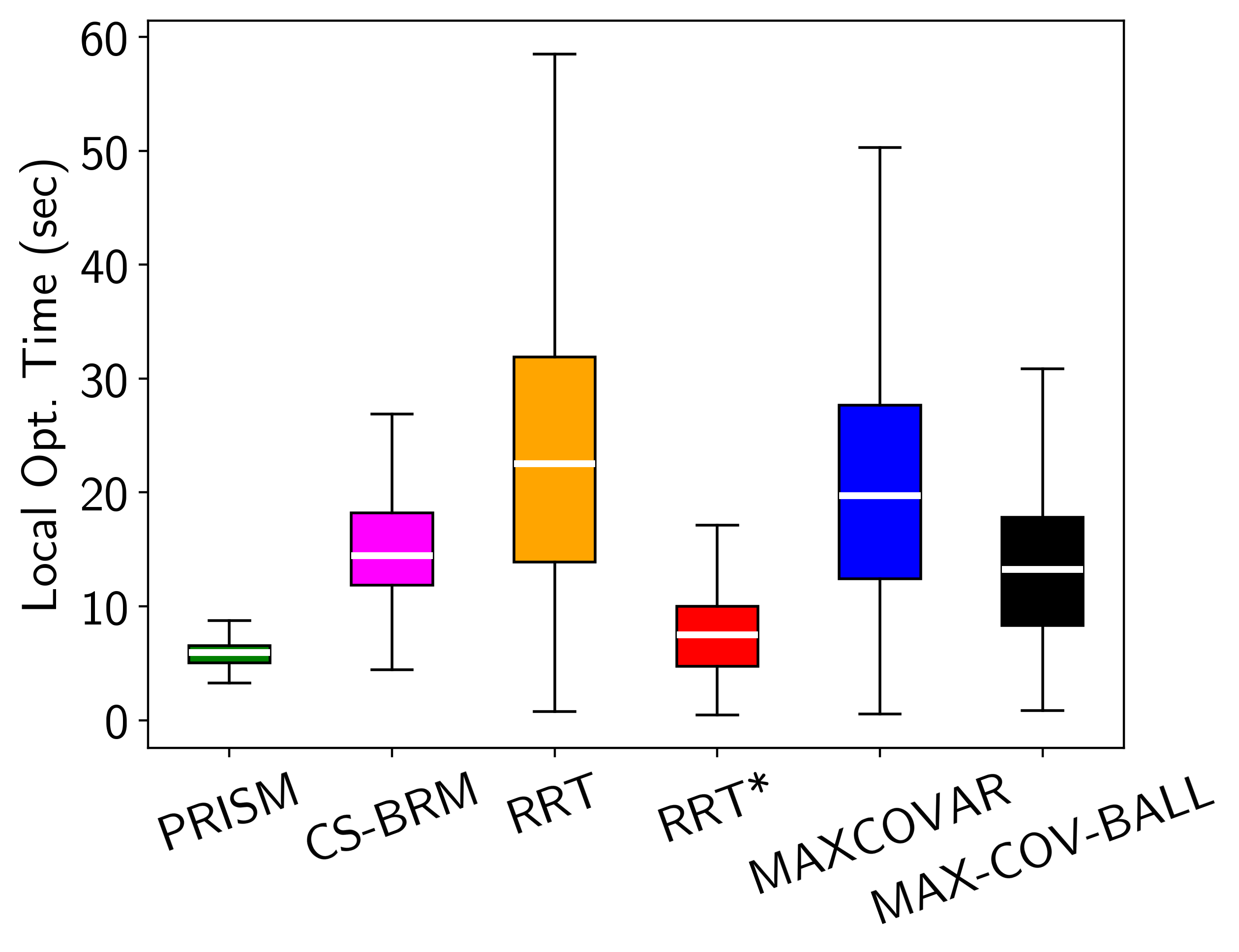}
\caption{Left: Cost vs.\ online computation time during local optimization, with the area shaded between the cost curves for the trajectories with lowest and highest initial feasible cost for each method. Right: Local optimization time.}
\label{fig: pareto_front}
\end{figure}

We evaluate PRISM by conducting motion planning experiments for a 6DoF quadrotor in a 2D plane. The quadrotor state and control matrices are modeled as in  \cite{aggarwal2025tac, rose2025efficient}, with the actuator noise matrix $D$ given as in \cite{rose2025efficient}
such that the actuator noise matrix $D$ integrates the continuous actuator noise over each discrete time step. We use $\Delta t = 0.2$ sec and a planning horizon of $N = 20$ time steps. We restrict the safe control set $\mathcal{U}$ to $[-5, 5] \times [-5, 5]$ and enforce all chance constraints with a maximum failure probability $\epsilon = 0.01$.

We evaluate PRISM against MAXCOVAR \cite{aggarwal2024sdp}, MAX-COV-BALL \cite{rose2025efficient}, CS-BRM \cite{zheng2024cs}, RRT \cite{lavalle2001randomized} in belief space, and RRT$^\star$\cite{karaman2011sampling} in belief space (similar to \cite{pedram2022gaussian}). We run MAXCOVAR, MAX-COV-BALL, RRT, and RRT$^\star$ until 2000 nodes are added to each tree. We initialize the CS-BRM using the same position nodes as PRISM to ensure good coverage of the position space, then apply the velocity heuristic to add extra nodes with nonzero velocity and find edges as described in \cite{zheng2024cs}. We construct roadmaps using each method in a confined office environment with varying corridor width, and in a cluttered obstacle field environment with varying actuator noise magnitude $W_c$. We evaluate roadmap coverage for both sets of experiments. We also apply PRISM's local optimization module to paths generated by each method in both environments, and compare path costs before and after local optimization. We run all experiments on an AMD Ryzen Threadripper 3960X 24-core processor and solve all optimization problems using the MOSEK Optimizer API \cite{mosek}.

\subsection{Coverage Evaluation - Confined Office Environment}
First, we evaluate coverage and offline computation time in a confined office environment with primary corridor width varying between 1 and 1.8 m after accounting for the geometry of the robot (see \cref{fig:varying_office_envs}) and fixed noise magnitude $W_c = 0.05$ m/s$^3$. The office environments also have narrow corridors that are only 5 cm wider than the diameter of the robot, offering little room to maneuver. We generate a medium-difficulty evaluation set that consists of distributions with stationary initial means and $\Sigma_\mathcal{I} = 0.5 cI$ as calculated in the environment with 1 m passage width, with $c$ defined as in \cref{thm: completeness_with_obstacles}. PRISM is guaranteed by \cref{thm: completeness_with_obstacles} to achieve 100\% coverage on this evaluation set for all corridor widths, and is the only method that consistently maintains high coverage as the passage width narrows. Coverage for all methods across varying passage widths is shown in \cref{fig: office_comp_time}.

For RRT \cite{lavalle2001randomized}, RRT$^\star$ \cite{karaman2011sampling}, MAXCOVAR \cite{aggarwal2025tac}, and MAX-COV-BALL \cite{rose2025efficient}, we find that as the passage width narrows, more samples are required to find each feasible node, increasing offline computation time. MAX-COV-BALL scales best, likely because it is sample efficient by design (see \cref{fig: office_comp_time}). However, PRISM requires less than 1 minute to generate a planning graph for each passage width, and the planning graph only contains 21 nodes. Similarly, CS-BRM scales well across passage widths, requiring around two minutes on average to generate a graph with 83 nodes; however, CS-BRM achieves poor coverage with this sparse graph.

\subsection{Cost Evaluation - Confined Office Environment}
Only PRISM is able to successfully find paths through the very narrow passages in the confined office environment. We evaluate cost after local optimization on different sets of initial distributions in the office environment with 1.8 m passage width. We generate three sets of distributions where all sampling-based methods generally achieve high coverage -- the first set has $x \in [1, 3]$ and $y \in [7, 10]$ (i.e. the mean is in the top left, and there is an obvious shortcut through a narrow passage); the second set has $x, y$ uniformly sampled throughout the collision-free state space; the third set has $x \in [9, 11]$ and $y \in [5, 10]$ (the mean is in the top right, so there are no useful shortcuts through narrow passages). All sets have $\Sigma = 0.1 cI$, with $c$ defined as in \cref{thm: completeness_with_obstacles}. We apply PRISM's local optimization module to feasible trajectories generated by PRISM and all sampling-based methods for these sets. Representative trajectories after local optimization are shown in \cref{fig: cost_office}. 

PRISM achieves the lowest median path cost after local optimization across all three sets. However, PRISM performs relatively best on initial distributions from the top left of the environment; in these cases, PRISM takes a direct route to the goal through a narrow passage. When the quadrotor starts in the top right, this advantage vanishes, but PRISM still outperforms other methods in median final path cost by 3-5\%. When the initial means are randomly sampled throughout the state space, PRISM sometimes finds better modes than other methods and performs about 10\% better than other methods in median cost. Ratios of optimized cost for all sampling-based methods to the optimized cost for PRISM are shown in \cref{fig: cost_office}.

\subsection{Coverage Evaluation - Cluttered Obstacle Field}
We evaluate coverage in a cluttered environment under varying actuator noise ($W_c = 0.05 \text{ m}/\text{s}^3, W_c = 0.25 \text{ m}/\text{s}^3, W_c = 0.5 \text{ m}/\text{s}^3$), with $\mathcal{P}^+$ for each set shown in \cref{fig:varying_office_envs}. For each level of actuator noise, we generate three evaluation sets of varying difficulty. The easy and medium-difficulty sets consist of distributions that satisfy all assumptions required for \cref{thm: completeness_with_obstacles}. The easy set has $\Sigma = 0.1 cI$, and the medium-difficulty set has $\Sigma = 0.5 cI$, with $c$ defined as in \cref{thm: completeness_with_obstacles}. The hard set consists of distributions with non-stationary initial means and large initial covariance.

We show coverage results for these sets in \cref{fig: stationary_coverage_results}. Consistent with \cref{thm: completeness_with_obstacles}, PRISM achieves 100\% coverage on the easy and medium-difficulty evaluation sets regardless of actuator noise magnitude, while performance degrades for all baseline methods as the actuator noise magnitude increases. PRISM achieves 97-100\% coverage on the hard evaluation set, although \cref{thm: completeness_with_obstacles} is insufficient to prove coverage. All other methods achieve less than 45\% coverage on the hard evaluation set regardless of actuator noise magnitude.

PRISM takes between 4 and 5 minutes to build a roadmap for each level of actuator noise, scaling well to a more cluttered environment. In contrast, CS-BRM takes hours to generate a roadmap for each noise level; its velocity heuristic creates thousands of edges that all must be checked for collisions to compute path costs, slowing down roadmap generation.
\cref{fig: coverage_per_iteration} visualizes roadmap coverage as a function of offline computation time for the easy and medium-difficulty evaluation sets with small actuator noise. Because PRISM does not rely on incremental sampling, it achieves 100\% coverage much more efficiently than all other methods. We also find that MAX-COV-BALL achieves higher coverage for a given time than MAXCOVAR, and MAXCOVAR achieves higher coverage for a given time than RRT, which is faster than RRT$^\star$.

\begin{table}[t] 
\centering
\caption{Minimum, median, mean, 75th percentile, 90th percentile, and maximum time (in seconds) needed to compute a feasible path across all cluttered environment scenarios.}
\label{tab: online_comp_time}
\begin{tabular}{|c|c|c|c|c|c|c|}
\hline
\textbf{Method} & \textbf{Min} & \textbf{Med.} & \textbf{Mean} & \textbf{P75} & \textbf{P90} & \textbf{Max} \\
\hline
RRT & \textbf{0.03} & 0.72 & \textbf{0.65} & \textbf{0.74} & 0.77 & \textbf{0.80} \\
\hline
RRT$^\star$ & \textbf{0.03} & 0.71 & \textbf{0.65} & \textbf{0.74} & 0.77 & 0.82 \\
\hline
MAXCOVAR & 0.65 & 0.71 & 0.72 & \textbf{0.74} & \textbf{0.76} & 0.99 \\
\hline
MAX-COV-BALL & 1.22 & 3.60 & 3.49 & 3.71 & 3.78 & 3.89 \\
\hline
\hline
CS-BRM & 8.53 & 11.86 & 11.79 & 12.98 & 13.28 & 15.12 \\
\hline
\textbf{PRISM} & 0.38 & \textbf{0.59} & \textbf{0.65} & 0.80 & 0.86 & 1.10 \\
\hline
\end{tabular}
\end{table}
\cref{tab: online_comp_time} presents statistics for the time required to search for a feasible path across all cluttered environment scenarios for each method. CS-BRM is slow to find a feasible path due to the computational cost incurred by checking for collisions. MAX-COV-BALL is slowed by solving multiple SDPs to connect each node to the tree. PRISM takes 0.65 seconds on average to find a feasible path, tied with the fastest baselines.

\subsection{Cost Evaluation - Cluttered Environment}

We apply PRISM's local optimization module (see \cref{subsec: local_refinement}) to the feasible trajectories found by all methods for the easy evaluation set with small process noise and compare the final costs. We visualize the cost results in \cref{fig: cost_scatter}. After local optimization, PRISM consistently achieves low cost, with median cost of 97.9. CS-BRM and RRT$^\star$ achieve median costs of 311.3 and 531.8 prior to local optimization, and reduce median costs to 110.8 (2.8x reduction) and 114.5 (4.6x reduction) after local optimization. Nevertheless, both have higher median, 25th percentile, and 75th percentile cost than PRISM after local optimization. RRT, MAXCOVAR, and MAX-COV-BALL perform worse, probably because they are more likely to return initial guesses that belong to suboptimal modes. \cref{tab: local_opt_time} and \cref{fig: pareto_front} present results on the local optimization time for each method, and \cref{fig: pareto_front} also visualizes the evolution of trajectory cost over time. Although PRISM typically has higher initial feasible cost than CS-BRM and RRT$^\star$, local optimization converges more efficiently for PRISM trajectories and is feasible to run online. This is because PRISM's global planner produces trajectories with fewer edges than other methods, which speeds up local optimization, especially the shortcut finding phase.

\begin{table}[t] 
\centering
\caption{Median local optimization time, median final trajectory length, and median ratios of local optimization time to final trajectory length and initial feasible trajectory length.} \label{tab: local_opt_time}
\begin{tabular}{|p{2.15cm}|p{0.65cm}|p{1.25cm}|p{1.35cm}|p{1.25cm}|}
\hline
\textbf{Method} & \textbf{Comp. Time (s)} & \textbf{Final Traj. Length (s)} & \textbf{(Time/Final Traj. Length)} & \textbf{(Time/Init. Traj. Length)}\\
\hline
RRT & 22.5 & 12.4 & 1.74 & 0.14 \\
\hline
RRT$^\star$ & 7.51 & 12.0 & \textbf{0.57} & 0.17 \\
\hline
MAXCOVAR & 19.7 & 14.0 & 1.40 & 0.15 \\
\hline
MAX-COV-BALL & 13.2 & 13.4 & 0.95 & \textbf{0.13} \\
\hline 
\hline
CS-BRM & 14.4 & 11.5 & 1.28 & 0.43 \\
\hline
\textbf{PRISM} & \textbf{5.93} & \textbf{9.8} & 0.60 & 0.20\\
\hline
\end{tabular}
\end{table}
\section{Conclusion}\label{sec: conclusion}
We introduce PRISM, a new multi-query planning algorithm for belief space planning under state and control constraints. The main idea of PRISM is to separate mean and covariance steering into a covariance shrinking phase followed by a deterministic mean steering phase. This approach allows PRISM to guarantee completeness over a planning problem with tightened constraints and fixed covariance and to maintain high coverage over the entire belief space. PRISM also includes a novel online local optimization method that co-optimizes the mean and covariance trajectories to reduce path cost. Compared to state-of-the-art sampling-based methods for constructing multi-query roadmaps in belief space \cite{rose2025efficient, aggarwal2024sdp, zheng2024cs, karaman2011sampling}, PRISM achieves superior coverage and produces lower-cost trajectories when planning in cluttered or confined spaces.
\appendix \section{} \label{sec:appendix}
\textbf{Proof of \cref{lemma: delta_shrinkability}}
\begin{proof}
From the proof of \cref{lemma: delta_shrinkability_induction}, $c_{1, \delta}^{(N)} = \min \mathcal{C}_\delta$ and $c_{2, \delta}^{(N)} = \max \mathcal{C}_\delta$ if $\mathcal{C}_\delta$ is nonempty. Suppose now that $\Sigma_1 \prec c_{1, \delta}^{(N)}I$ and that $\Sigma_1$ is $\delta$-shrinkable over $N$ time steps. Then, there exists $K_{1:N-1}, \Sigma_{1:N}$ which are a feasible solution to \cref{prob: covariance_shrinking_delta}. When $K_{1:N-1}$ is applied to $\lambda_{\min}(\Sigma_1)I$, the resulting covariance trajectory is $\Sigma^{(a)}_{1:N}$. The new solution $K_{1:N-1}, \Sigma^{(a)}_{1:N}$ satisfies all dynamics, control, and non-position state constraints in \cref{prob: covariance_shrinking_delta}. Further, because $\Sigma^{(a)}_N \preceq \Sigma_N$, we have that $\lambda_{\max}(\Sigma^{(a)}_N) \leq \lambda_{\min}(\Sigma_1)-\delta$, and so $K_{1:N-1}, \Sigma^{(a)}_{1:N}$ is a feasible solution to \cref{prob: covariance_shrinking_delta} for $cI$ with $c < c_{1, \delta}^{(N)}$. This is a contradiction; therefore, $\Sigma_1$ cannot be $\delta$-shrinkable over $N$ time steps. The same argument can be used to show that $\Sigma_1 \succ c_{2, \delta}^{(N)}I$ $\implies$ $\Sigma_1$ cannot be $\delta$-shrinkable over $N$ time steps. Now, consider $c_{1, \delta}^{(N)}I \preceq \Sigma_1 \preceq c_{2, \delta}^{(N)}I$. By \cref{lemma: delta_shrinkability_induction}, it follows that $\Sigma_1$ is $\delta$-shrinkable over $T$ time steps, with $T = \mathbb{T}(\lambda_{\max}(\Sigma_1), \lambda_{\min}(\Sigma_1), \delta, N)$.
\end{proof}
\textbf{Proof of \cref{lemma: delta_shrinkability_induction}}:
\begin{proof}
 Recall that $c_{1, \delta}^{(N)} = \inf \mathcal{C}_\delta$ and $c_{2, \delta}^{(N)} = \sup \mathcal{C}_\delta$, with $\mathcal{C}_\delta \!:=\! \{c \!\in\! \mathbb{R}_{\geq 0}\!:\!\! cI \text{ is }\delta-\text{strictly shrinkable} \text{ in } N \text{ steps} \}$ for fixed $\delta$. Because \cref{prob: covariance_shrinking_delta} is bounded by equalities and non-strict inequalities, $\mathcal{C}_\delta$ must be closed; if it is nonempty, $c_{1, \delta}^{(N)} = \min \mathcal{C}_\delta$ and $c_{2, \delta}^{(N)} = \max \mathcal{C}_\delta$. Suppose that ($K_{1:N-1}, \Sigma_{1:N}$) is a feasible solution to \cref{prob: covariance_shrinking_delta}. Then, when $K_{1:N-1}$ is applied to any $\Sigma_1^{(a)} \preceq \Sigma_1$, we can propagate the covariance trajectory $\Sigma^{(a)}_{1:N}$ using \cref{eq: covariance_dynamics}. The new solution $K_{1:N-1}, \Sigma^{(a)}_{1:N}$ may not satisfy $\lambda_{\max}(\Sigma^{(a)}_N) \leq \lambda_{\min}(\Sigma_1^{(a)})-\delta$, but it will satisfy all other constraints in \cref{prob: covariance_shrinking_delta} and will satisfy $\Sigma^{(a)}_i \preceq \Sigma_i$ $\forall i \in [N]$. $\Sigma_1 \preceq c_{2, \delta}^{(N)}I$ and $c \in [c_{1, \delta}, \lambda_{\max}(\Sigma_1)]$, so $c \in \mathcal{C}_\delta$. Define $z = \left\lceil \frac{\lambda_{\max}(\Sigma_1)-c}{\delta}\right\rceil$. Because $c \in \mathcal{C}_\delta$, for every $i \in [z-1]$, $\lambda_{\max}(\Sigma_1) - \delta i \in \mathcal{C}_\delta$ and so there exists a feedback control policy $K_{1:N-1}$ that steers any $\Sigma_1 \preceq (\lambda_{\max}(\Sigma_1) - \delta i)I$ to $\Sigma_N \preceq (\lambda_{\max}-\delta(i+1))I$ while obeying all control and non-position state constraints in \cref{prob: covariance_shrinking_delta}. There also exists a feedback control policy $K_{1:N-1}$ that steers $\Sigma_1 \preceq cI$ (with $cI \succeq (\lambda_{\max}(\Sigma_1)-\delta z)I$) to $\Sigma_N \preceq (c-\delta)I$ while obeying all control and non-position state constraints in \cref{prob: covariance_shrinking_delta}. Concatenating all of these control policies results in a feedback control policy $K_{1:(z+1)(N-1)}$ that steers $\Sigma_1 \preceq \lambda_{\max}(\Sigma_1)I$ to $\Sigma_{z(N-1)+N} \preceq (c-\delta)I$ over $z(N-1)+N$ steps while obeying all control and non-position state constraints in \cref{prob: covariance_shrinking_delta}.
\end{proof}
\textbf{Proof of \cref{corollary: finite_time_shrinkability}}:
\begin{proof}
\cref{corollary: finite_time_shrinkability} is a special case of \cref{lemma: delta_shrinkability_induction}.
\end{proof}
\textbf{Proof of \cref{corollary: maintainability}}
\begin{proof}
Any valid solution to \cref{prob: covariance_shrinking_delta} with $\Sigma_1 = c_{1, \delta} I$ will satisfy $\Sigma_N \preceq (c_{1, \delta}^{(N)}-\delta) I \preceq c_{1, \delta}^{(N)} I$.
\end{proof}
\textbf{Proof of \cref{lemma: r_delta_strict_shrinkability}}
\begin{proof}
Because $\Sigma_1$ is $\delta$-strictly shrinkable in $N$ steps and $\mu_1$ is stationary, \cref{prob: covariance_shrinking_delta} must have a feasible solution with covariance trajectory $\Sigma_{1:N}$. Then, \cref{prob: covariance_shrinking_r} is feasible for any $r \geq \max_k \lambda_{\max}(P\Sigma_kP^T)$. \cref{prob: covariance_shrinking_r} has a convex and nonempty feasible set and a continuous and convex objective function; therefore, by the Weierstrass extreme value theorem, there must exist $r^\star$ such that $r^\star$ is the minimizer of $r$ over the feasible set of \cref{prob: covariance_shrinking_r}.

If $\mathbb{B}_p(P\mu_1, \Phi^{-1}(1-\epsilon)\sqrt{r^\star}) \in \mathcal{S}$, then $\mathbf{x} \in \mathbb{B}_p(P\mu_1, \Phi^{-1}(1-\epsilon)\sqrt{r^\star}) \implies \mathbf{x} \in \mathcal{S}$. It follows that for any $\mathbf{x} \in \mathbb{R}^p$ such that $||\mathbf{x}||_2 = 1$, ${a_i^s}^T(P\mu_1 + \Phi^{-1}(1-\epsilon)\sqrt{r^\star}\mathbf{x}) \leq b_i^s$ for all $a_i^s, b_i^s$ such that $\mathcal{S} = \{\mathbf{x} \in \mathbb{R}^p : \bigcap_i {a_i^s}^T\mathbf{x} \leq b_i^s\}$. Because $||a_i^s||_2 = 1$, $||\mathbf{x}||_2 = 1$, ${a_i^s}^T\mathbf{x} = 1$ when $\mathbf{x} = a_i^s$. Therefore, ${a_i^s}^TP\mu_1 + \Phi^{-1}(1-\epsilon)\sqrt{r^\star}\leq b_i^s$. Because $r^\star \geq \max_k \lambda_{\max}(P\Sigma_kP^T)$ and $\lambda_{\max}(P\Sigma_kP^T) \geq {a_i^s}^TP\Sigma_kP^Ta_i^s$, it then follows that ${a_i^s}^TP\mu_1 + \Phi^{-1}(1-\epsilon)\sqrt{{a_i^s}^TP\Sigma_kP^Ta_i^s} \leq b_i^s$ for all $i, k$. Rearranging yields $\Phi^{-1}(1-\epsilon)^2({a_i^s}^TP\Sigma_kP^Ta_i^s) \leq (b_i^s - {a_i^s}^TP\mu_1)^2$ for all $i, k$; therefore, $\mathcal{N}(\mu_1, \Sigma_1)$ is $\delta$-strictly shrinkable in $\mathcal{S}$ by the definition of \cref{prob: covariance_shrinking_delta_S}. 
\end{proof}
\textbf{Proof of \cref{thm: r_c_shrinkability}}
\begin{proof}
Suppose $\Sigma_1 \prec c_{1, \delta}^{(N)}I$. By \cref{lemma: delta_shrinkability}, $\mathcal{N}(\mu_1, \Sigma_1)$ is not $\delta$-strictly shrinkable in $N$ steps, and so cannot be $\delta$-strictly shrinkable in $\mathcal{S}$ in $N$ steps. Now, suppose that $c_{1, \delta}^{(N)}I \preceq \Sigma_1 \preceq c_{d_1} I$. By \cref{lemma: r_delta_strict_shrinkability}, $\lambda_{\max}(\Sigma_1)I$ is $\delta$-strictly shrinkable in $\mathcal{S}$ in $N$ steps and by \cref{corollary: constrained_path_existence}, there exists a control policy over $T = \mathbb{T}(\lambda_{\max}(\Sigma_1), \lambda_{\min}(\Sigma_1), \delta, N)$ steps under which $\Sigma_1$ is $\delta$-strictly shrinkable in $\mathcal{S}$.

Finally, suppose that $\Sigma_1 \succ c_{d_2}I$. It follows that $r^\star(\Sigma_1) > r^\star(c_{d_2}I)$ and so any $N$-step control policy applied to $\Sigma_1$ produces a covariance trajectory $\Sigma_{1:N}$ such that $\max_k \lambda_{\max}(P\Sigma_kP^T) = r^\star(\Sigma_1) > r^\star(c_{d_2})$. $\mathcal{S} \in \text{Int}(\mathbb{B}_p(P\mu_1, d_2))$, so for any $\mathbf{x}$ such that $||\mathbf{x}||_2 = 1$ and for any $a_i^s, b_i^s$ such that $\mathcal{S} = \{\mathbf{x} \in \mathbb{R}^p : \bigcap_i {a_i^s}^T\mathbf{x} \leq b_i^s\}$, ${a_i^s}^T(P\mu_1 + d_2 \mathbf{x}) > b_i^s$. Now, for any $k$, $\sqrt{{a_i^s}^TP\Sigma_kP^Ta_i^s} = \sqrt{{a_i^s}^T \sum_j \lambda_j \mathbf{v}_j\mathbf{v}_j^T a_i^s}$ by the spectral theorem. For some $k$, $r^\star(\Sigma_1)$ is an eigenvalue of $P\Sigma_kP^T$, we can say that $\max_k \sqrt{{a_i^s}^TP\Sigma_kP^Ta_i^s} \geq \sqrt{r^\star(\Sigma_1)}{a_i^s}^T\mathbf{v}_j$ for some $\mathbf{v}_j$. Because $\Phi^{-1}(1-\epsilon)\sqrt{r^\star(\Sigma_1)}{a_i^s}^T\mathbf{v}_j \geq d_2 {a_i^s}^T\mathbf{v}_j$, it follows that $\max_k {a_i^s}^TP\mu_1 + \Phi^{-1}(1-\epsilon)\sqrt{{a_i^s}^TP\Sigma_kP^Ta_i^s} > b_i^s$ and therefore that $\Sigma_1$ is not $\delta$-shrinkable in $N$ steps in $\mathcal{S}$.
\end{proof}
\textbf{Proof of \cref{corollary: constrained_path_existence}}
\begin{proof}
\cref{corollary: constrained_path_existence} follows from the same argument as \cref{lemma: delta_shrinkability_induction} for the $\delta$-shrinkability in $\mathcal{S}$ case.
\end{proof}
\textbf{Proof of \cref{lemma: decomposition}}
\begin{proof}
Suppose that $\mathcal{X}_p^- \setminus \mathcal{P}^+$ is equal to the Minkowski sum $\mathcal{S}^- \oplus \mathbb{B}_p(0, \epsilon_p)$ for some $\epsilon_p > 0$, with $\mathcal{S}^-$ connected. $\mathcal{X}_p^- \setminus \mathcal{P}^+$ must therefore also be connected and can be written exactly as the union of a finite number of convex polytopes $\bigcup_{i} \mathcal{S}_i$ with $\text{Int}(\mathcal{S}_i) \neq \emptyset$ for all $i$.

Consider $\mathcal{S}_i$. There must exist some point $x \in \text{Int}(\mathcal{S}_i)$ s.t. $x \in \mathbb{B}_p(c, \epsilon_p)$ for some $c \in \mathcal{S}^-$, but $\mathbb{B}_p(c, \epsilon_p) \not \subset \mathcal{S}_i$. Otherwise, for every $\mathbb{B}_p(c, \epsilon_p) \in \mathcal{X}_p^- \setminus \mathcal{P}^+$ s.t. $c \in \mathcal{S}^-$, $\mathbb{B}_p(c, \epsilon_p) \subset \mathcal{S}_i$ or $\mathbb{B}_p(c, \epsilon_p) \cap \text{Int}(\mathcal{S}_i) = \emptyset$. This is impossible unless $\mathcal{S}_i = \mathcal{X}_p^- \setminus \mathcal{P}^+$, because $\mathcal{S}^-$ is connected. Then, $\text{Int}(\mathbb{B}_p(c, \epsilon_p) \cap \mathcal{S}_i) \neq \emptyset$, and $\text{Int}(\mathbb{B}_p(c, \epsilon_p) \cap ((\mathcal{X}_p^- \setminus \mathcal{P}^+) \setminus \mathcal{S}_i)) \neq \emptyset$ for some $c$. Therefore, $\text{Int}(\mathbb{B}_p(c, \epsilon_p) \cap \mathcal{S}_j) \neq \emptyset$ for some $j$. Then, we can inscribe a polytope $\mathcal{S}_{i \cap j} \subset \mathbb{B}_p(c, \epsilon_p)$ closely approximating $\mathbb{B}_p(c, \epsilon_p)$ (i.e. s.t. $\mathbb{B}_p(c, \epsilon_p-\delta) \subset \mathcal{S}_{i \cap j}$ for small $\delta$) s.t. $\text{Int}(\mathcal{S}_{i \cap j} \cap \mathcal{S}_i) \neq \emptyset$, $\text{Int}(\mathcal{S}_{i \cap j} \cap \mathcal{S}_j) \neq \emptyset$. Suppose we add $\mathcal{S}_{i \cap j}$ to the cover $\bigcup_i \mathcal{S}_i$ for every $i, j$ such that a ball can be constructed such that $\text{Int}(\mathbb{B}_p(c, \epsilon_p) \cap \mathcal{S}_i) \neq \emptyset$ and $\text{Int}(\mathbb{B}_p(c, \epsilon_p) \cap \mathcal{S}_j) \neq \emptyset$ for some $j$.

Now consider $\mathcal{S}_A, \mathcal{S}_B$ s.t. $\mathcal{X}_p^- \setminus \mathcal{P}^+ = \mathcal{S}_A \cup \mathcal{S}_B = \bigcup_{i} \mathcal{S}_i$. $\mathcal{X}_p^- \setminus \mathcal{P}^+$ is connected and therefore $\mathcal{S}_A \cap \mathcal{S}_B \neq \emptyset$. Then, there must exist $\mathbb{B}_p(c, \epsilon_p) \in \mathcal{X}_p^- \setminus \mathcal{P}^+$ s.t. $c \in \mathcal{S}^-$ and $\text{Int}(\mathbb{B}_p(c, \epsilon_p) \cap \mathcal{S}_i) \neq \emptyset$ and $\text{Int}(\mathbb{B}_p(c, \epsilon_p) \cap \mathcal{S}_j) \neq \emptyset$ for some $\mathcal{S}_i \in \mathcal{S}_A, \mathcal{S}_j \in \mathcal{S}_B$ (otherwise, $\mathcal{S}_A = \mathcal{X}_p^- \setminus \mathcal{P}^+, \mathcal{S}_B = \emptyset$ or $\mathcal{S}_A = \emptyset, \mathcal{S}_B = \mathcal{X}_p^- \setminus \mathcal{P}^+$). Therefore, $\mathcal{S}_{i \cap j}$ is in the cover and either $\mathcal{S}_{i \cap j} \in \mathcal{S}_A$ or $\mathcal{S}_{i \cap j} \in \mathcal{S}_B$. Either way, $\text{Int}(\mathcal{S}_A \cap \mathcal{S}_B) \neq \emptyset$.
\end{proof}
\textbf{Proof of \cref{thm: controllability}}
\begin{proof}
Suppose that $\mu_\mathcal{I}, \mu_\mathcal{G} \in \text{Int}(\mathcal{S})$ and $(A, B)$ is controllable. It follows that $\exists \epsilon > 0$ such that $\mathbb{B}_p(P\mu_\mathcal{I}, \epsilon) \subset \mathcal{S}$ and such that $\mathbb{B}_p(P\mu_\mathcal{G}, \epsilon) \subset \mathcal{S}$, or equivalently ${a_i^s}^T\mu_\mathcal{I} \leq b_i^s - \epsilon$ and ${a_i^s}^T\mu_\mathcal{G} \leq b_i^s - \epsilon$ for all $i$ such that $\mathcal{S} := \{\mathbf{x} \in \mathbb{R}^p: \bigcap_i {a_i^s}^T\mathbf{x} \leq b_i^s \}$. It follows by convexity that for any $\alpha \in [0, 1]$, $\mathbb{B}_p(P(\alpha \mu_\mathcal{I} + (1-\alpha)\mu_\mathcal{G}), \epsilon) \subset \mathcal{S}$ and also that $A(\alpha \mu_\mathcal{I} + (1-\alpha)\mu_\mathcal{G}) = \alpha \mu_\mathcal{I} + (1-\alpha)\mu_\mathcal{G}$. We will assume without loss of generality that $\mathbb{B}_p(0, \epsilon) \subset \mathcal{U}^-$ and $\mathbb{B}_p(0, \epsilon) \subset \mathcal{X}_s^-$. Now, consider the minimum-fuel optimal control to reach $(\mu_\mathcal{G}-\mu_\mathcal{I})\delta + \mu_\mathcal{I}$ from $\mu_\mathcal{I}$ over $N-1$ control steps, with $\mu_1 = \mu_\mathcal{I}$, $\mu_N = \mu_\mathcal{G}$. If $(A, B)$ is controllable, such a trajectory must exist if $N \geq n$. The control inputs at time steps $k \in [N-1]$ can be expressed as $\mathbf{v}_k = B^T{A^T}^{N-2-k}(\sum_{t=1}^{N-1}A^tBB^T{A^T}^t)^{-1}(\mu_\mathcal{G}-\mu_\mathcal{I})\delta$. Because $\mu_\mathcal{I}$ is stationary, the state trajectory can be expressed as $\mu_k = \mu_1 + \sum_{t=1}^{k-1}A^{k-2-t}B\mathbf{v}_t$. Defining $C_1 = \max_k ||B^T{A^T}^{N-2-k}(\sum_{t=1}^{N-1}A^tBB^T{A^T}^t)^{-1}(\mu_\mathcal{G}-\mu_\mathcal{I})||_2$ and $C_2 = \max_k ||\sum_{t=1}^{k-1}A^{k-2-t}B||_2 C_1$, we can see that $C_1, C_2$ are finite constant expressions of $A, B, \mu_\mathcal{G}$, and $\mu_\mathcal{I}$, with $||\mathbf{v}_k||_2 \leq C_1 \delta$ for all $k$ and $||\mu_k - \mu_\mathcal{I}||_2 \leq C_2 \delta$ for all $k$. Therefore, if $0 < \delta \leq \min(\frac{\epsilon}{C_2}, \frac{\epsilon}{C_1})$, the state and control trajectory satisfy all constraints, and so there exists a solution to \cref{prob: mean_steering_single_polytope} with the initial state equal to $\mu_\mathcal{I}$ and the goal state equal to $\mu_\mathcal{I} + \delta(\mu_\mathcal{G}-\mu_\mathcal{I})$ over $N$ steps. We can chain $\lceil1/\delta\rceil$ of these trajectories to reach $\mu_\mathcal{G}$; therefore, \cref{prob: mean_steering_single_polytope} must have a feasible solution from $\mu_\mathcal{I}$ to $\mu_\mathcal{G}$ in $N\lceil \frac{1}{\delta} \rceil$ steps.
\end{proof}
\textbf{Proof of \cref{thm: completeness_without_obstacles}}
\begin{proof}
In the absence of position constraints, PRISM does not build a graph offline (see \cref{subsec: planning_without_position_constraints}). Because $\mu_\mathcal{I} = A\mu_\mathcal{I}$, PRISM begins by repeatedly solving \cref{prob: shrinking_without_position_constraints} to shrink $\Sigma_1$ to $\Sigma_T \preceq c_{1, \delta}^{(N)}I$. By \cref{corollary: finite_time_shrinkability}, there exists a control policy satisfying all control and non-position state constraints that shrinks $\Sigma_1$ to $\Sigma_T \preceq c_{1, \delta}^{(N)}I$ over $T = \mathbb{T}(\lambda_{\max}(\Sigma_1), c_{1, \delta}^{(N)}, \delta, N)$ time steps. Therefore, solving \cref{prob: shrinking_without_position_constraints} will be (recursively) feasible and solving it $\lceil\frac{\lambda_{\max}(\Sigma_1) - c_{1, \delta}^{(N)}}{\delta} \rceil$ times will reach $\Sigma_T \preceq c_{1, \delta}^{(N)}I$. Because $||\mu_\mathcal{I} - \mu_\mathcal{G}||_2$ is finite, we can construct a safe polytope $\mathcal{S}$ of finite volume such that $\mu_\mathcal{I}, \mu_\mathcal{G} \in \text{Int}(\mathcal{S})$. Then, by \cref{thm: controllability}, \cref{alg: prism_mean_steering} will find a finite-time belief steering edge connecting $\mathcal{N}(\mu_\mathcal{I}, c_{1, \delta}^{(N)}I)$ to $\mathcal{N}(\mu_\mathcal{G}, c_{1, \delta}^{(N)}I)$. The final distribution satisfies $c_{1, \delta}^{(N)}I \preceq \Sigma_\mathcal{G}$, so the sequence of edges found by PRISM is a path from $\mathcal{I}$ to $\mathcal{G}$. Assuming each subproblem (e.g. solving a finite-size SDP or solving a finite-size graph search problem) is solvable in finite time, PRISM will return a trajectory of finite length in finite time.
\end{proof}
\textbf{Proof of \cref{thm: completeness_with_obstacles}}
\begin{proof}
By \cref{lemma: graph_connectivity}, PRISM constructs a strongly connected mean space planning graph in its offline planning phase. The convex sets represented in the graph completely cover $\mathcal{X}_p^- \setminus \mathcal{P}^+$. When planning online, PRISM first finds $\mathcal{S} \subset (\mathcal{X}_p \setminus \mathcal{P})$ such that $\mathbb{B}_p(P\mu_\mathcal{I}, d) \in \mathcal{S}$ by finding separating hyperplanes between $\mathbb{B}_p(P\mu_\mathcal{I}, d)$ and each obstacle in $\mathcal{P}$.

Next, PRISM repeatedly solves \cref{prob: shrinking_subject_to_constraints}. Because $\Sigma_\mathcal{I} \preceq cI$ with $\Phi^{-1}(1-\epsilon)\sqrt{r^\star(c)} = d$, there exists a control policy satisfying all control, non-position state constraints, and position constraints from $\mathcal{S}$ that shrinks $\Sigma_\mathcal{I}$ to $\Sigma_1 \preceq c_{1, \delta}^{(N)} I$ over $T = \mathbb{T}(\lambda_{\max}(\Sigma_\mathcal{I}), c_{1, \delta}^{(N)}, \delta, N)$ steps (by \cref{corollary: constrained_path_existence}). Therefore, \cref{prob: shrinking_subject_to_constraints} must be (recursively) feasible, and solving it repeatedly will lead to $\Sigma_T \preceq c_{1, \delta}^{(N)}I$.

Next, PRISM will find a polytope $\mathcal{S}_i$ in the convex set decomposition of $\mathcal{X}_p^- \setminus \mathcal{P}^+$ such that $\mu_\mathcal{I} \in \text{Int}(\mathcal{S}_i)$. The convex set decomposition is complete and connected (by \cref{asm: connectedness}) and $\mu_\mathcal{I} \in \text{Int}(\mathcal{X}_p^- \setminus \mathcal{P}^+)$, so such a set must exist. By \cref{thm: controllability}, \cref{alg: prism_mean_steering} will find a finite-time belief steering edge connecting $\mathcal{N}(\mu_\mathcal{I}, c_{1, \delta}^{(N)}I)$ to $\mathcal{N}(\mu_i, c_{1, \delta}^{(N)}I)$, where $\mu_i = \text{CENTROID}(\mathcal{S}_i)$. Similarly, PRISM will find a polytope $\mathcal{S}_j$ in the convex set decomposition of $\mathcal{X}_p^- \setminus \mathcal{P}^+$ such that $\mu_\mathcal{G} \in \text{Int}(\mathcal{S}_j)$. Because $\mu_\mathcal{G} \in \text{Int}(\mathcal{X}_p^- \setminus \mathcal{P}^+)$ and the convex set decomposition is complete, such a set must exist. By \cref{lemma: graph_connectivity}, the mean space planning graph is strongly connected, so graph search will find a path (or sequence of edges) connecting $v_i$, the vertex associated with $\mathcal{S}_i$, to $v_j$, the vertex associated with $\mathcal{S}_j$. By \cref{thm: controllability}, \cref{alg: prism_mean_steering} will find a finite-time belief steering edge connecting $\mathcal{N}(\mu_j, c_{1, \delta}^{(N)}I)$ to $\mathcal{N}(\mu_\mathcal{G}, c_{1, \delta}^{(N)}I)$, where $\mu_j = \text{CENTROID}(\mathcal{S}_j)$. The final distribution satisfies $c_{1, \delta}^{(N)}I \preceq \Sigma_\mathcal{G}$, so the sequence of edges found by PRISM is a path from $\mathcal{I}$ to $\mathcal{G}$. Assuming each subproblem (e.g. solving a finite-size SDP or solving a finite-size graph search problem) is solvable in finite time, PRISM will return a trajectory of finite length in finite time.
\end{proof}
\textbf{Proof of \cref{thm: soundness}}
\begin{proof}
Any feasible global trajectory found by PRISM can only contain edges found by solving \cref{prob: stopping_phase_convex} (or \cref{prob: stopping_phase_convex_no_pos_constraints} in the absence of position constraints), edges found by combining a covariance trajectory found by solving \cref{prob: shrinking_subject_to_constraints} (or \cref{prob: shrinking_without_position_constraints} in the absence of position constraints) with zero feedforward control, and edges found by \cref{alg: prism_mean_steering}. PRISM applies local optimization to this sequence of edges. PRISM chains maneuvers such that boundary constraints between edges are always satisfied. Any solution to \cref{prob: stopping_phase_convex} satisfies \cref{eq: mean_dynamics}, \cref{eq: convex_covariance_dynamics}, \cref{eq: schur_complement_constraint}, \cref{eq: convex_ctrl_chance_constraint}, \cref{eq: convex_pos_chance_constraint}, and \cref{eq: convex_nonpos_state_chance_constraint}. Satisfying \cref{eq: convex_covariance_dynamics} and \cref{eq: schur_complement_constraint} implies satisfaction of \cref{eq: covariance_dynamics}. Satisfying \cref{eq: convex_ctrl_chance_constraint} implies satisfaction of \cref{eq: control_chance_constraint}, satisfying \cref{eq: convex_pos_chance_constraint} implies satisfaction of \cref{eq: pos_chance_constraint}, and satisfying \cref{eq: convex_nonpos_state_chance_constraint} implies satisfaction of \cref{eq: nonpos_state_chance_constraint}. Then, any solution to \cref{prob: stopping_phase_convex} is sound by \cref{def: soundness}. Similarly, any solution to \cref{prob: stopping_phase_convex_no_pos_constraints} satisfies \cref{eq: mean_dynamics}, \cref{eq: covariance_dynamics}, \cref{eq: control_chance_constraint} and \cref{eq: nonpos_state_chance_constraint}, and has no position constraints (i.e. $M_{s_k} = 0$ for all $k$).

Now, consider edges found by combining a solution to \cref{prob: shrinking_subject_to_constraints} (or \cref{prob: shrinking_without_position_constraints}) with zero feedforward control, such that $\mathbf{v}_k = 0$ for all $k \in [N-1]$ and $\mu_k = \mu_T$ for all $k \in [N]$. PRISM only finds such edges if $\mu_T = A\mu_T$; therefore, \cref{eq: mean_dynamics} is satisfied for $(\mu_{1:N}, \mathbf{v}_{1:N-1})$. Any solution $(\Sigma_{1:N}, K_{1:N-1})$ to \cref{prob: shrinking_subject_to_constraints} must satisfy \cref{eq: convex_covariance_dynamics} and \cref{eq: schur_complement_constraint} and therefore must satisfy \cref{eq: covariance_dynamics}. Any solution $(\Sigma_{1:N}, K_{1:N-1})$ to \cref{prob: shrinking_subject_to_constraints} must also satisfy $\Phi^{-1}(1-\epsilon)^2({a_j^u}^TY_ka_j^u) \leq (b_j^u - \delta)^2$ for $k \in [N-1]$, $j \in [M_u]$. Therefore, it must also satisfy $\Phi^{-1}(1-\epsilon)\sqrt{{a_j^u}^TY_ka_j^u} \leq b_j^u$, and because $Y_k \succeq K_k\Sigma_kK_k^T$ by \cref{eq: schur_complement_constraint}, must satisfy $\Phi^{-1}(1-\epsilon)\sqrt{{a_j^u}^TK_k\Sigma_kK_k^Ta_j^u} \leq b_j^u$. Because $\mathbf{v}_k =0$ for all $k$, $\Phi^{-1}(1-\epsilon)\sqrt{{a_j^u}^TK_k\Sigma_kK_k^Ta_j^u} + {a_j^u}^T\mathbf{v}_k \leq b_j^u$ for $k \in [N-1]$, $j \in [M_u]$. Similarly, any solution $(\Sigma_{1:N}, K_{1:N-1})$ to \cref{prob: shrinking_subject_to_constraints} must also satisfy $\Phi^{-1}(1-\epsilon)^2({a_i^{x_s}}^TS\Sigma_kS^Ta_i^{x_s}) \leq (b_i^{x_s} - \delta)^2$ and therefore $\Phi^{-1}(1-\epsilon)\sqrt{{a_i^{x_s}}^TS\Sigma_kS^Ta_i^{x_s}} \leq b_i^{x_s}$ for $k \in [N]$, $i \in [M_{x_s}]$. Recalling that $\mu_k = \mu_T$ for all $k$ and that $A\mu_T = \mu_T$ and therefore that $S\mu_T = 0$, it follows that $\Phi^{-1}(1-\epsilon)\sqrt{{a_i^{x_s}}^TS\Sigma_kS^Ta_i^{x_s}} + {a_i^{x_s}}^TS\mu_k \leq b_i^{x_s}$ for $k \in [N]$, $i \in [M_{x_s}]$ for $k \in [N]$, $i \in [M_{x_s}]$. Finally, any solution $(\Sigma_{1:N}, K_{1:N-1})$ to \cref{prob: shrinking_subject_to_constraints} must also satisfy $\Phi^{-1}(1-\epsilon)^2({a_i^{s}}^TP\Sigma_kP^Ta_i^{s}) \leq (b_i^{s} - {a_i^s}^TP\mu_1)^2$ for all $k \in [N]$, $\ell \in [M_s]$. With $\mu_k = \mu_T = \mu_1$ for all $k \in [N]$, ${a_i^s}^TP\mu_k + \Phi^{-1}(1-\epsilon)\sqrt{{a_i^{s}}^TP\Sigma_kP^Ta_i^{s}} \leq b_i^{s}$ for all $k \in [N]$, $\ell \in [M_s]$. Then, any edge combining a solution to \cref{prob: shrinking_subject_to_constraints} (or, by a similar argument, a solution to \cref{prob: shrinking_without_position_constraints}) with zero feedforward control must satisfy \cref{def: soundness}. By \cref{lemma: mean_soundness}, edges found by \cref{alg: prism_mean_steering} must be sound. So, any global PRISM trajectory is a sequence of sound edges. By \cref{lemma: recursive_feasibility}, local optimization preserves soundness of the input sequence of edges; then, any trajectory found by PRISM satisfies \cref{def: soundness}.
\end{proof}
\textbf{Proof of \cref{lemma: graph_connectivity}}
\begin{proof}
By \cref{asm: connectedness} and \cref{lemma: decomposition}, $\mathcal{X}_p^-\setminus\mathcal{P}^+$ can be decomposed into a union of a finite number of convex polytopes such that $\mathcal{X}_p^- \setminus \mathcal{P}^+ = \bigcup_{i=1}^{M_s^-} \mathcal{S}_i$ (e.g. using \cite{shaikh2025exact}) and such that for any $\mathcal{S}_A, \mathcal{S}_B$ such that $\mathcal{X}_p^- \setminus \mathcal{P}^+ = \mathcal{S}_A \cup \mathcal{S}_B = \bigcup_{i=1}^{M_s^-} \mathcal{S}_i$, $\text{Int}(\mathcal{S}_A \cap \mathcal{S}_B) \neq \emptyset$. When $\delta > 0$, $\mathcal{X}_s^-$ and $\mathcal{U}^-$ both contain the origin and have nonempty interior. Then, by \cref{asm: controllability} and \cref{thm: controllability}, if $\mathcal{S}_i, \mathcal{S}_j \subset (\mathcal{X}_p^- \setminus \mathcal{P}^+)$ and $\mathcal{S}_i, \mathcal{S}_j$ have nonempty interior, with $\mu_i, \mu_j \in \text{Int}(\mathcal{S}_i), \text{Int}(\mathcal{S}_j)$, if $\text{Int}(\mathcal{S}_i \cap \mathcal{S}_j) \neq \emptyset$, \cref{prob: mean_steering}($\mu_i, \mu_j, \mathcal{U}^-, \mathcal{X}_s^-, \mathcal{S}_i, \mathcal{S}_j, T$) has a solution for some $T$. Therefore, for any $\mathcal{S}_i, \mathcal{S}_j$ in the union of convex sets covering $\mathcal{X}_p^- \setminus \mathcal{P}^+$ such that $\text{Int}(\mathcal{S}_i \cap \mathcal{S}_j) \neq \emptyset$, there exist vertices $v_i$ and $v_j$ in the mean space graph corresponding to $\mathcal{S}_i$ and $\mathcal{S}_j$, and there exist edges connecting $v_i$ to $v_j$ and $v_j$ to $v_i$. For any $\mathcal{S}_A, \mathcal{S}_B$ such that $\mathcal{X}_p^- \setminus \mathcal{P}^+ = \mathcal{S}_A \cup \mathcal{S}_B = \bigcup_{i=1}^{M_s^-} \mathcal{S}_i$, $\text{Int}(\mathcal{S}_A \cap \mathcal{S}_B) \neq \emptyset$, so it follows that $\exists \mathcal{S}_i \in \mathcal{S}_A$, $\mathcal{S}_j \in \mathcal{S}_B$ such that $\text{Int}(\mathcal{S}_i \cap \mathcal{S}_j) \neq \emptyset$. Thus, the mean space graph must be strongly connected.
\end{proof}
\textbf{Proof of \cref{lemma: mean_soundness}}
\begin{proof}
By \cref{thm: controllability}, \cref{prob: mean_steering_single_polytope} has a feasible solution from $\mu_\mathcal{I}$ to any mean in $\text{Int}(\mathcal{S}_\mathcal{I} \cap \mathcal{S}_\mathcal{J})$, and \cref{prob: mean_steering_single_polytope} also has a feasible solution from any mean in $\text{Int}(\mathcal{S}_\mathcal{I} \cap \mathcal{S}_\mathcal{J})$ to $\mu_\mathcal{J}$. Therefore, \cref{prob: mean_steering} must have a feasible solution $\mu_{1:2T-1}, \mathbf{v}_{1:2T-2}$ from $\mu_\mathcal{I}$ to $\mu_\mathcal{J}$ over $2T-1$ steps for some finite $T$ such that $T = t(N-1) + 1$ for some finite $t \in \mathbb{Z}_{> 0}$. By the constraints given in \cref{prob: mean_steering}, ($\mu_{1:2T-1}, \mathbf{v}_{1:2T-2}$) must satisfy \cref{eq: mean_dynamics} for $k \in [2T-2]$.

$K_{1:2T-2}$ is given by concatenating $K_{1:N-1}^{(c_{1, \delta})}$ $2t$ times, and that $\Sigma_{1:2T-1}$ is given by propagating the covariance dynamics in \cref{eq: covariance_dynamics} with $\Sigma_1 = c_{1, \delta}^{(N)}I$ and the feedback control policy given by $K_{1:2T-2}$. It follows that $\Sigma_{1:N} = \Sigma_{1:N}^{(c_{1, \delta})}$ by definition. It follows also that $\Sigma_N \preceq \Sigma_1$ and therefore that $\Sigma_{N:2N-1} \preceq \Sigma_{1:N}$. Then, $\Sigma_{z(N-1) + i} \preceq \Sigma_i^{(c_{1, \delta})}$ for any $i \in [N]$, $z \in [2t-1]$ and \cref{eq: covariance_dynamics} is satisfied for all $k \in [2T-2]$.

For all $k \in [2T-1]$, $i \in [M_{x_s}]$, ($\mu_{1:2T-1}, \mathbf{v}_{1:2T-2}$) obeys ${a_i^{x_s}}^TS\mu_k \leq b_i^{x_s} - \Phi^{-1}(1-\epsilon)\delta_{b_i^{x_s}}$. Recalling that $\delta_{b_i^{x_s}} = \max_k \sqrt{{a_i^{x_s}}^TS\Sigma_k^{(c_{1, \delta})}S^Ta_i^{x_s}}$ and that $\Sigma_{k} \preceq \Sigma_{\text{mod}(k, N)}^{(c_{1, \delta})}$ for all $k \in [2T-1]$, ${a_i^{x_s}}^TS\mu_k + \Phi^{-1}(1-\epsilon)\sqrt{{a_i^{x_s}}^TS\Sigma_kS^Ta_i^{x_s}} \leq b_i^{x_s}$ for all $k \in [2T-1]$,$i \in [M_{x_s}]$. Similarly, for all $k\in[2T-2]$, $j \in [M_u]$, ($\mu_{1:2T-1}, \mathbf{v}_{1:2T-2}$) obeys ${a_j^u}^T\mathbf{v}_k \leq b_j^u - \Phi^{-1}(1-\epsilon)\delta_{b_j^u}$. Recall that $\delta_{b_j^u} = \max_k \sqrt{{a_j^u}^TK_k^{(c_{1, \delta})}\Sigma_k^{(c_1, \delta)}{K_k^{(c_{1, \delta})}}^Ta_j^u}$, that $K_k = K_{\text{mod}(k, N-1)}^{(c_{1, \delta})}$ and that $\Sigma_{k} \preceq \Sigma_{\text{mod}(k, N)}^{(c_{1, \delta})}$ for all $k\in[2T-1]$. Then, ${a_j^u}^T\mathbf{v}_k + \Phi^{-1}(1-\epsilon)\sqrt{{a_j^u}^TK_k\Sigma_kK_k^Ta_j^u} \leq b_j^u$ for all $k \in [2T-2]$, $j \in [M_u]$.

\cref{alg: prism_mean_steering} inflates $\mathcal{S}_i$, $\mathcal{S}_j$, and $\mathcal{S}_{i \cap j}$ to find a sequence of safe sets in $\mathcal{X}_p\setminus\mathcal{P}$ that contains the lifted belief space trajectory. For all $k \in [2T-1]$, $\mu_k \in \mathcal{S}_k \subset (\mathcal{X}_p^- \setminus \mathcal{P}^+)$, where $\mathcal{S}_k = \mathcal{S}_i$ for $k \in [T-1]$, $\mathcal{S}_T = \mathcal{S}_i \cap \mathcal{S}_j$, and $\mathcal{S}_k = \mathcal{S}_j$ for $k = T+[2T-1]$. For all $\mathcal{P}_i \in \mathcal{P}$, $\mathcal{P}_i := \{\mathbf{x} \in \mathbb{R}^n: \cap_{j=1}^{M_{p_i}} {a_\ell^{p_j}}^TP\mathbf{x} \leq b_\ell^{p_j} \}$, and $\mathcal{P}_i^+ := \{\mathbf{x} \in \mathbb{R}^n: \cap_{j=1}^{M_{p_i}} {a_\ell^{p_j}}^TP\mathbf{x} \leq b_\ell^{p_j} - \Phi^{-1}(1-\epsilon)\sqrt{r^\star} \}$, with $r^\star = \max_k\lambda_{\max}(P\Sigma_k^{(c_{1, \delta})}P^T)$. Then, $\mathcal{S}_k \oplus \mathbb{B}_p(0, \Phi^{-1}(1-\epsilon)\sqrt{r^\star})$ can be separated from each obstacle in $\mathcal{P}$ by a hyperplane. We define $\mathcal{S}_k^+ = \text{INFLATE}(\mathcal{S}_k)$ as the intersection of $\mathcal{X}_p$ and these separating hyperplanes. Then, $\mathcal{S}_k^+ \in \mathcal{X}_p \setminus \mathcal{P}$, and $\mathcal{S}_k^+ = \{\mathbf{x} \in \mathbb{R}^n: \cap_{i=1}^{M_{s_k}^+}{a_i^{x_s^+}}^TP\mathbf{x} \leq b_i^{x_s^+} \}$. Further, $\mathcal{X}_p^- = \{\mathbf{x}\in\mathbb{R}^n: \cap_{i=1}^{M_{x_p}} {a_i^{x_p}}^TP\mathbf{x} \leq b_i^{x_p} - \Phi^{-1}(1-\epsilon)\delta_{b_i}^{x_p} \}$. Because $\delta_{b_i}^{x_p} = \max_k \sqrt{{a_i^{x_p}}^TP\Sigma_k^{(c_{1, \delta})}P^Ta_i^{x_p}} \geq \sqrt{{a_i^{x_p}}^TP\Sigma_kP^Ta_i^{x_p}}$ for all $k$ and $\sqrt{r^\star} \geq \sqrt{a^TP\Sigma_kP^Ta}$ for any $k$ and $a \in \mathbb{R}^p$ such that $||a||_2 = 1$, ${{a_i}^{x_s^+}}^T\mu_k + \Phi^{-1}(1-\epsilon)\sqrt{{{a_i}^{x_s^+}}^TP\Sigma_kP^T{{a_i}^{x_s^+}}} \leq b_i^{x_s^+}$ for all $k$ and all $i$. It follows that that \cref{eq: pos_chance_constraint} is satisfied for all $k \in [2T-1]$, $\ell \in [M_{s_k^+}]$. Finally, $\mu_1 = \mu_\mathcal{I}$ and $\mu_{2T-1} = \mu_\mathcal{J}$, $\Sigma_1 = c_{1, \delta}^{(N)}I$ and $\Sigma_{2T-1} \preceq c_{1, \delta}^{(N)}I$. Then, the boundary constraints are satisfied and so the edge $(\mu_{1:2T-1}, \mathbf{v}_{1:2T-2}, \Sigma_{1:2T-1}, K_{1:2T-2})$ found by \cref{alg: prism_mean_steering} must satisfy \cref{def: soundness}.
\end{proof}
\textbf{Proof of \cref{lemma: recursive_feasibility}}
\begin{proof}
During preprocessing, PRISM's local optimization module finds a set of reference values so that \cref{eq: convex_pos_chance_constraint} -- \cref{eq: convex_ctrl_chance_constraint} hold for each edge. This is feasible if the input $\mathcal{T}$ is a sequence of sound edges (see \cref{sec: preliminaries}), and ensures that \cref{prob: steer_edge} is feasible for the constraints and time allocation for each edge under the chosen set of reference values. Solutions to \cref{prob: steer_edge} must satisfy \cref{def: soundness}, because satisfying \cref{eq: convex_ctrl_chance_constraint} implies satisfaction of \cref{eq: control_chance_constraint}, satisfying \cref{eq: convex_nonpos_state_chance_constraint} implies satisfaction of \cref{eq: nonpos_state_chance_constraint}, satisfying \cref{eq: convex_pos_chance_constraint} implies satisfaction of \cref{eq: pos_chance_constraint} and satisfying \cref{eq: convex_covariance_dynamics} and \cref{eq: schur_complement_constraint} implies satisfaction of \cref{eq: covariance_dynamics}.

PRISM's local optimization module alternates between shortening edges, finding and applying shortcuts, and locally optimizing the trajectory. Suppose $\mathcal{T}$ is a sequence of edges $e_1, \ldots, e_T$ such that $e_i$ is sound and \cref{eq: convex_pos_chance_constraint} -- \cref{eq: convex_ctrl_chance_constraint} hold (after finding appropriate reference values) for all $i$. The edge shortening step consists of solving \cref{prob: steer_edge} for $N \in [N_c]$, where $N_c$ is the current edge length, for each edge and then selecting the lowest-cost solution. \cref{prob: steer_edge} is always feasible when $N = N_c$ and returns sound edges, so this step preserves soundness and completeness. The next step finds shortcuts by solving \cref{prob: steer_edge} between different nodes and then finds the minimum-cost path over the graph of current edges and candidate shortcut edges. Any shortcuts found will satisfy \cref{def: soundness}. If no shortcuts are found, this step will return the input sequence of edges; soundness and completeness are preserved. The next step concatenates all edges (including constraints, reference values, and time allocations) and solves \cref{prob: steer_edge} over the full trajectory. Each edge is a feasible solution to \cref{prob: steer_edge}, so solving \cref{prob: steer_edge} over the concatenated sequence of edges is guaranteed to be feasible and preserve soundness. Converting the resulting trajectory back into a sequence of edges also preserves soundness. Repeating these steps will always result in returning a sound trajectory that satisfies \cref{def: soundness}, and every intermediate trajectory will also satisfy \cref{def: soundness}.
\end{proof}


\small
\renewcommand{\baselinestretch}{0.9}
\bibliographystyle{plain}
\bibliography{references}
\begin{IEEEbiography}[{\includegraphics[width=1in,height=1.25in,clip,keepaspectratio]{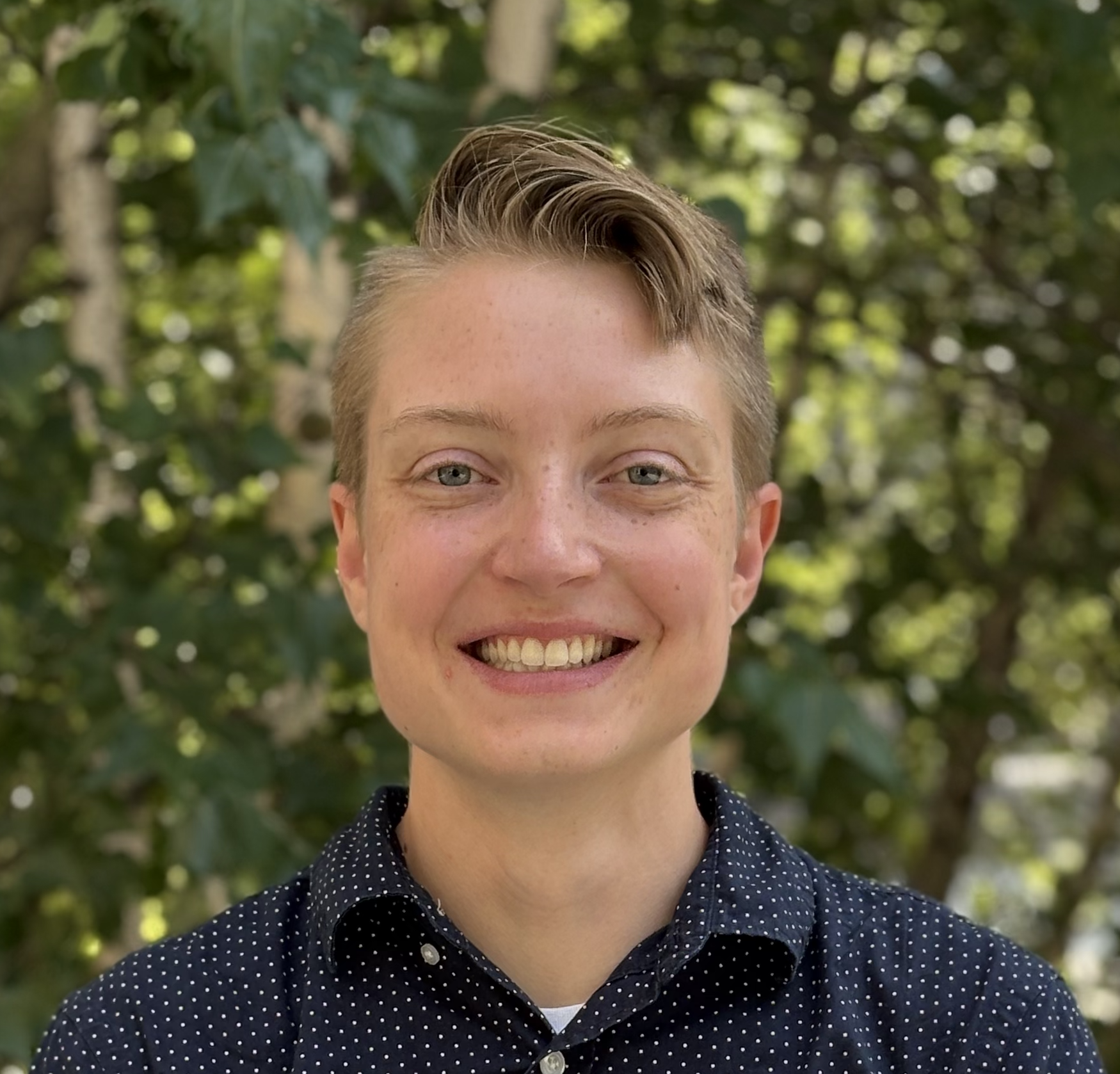}}]%
{Alex Rose} received the S.B. degree and S.M. degrees in aeronautics and astronautics from the Massachusetts Institute of Technology (MIT), Cambridge, MA, USA, in 2021 and 2023, respectively, and is currently working towards the Ph.D. degree in aeronautics and astronautics at MIT.

Alex's research interests include control theory and optimization with applications to space systems and robotics.
\end{IEEEbiography}
\begin{IEEEbiography}[{\includegraphics[width=1in,height=1.25in,clip,keepaspectratio]{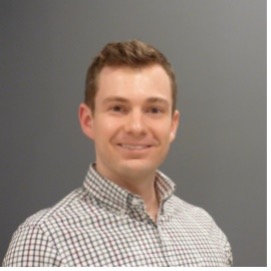}}]%
{Christopher Jewison} is a GN\&C Engineer at Draper. He received his bachelor’s degree (2012) from Cornell University and his master’s degree (2014) and PhD (2017) from MIT.
\end{IEEEbiography}
\begin{IEEEbiography}[{\includegraphics[width=1in,height=1.25in,clip,keepaspectratio]{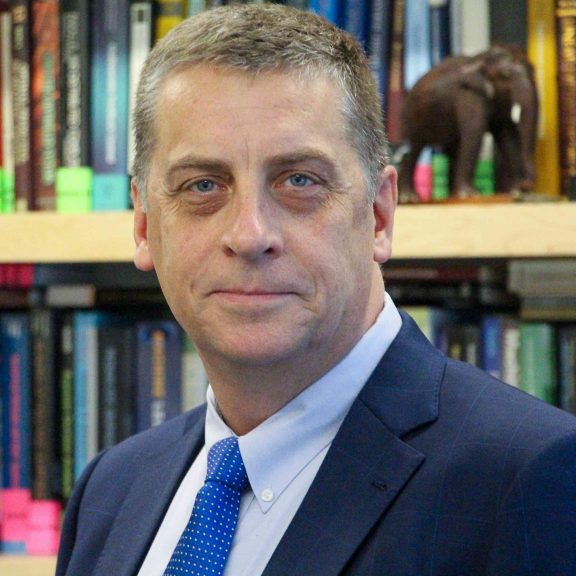}}]%
{Jonathan P. How}
(Fellow, IEEE) received the B.A.Sc. degree in Engineering Science (Aerospace) from the University of Toronto, Toronto, ON, Canada, in 1987, and the S.M. and Ph.D. degrees in aeronautics and astronautics from the Massachusetts Institute of Technology (MIT), Cambridge, MA, USA, in 1990 and 1993, respectively.
In 2000, he joined MIT, where he is currently the Ford Professor of Engineering. Prior to this, he was an Assistant
Professor with Stanford University, Stanford, CA, USA.
Dr. How was the Recipient of the 
IEEE Transactions on Robotics King-Sun Fu Memorial Best Paper Award for 2022 and 2024, the IEEE Control Systems Society Distinguished Member Award in 2020, and the AIAA Intelligent Systems Award in 2020.  He was the Editor-in-Chief for IEEE Control Systems Magazine from 2015 to 2019. He is a Fellow of the AIAA. He was elected to the National Academy of Engineering in 2021.
\end{IEEEbiography}
\end{document}